\documentclass[11pt]{article}
\usepackage{bm}
\usepackage[font=small]{caption}
\RequirePackage{amsthm,amsmath,amsfonts,amssymb}

\let\hat\widehat
\let\tilde\widetilde

\newtheoremstyle{mytheoremstyle}
    {\topsep}
    {\topsep}
    {\normalfont}
    {}
    {\bfseries}
    {.}
    {.5em}
    {}

\theoremstyle{mytheoremstyle}

\ifx\BlackBox\undefined
\newcommand{\BlackBox}{\rule{1.5ex}{1.5ex}}
\fi

\ifx\QED\undefined
\def\QED{~\rule[-1pt]{5pt}{5pt}\par\medskip}
\fi

\ifx\proof\undefined
\newenvironment{proof}{\par\noindent{\bf Proof\ }}{\hfill\BlackBox\\[2mm]}
\fi

\ifx\theorem\undefined
\newtheorem{theorem}{Theorem}
\fi
\ifx\example\undefined
\newtheorem{example}[theorem]{Example}
\fi
\ifx\property\undefined
\newtheorem{property}{Property}
\fi
\ifx\lemma\undefined
\newtheorem{lemma}[theorem]{Lemma}
\fi
\ifx\proposition\undefined
\newtheorem{proposition}[theorem]{Proposition}
\fi
\ifx\remark\undefined
\newtheorem{remark}[theorem]{Remark}
\fi
\ifx\corollary\undefined
\newtheorem{corollary}[theorem]{Corollary}
\fi
\ifx\definition\undefined
\newtheorem{definition}{Definition}[section]
\fi
\ifx\conjecture\undefined
\newtheorem{conjecture}[theorem]{Conjecture}
\fi
\ifx\fact\undefined
\newtheorem{fact}[theorem]{Fact}
\fi
\ifx\claim\undefined
\newtheorem{claim}[theorem]{Claim}
\fi
\ifx\assumption\undefined
\newtheorem{assumption}{Assumption}[section]
\fi
\ifx\cond\undefined
\newtheorem{cond}[theorem]{Condition}
\fi

\newtheoremstyle{anothertheoremstyle}
    {\topsep}
    {\topsep}
    {\normalfont}
    {}
    {\bfseries}
    {.}
    {.5em}
    {\thmname{#1}\thmnumber{ #2}\thmnote{ {\normalfont #3}}}
\ifx\condition\undefined
\theoremstyle{anothertheoremstyle}
\newtheorem{condition}[theorem]{Condition}
\fi
\numberwithin{equation}{section}
\numberwithin{theorem}{section}

\makeatletter
\newcommand{\raisemath}[1]{\mathpalette{\raisem@th{#1}}}
\newcommand{\raisem@th}[3]{\raisebox{#1}{$#2#3$}}
\makeatother

\let\oldabstract\abstract
\let\oldendabstract\endabstract
\makeatletter
\renewenvironment{abstract}
{%
               {\list{}{\addtolength{\leftmargin}{1.7em}
                       \listparindent 0.75em%
                        \itemindent    \listparindent%
                        \rightmargin   \leftmargin%
                        \parsep        \z@ \@plus\p@}%
                \item\relax}%
               {\endlist}%
\oldabstract}
{\oldendabstract}
\makeatother

\usepackage[a4paper,left=2.5cm,right=2.5cm,top=2.5cm,bottom=2.5cm]{geometry}

\usepackage{bm}
\usepackage{natbib}
\usepackage[ruled]{algorithm2e}
\usepackage[dvipsnames]{xcolor}
\usepackage[colorlinks = true,
            linkcolor = NavyBlue,
            urlcolor  = NavyBlue,
            citecolor = NavyBlue,
            anchorcolor = NavyBlue]{hyperref}

 \usepackage{booktabs}

\usepackage{chngcntr}
\usepackage{colonequals}
\usepackage[graphicx]{realboxes}
\usepackage{enumerate}
\usepackage{comment}
\usepackage{subfigure}
\RequirePackage{subfigure}
\RequirePackage{algorithmic}

\newcommand{\RR}{\mathbb{R}}
\newcommand{\ee}{\end{aligned} \end{equation}}
\newcommand{\eq}{\end{quote}}
\newcommand{\diag}{\mathrm{diag}}

\newcommand{\ep}{\end{parts}}

\newcommand{\bqp}{\begin{quote}\begin{parts}}

\newcommand{\epq}{\end{parts}\end{quote}}

\newcommand{\E}{\mathbb{E}}

\DeclareMathOperator*{\argmin}{argmin}
\newcommand{\Rom}[1]{\text{\uppercase\expandafter{\romannumeral #1\relax}}}
\newcommand{\bee}{\begin{equation}\begin{aligned}}

\newcommand{\emm}{\end{bmatrix}}

\newcommand{\argmax}{\operatornamewithlimits{argmax}}
\numberwithin{equation}{section}
\newcommand{\vertiii}[1]{{\vert\kern-0.25ex\vert\kern-0.25ex\vert #1
    \vert\kern-0.25ex\vert\kern-0.25ex\vert}}

\newcommand{\ma}{\bm{A}}
\newcommand{\mb}{\bm{B}}

\newcommand{\mi}{\bm{I}}

\newcommand{\mv}{\bm{V}}
\newcommand{\mx}{\bm{X}}
\newcommand{\mn}{\bm{N}}

\newcommand{\mw}{\bm{W}}
\newcommand{\mg}{\bm{G}}
\newcommand{\mm}{\bm{M}}

\newcommand{\mpp}{\bm{P}}
\newcommand{\mo}{\bm{O}}
\newcommand{\my}{\bm{Y}}
\newcommand{\muu}{\bm{U}}

\newcommand{\mSigma}{\bm{\Sigma}}

\newcommand{\mLambda}{\bm{\Lambda}}

\newtheorem{Theorem}{Theorem}

\newtheorem{Remark}{Remark}

\newtheorem{Assumption}{Assumption}

 \title{
Global Synchronization for Multi-Source Data Integration under Blockwise Missing Patterns}

  \usepackage{authblk}

  \author[1]
  {Runbing Zheng\thanks{Email: \texttt{runbing@stanford.edu}.}}
  \author[2]
  {Debsurya De\thanks{Email: \texttt{dde4@jhu.edu}.}}
  \author[2]
  {Dmitriy Kunisky\thanks{Email: \texttt{kunisky@jhu.edu}.}}

  \affil[1]{Department of Neurology and Neurological Sciences, Stanford University}
  \affil[2]
  {Department of Applied Mathematics \& Statistics, Johns Hopkins University}

\date{\today}

\begin{document}

\maketitle
\thispagestyle{empty}

\begin{abstract}

Multi-source data integration problems over datasets from different sources covering different but possibly overlapping sets of entities have become increasingly important in many applications, including genomics, single-cell analysis, and healthcare research. We study a model of multi-source data integration where we observe overlapping submatrices of a low-rank matrix, each corrupted by additive noise. In such problems, one often first learns a low-dimensional representation of the entities within each source and then integrates these representations across sources. As the representations are only identifiable up to some transformation, aligning them using the sources' overlapping entities becomes a key challenge. Existing methods align these estimates in a sequential or tree-structured manner, and are therefore sensitive to the chosen order and exploit only part of the overlapping information. We propose Global Synchronized Multiple Matrix Integration (GSMMI), which formulates this alignment problem as a global synchronization problem and jointly aligns all sources using all pairwise overlaps at once. We develop an efficient, scalable iterative algorithm for GSMMI, and show both theoretically and empirically that it improves alignment accuracy, with clear gains over prior approaches even under moderate overlap. We additionally establish information-theoretic and computational lower bounds for detection in data integration, and find GSMMI to be empirically near-optimal. GSMMI applies broadly across symmetric positive semidefinite, symmetric indefinite, and asymmetric or rectangular matrices, as well as settings where sources overlap only in rows or only in columns.

\end{abstract}

\vspace{1em}


\clearpage

\tableofcontents
\thispagestyle{empty}

\clearpage

\begin{sloppypar}

\section{Introduction}
\pagenumbering{arabic}

The rapid growth of large-scale data collection and sharing has made \emph{multi-source data integration} an increasingly important problem in modern statistics, with applications ranging across healthcare research \citep{zhou2021multi}, genomics \citep{maneck2011genomic,tseng2015integrating,zang2016high,cai2016structured}, single-cell analysis \citep{stuart2019comprehensive,argelaguet2021computational,ma2023your}, and chemometrics \citep{mishra2021recent}.
We focus on the specific case of \emph{matrix integration} where one observes several submatrices of an underlying low-rank matrix signal and seeks to integrate these into one estimate of the entire signal.
Multi-source matrix integration thus gives rise to matrix-valued estimation problems with missing blocks of data, where each source observes a submatrix of the entire underlying data matrix, with different sources covering different but potentially overlapping subsets of its rows and columns.
We give three motivating examples here to illustrate how differences in coverage across sources give rise to different patterns of blockwise missing data.

First, consider pointwise mutual information (PMI) matrices derived from electronic health record (EHR) data.
PMI measures associations between pairs of clinical concepts based on their co-occurrence \citep{ahuja2020surelda,zhou2022multiview}.
Because healthcare systems have limited interoperability \citep{rajkomar2018scalable}, different EHR datasets often cover partially overlapping sets of concepts.
Each PMI matrix therefore provides a principal submatrix of the full concept-by-concept matrix, with shared concepts giving rise to overlapping blocks, as shown in Figure~\ref{fig:motivating_examples}(a).
Second, in single-cell data, each dataset records measurements of genomic features across cells. Experiments spanning different batches, tissues, or technologies may cover different, potentially overlapping sets of cells and features \citep{ma2020integrative,lahnemann2020eleven}. The resulting matrices provide rectangular asymmetric blocks that may overlap along both dimensions, as illustrated in Figure~\ref{fig:motivating_examples}(b).
Third, in multimodal health data, different modalities measure distinct sets of features, but not every subject has data available for every modality, for example due to the cost of measurement. The resulting matrices have disjoint feature sets and partially overlapping subject sets, producing the pattern in Figure~\ref{fig:motivating_examples}(c).
\begin{figure}[htbp]
    \centering
    \begin{minipage}[b]{0.325\textwidth}
        \centering
        \includegraphics[height=3.3cm]{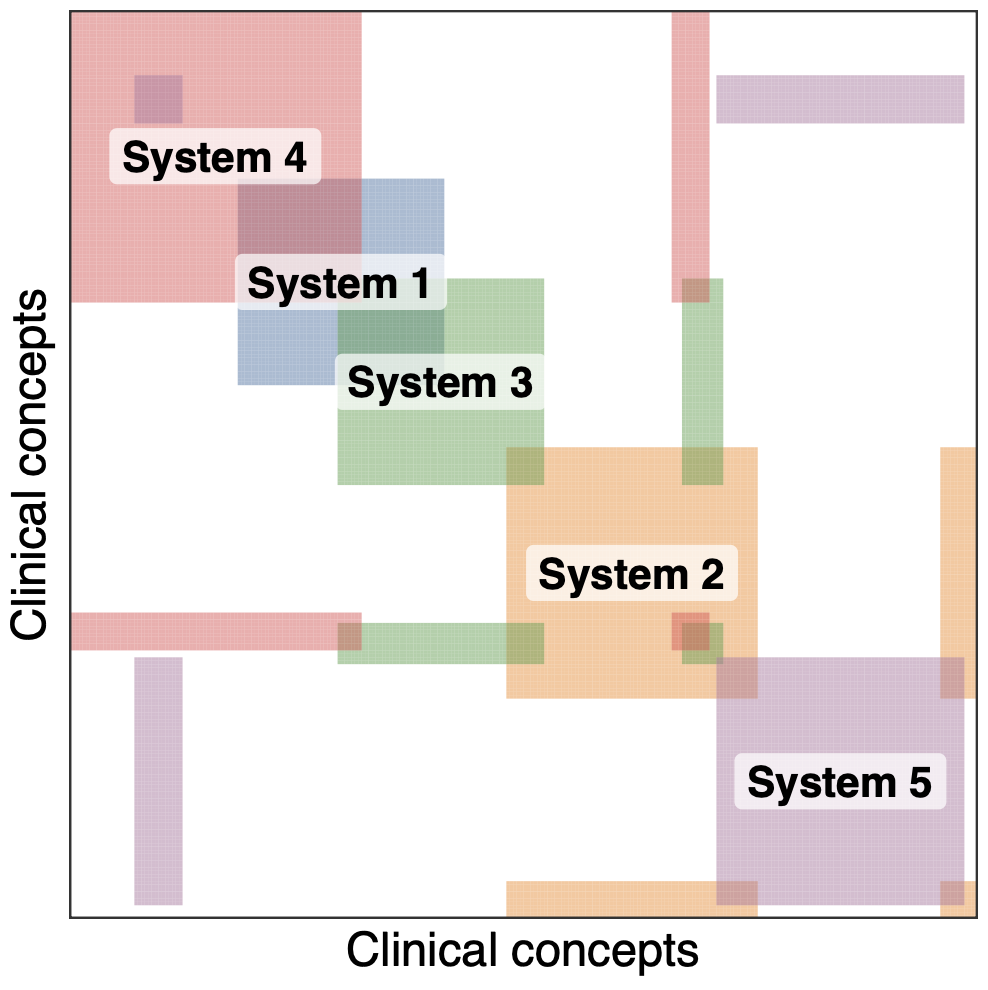}\\
    \footnotesize (a) PMI matrices
    \end{minipage}
    \begin{minipage}[b]{0.325\textwidth}
        \centering
        \includegraphics[height=3.3cm]{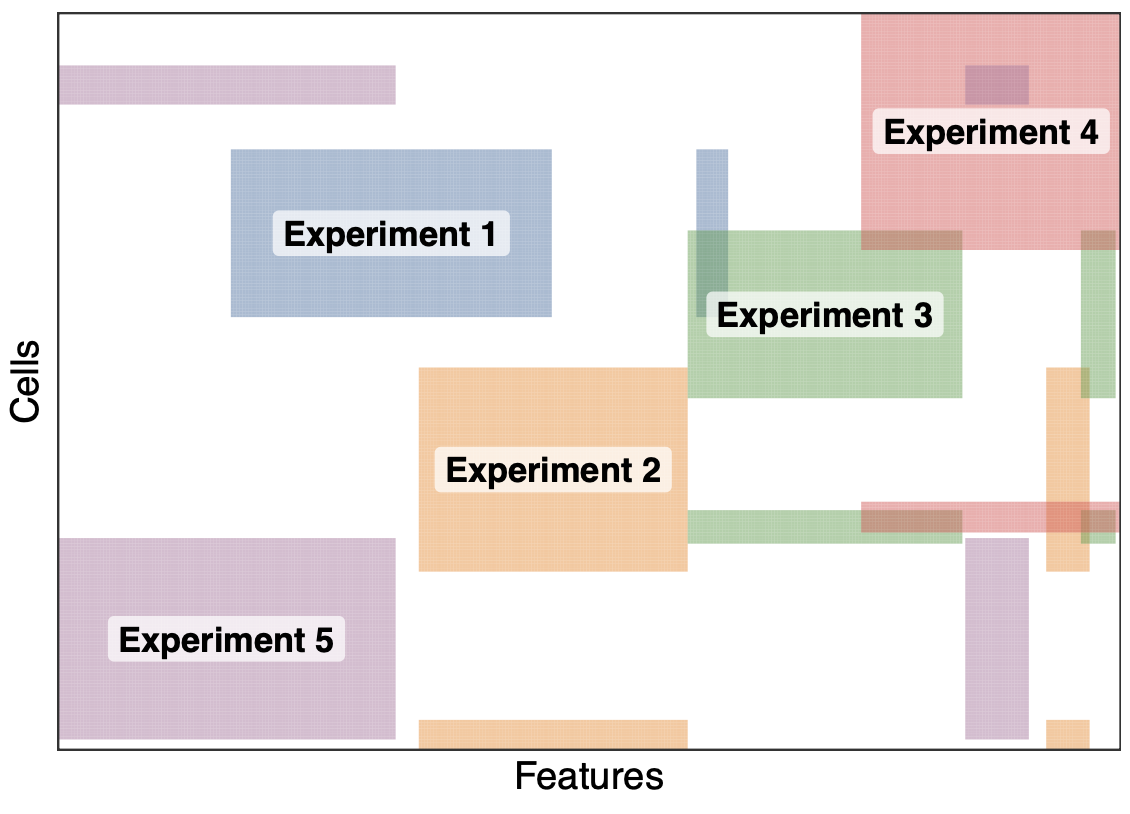}\\
     \footnotesize (b) Single-cell matrices
    \end{minipage}
    \begin{minipage}[b]{0.325\textwidth}
        \centering
        \includegraphics[height=3.3cm]{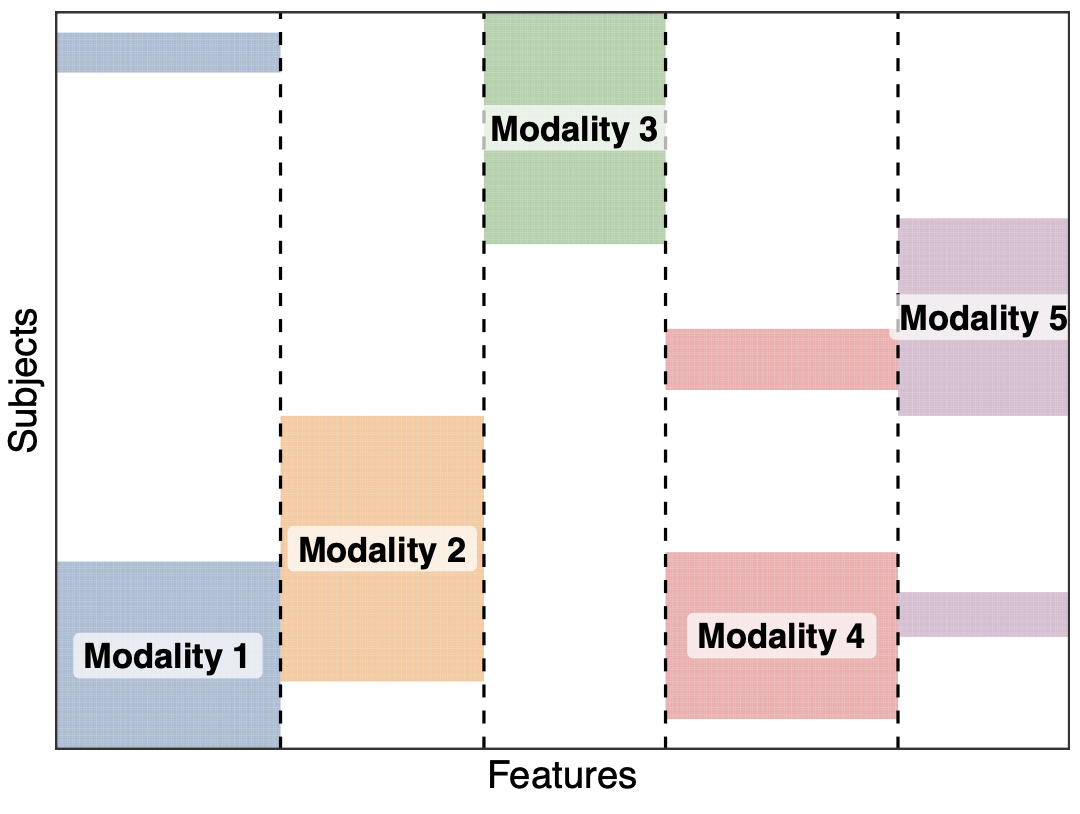}\\
      \footnotesize (c) Multimodal health data
    \end{minipage}
    \caption{\footnotesize Blockwise observation patterns in three multi-source integration settings: (a) PMI matrices from different EHR systems, (b) single-cell matrices from different experiments, and (c) multimodal health data. Different colors indicate blocks observed by different sources, and blended colors indicate overlaps between sources.}
    \label{fig:motivating_examples}
\end{figure}

\paragraph{Prior work}
The blockwise observation patterns in Figure~\ref{fig:motivating_examples} differ from the uniform or independent sampling of entries commonly assumed in matrix completion.
Despite extensive work on matrix completion (see, e.g., \citet{candes2012exact,candes2011tight,koltchinskii2011nuclear,chen2019inference,fornasier2011low,mohan2012iterative,vandereycken2013low,hu2012fast,chen2020nonconvex,cai2016matrix,foygel2011learning}), most of these methods are not well suited to the multi-source integration problem considered here.

Much of the work on blockwise missing data addresses specific tasks under specialized assumptions.
For regression with blockwise missing design matrices, \cite{xue2021integrating} propose multiple blockwise imputation for model selection in a linear model.
\cite{xue2021statistical} combine unbiased estimating equations with blockwise imputation to construct debiased estimators and valid confidence intervals for high-dimensional linear regression.
\cite{he2023variable} use sparse precision matrix estimation to build a single regression-based imputation algorithm for variable selection in a generalized linear model.
\cite{zhang2020imputed} use factor-model imputation to obtain an imputed factor regression for dimension reduction and prediction.
Beyond regression, \cite{sui2026multi} propose a multi-task learning framework that uses shared and task-specific representations to impute missing blocks and predict a target for each task.
\cite{wang2026multimodal} study multimodal domain adaptation under label shift, using labeled sources with blockwise missing modalities to predict outcomes in an unlabeled target population.
\cite{zhu2020generalized} propose generalized integrative principal component analysis for sources with common samples and distinct variables, allowing a different exponential-family distribution for each source.

Other methods aim to recover the full low-rank matrix, or equivalently the latent positions of all entities in a low-dimensional space, without targeting a specific downstream task or imposing assumptions tied to one \citep{zhou2021multi,bishop2014deterministic,zheng2026cluster,liu2026representation,zheng2026chain}.
This is the problem we study.
These methods typically first estimate latent positions separately within each source using a low-rank method, such as spectral estimation.
The estimates are identifiable only up to a transformation, so the methods then align them using entities shared by overlapping sources.

The main difference between these methods is the alignment step.
\cite{zhou2021multi} align each overlapping pair on its shared entities and recover the entries associated with that pair.
This does not produce a globally consistent representation, so only a limited portion of the matrix may be recovered correctly.
\cite{bishop2014deterministic} integrate symmetric positive semidefinite submatrices sequentially, aligning each new block to the entities already covered.
This can produce globally aligned latent positions for all covered entities, but its performance is sensitive to the integration order, and the choice of an effective order remains unresolved in that work.
Each alignment is fixed once made, so later overlaps cannot correct it and some overlap information goes unused.
\cite{zheng2026cluster} likewise align rectangular submatrices sequentially and provide a greedy algorithm for choosing their order.

\cite{liu2026representation} align rectangular sources using anchor features shared by all sources.
Their alignment uses only this anchor block, which may contain a small fraction of the available overlap information.
If no globally shared feature set exists, they first estimate overlapping groups using the features shared within each group, then align successive groups through their shared members in a prescribed order.
They provide no algorithm for choosing this order, and the procedure has limitations similar to those of the other sequential methods.

The most direct precursor to our work is chain-linked multiple matrix integration (CMMI), proposed by \cite{zheng2026chain}.
CMMI represents sources as vertices of a graph and overlapping pairs as edges, then selects a spanning tree to align all submatrices in a common coordinate system, giving a globally consistent alignment.
It also provides a method for choosing a tree that can reduce alignment error.
Going beyond sequential integration allows better estimation, but CMMI still uses only the overlaps along the chosen tree.
These limitations motivate an alignment scheme that uses the full structure of overlapping observations.

\paragraph{Contributions and organization}
In this paper, building on CMMI, we propose Global Synchronized Multiple Matrix Integration (GSMMI).
In this method, we formulate the alignment subproblem of matrix integration as a global non-convex synchronization problem, jointly seeking transformations for all sources that simultaneously match the latent positions of every pair of overlapping submatrices over all of their shared entities.
Unlike CMMI and the other methods above, which propagate alignments along a sequence or a spanning tree and thus rely on only part of the overlapping information, GSMMI aggregates the alignment information from all overlapping pairs at once, thereby making full use of the overlapping structure.
Computational efficiency is another important consideration for multi-source data integration, since applications in areas such as genomics, single-cell analysis, and healthcare research often involve large-scale data.
A notable advantage of CMMI is that its alignment mainly involves a sequence of singular value decompositions and is therefore computationally very fast.
To preserve this efficiency, we develop an iterative algorithm for the global synchronization problem, in which each update reduces to a Procrustes or least-squares subproblem that admits an efficient closed-form solution.
We further use the CMMI solution as a high-quality initialization, which experiments demonstrate allows GSMMI to converge rapidly to accurate estimates, so that its running time is close to, or at most a few times, that of CMMI.

The structure of the rest of the paper is as follows. In Section~\ref{sec:method}, we introduce the model for multiple observed submatrices of a whole symmetric positive semidefinite matrix, revisit the CMMI algorithm, and then present the proposed GSMMI algorithm. In Section~\ref{sec:theory}, we provide a theoretical analysis showing that GSMMI can achieve improved alignment accuracy over CMMI, particularly when the number of sources is large and the overlapping structure is rich. We also give computational and information-theoretic lower bounds for detection in matrix integration tasks, which we compare with the performance of GSMMI in the subsequent experiments.
 Numerical simulations and real-data experiments are presented in Sections~\ref{sec:simu} and \ref{sec:real}, illustrating that GSMMI outperforms CMMI in estimation accuracy while incurring only a modest computational overhead. Finally, although the main idea and analysis of GSMMI are presented for the case of symmetric positive semidefinite matrices for simplicity, GSMMI can be extended to symmetric indefinite matrices as well as asymmetric or rectangular matrices, as we show in Section~\ref{sec:extension}, so that it is applicable to the full range of multiple matrix integration problems, including all the settings illustrated in Figure~\ref{fig:motivating_examples}.

\subsection{Notation}
We summarize some notations used in this paper.
For any positive integer $n$, we denote by $[n]$ the set $\{1,2,\dots, n\}$.

\paragraph{Asymptotics}
For two non-negative sequences
$\{a_n\}_{n \geq 1}$ and $\{b_n\}_{n \geq 1}$, we write $a_n \lesssim
b_n$ (resp. $a_n \gtrsim b_n$) if there exists some constant $C>0$
such that $a_n \leq C b_n$ (resp. $a_n \geq C b_n$) for all  $n \geq 1$, and we write $a_n \asymp b_n$ if $a_n\lesssim b_n$ and $a_n\gtrsim b_n$.
If $a_n/b_n$ stays bounded away from $+\infty$, we write $a_n=O(b_n)$ and $b_n=\Omega(a_n)$, and we use the notation $a_n=\Theta(b_n)$ to indicate that $a_n=O(b_n)$ and $a_n=\Omega(b_n)$.
If $a_n/b_n\to 0$, we write $a_n=o(b_n)$ and $b_n=\omega(a_n)$.
We say a sequence of events $\mathcal{A}_n$ holds with high probability if for any $c > 0$ there exists a finite constant $n_0$ depending only on $c$ such that $\mathbb{P}(\mathcal{A}_n)\geq 1-n^{-c}$ for all $n \geq n_0$.
We write $a_n = O_p(b_n)$ (resp. $a_n = o_p(b_n)$) to denote that $a_n = O(b_n)$ (resp. $a_n = o(b_n)$) holds with high probability.

\paragraph{Linear algebra}
We denote by $\mathcal{O}_d$ the set of $d \times d$ orthogonal
matrices, and let $\mathcal{O}_{n\times d} \colonequals \{\muu\in\mathbb{R}^{n\times d}\mid \muu^\top\muu=\mi_d\}$.
For any matrix $\mm\in \mathbb{R}^{A\times B}$ and index sets $\mathcal{A}\subseteq [A]$, $\mathcal{B}\subseteq [B]$, we denote by $\mm_{\mathcal{A},\mathcal{B}}\in \mathbb{R}^{|\mathcal{A}|\times |\mathcal{B}|}$ the submatrix of $\mm$ formed from rows $\mathcal{A}$ and columns $\mathcal{B}$, and we denote by $\mm_{\mathcal{A}}\in\mathbb{R}^{|\mathcal{A}|\times B}$ the submatrix of $\mm$ consisting of the rows indexed by $\mathcal{A}$. The Hadamard product between conformal matrices $\bm{M}$ and $\bm{N}$ is denoted by $\bm{M} \circ \bm{N}$.
Given a matrix $\bm{M}$, we denote
its spectral, Frobenius, and infinity norms by $\|\bm{M}\|$,
$\|\bm{M}\|_{F}$, and $\|\bm{M}\|_{\infty}$.
We also denote %
the $2 \to \infty$ norm of $\bm{M}$ by $\|\bm{M}\|_{2 \to \infty} = \max_{\|\bm{x}\| = 1} \|\bm{M} \bm{x}\|_{\infty} = \max_{i} \|\bm{m}_i\|,$ where $\bm{m}_i$ denotes the $i$th row of $\bm{M}$, i.e., $\|\bm{M}\|_{2 \to \infty}$ is the maximum of the $\ell_2$ norms of the rows of $\bm{M}$.

\section{Methodology and proposed algorithm}
\label{sec:method}

We first formulate a general model of data integration problems.
Suppose we are interested in an unobserved low-rank population matrix $\mpp \in \mathbb{R}^{N \times N}$ associated with $N$ entities.
For now, we further suppose $\mpp$ is positive semidefinite with rank $d \ll N$; extensions to the case of symmetric but indefinite $\mpp$, as well as asymmetric or rectangular $\mpp$, including the corresponding algorithms, are presented in Section~\ref{sec:extension}.
Denote the eigendecomposition of $\mpp$ by $\mpp = \muu\mLambda\muu^\top$, where $\mLambda\in\mathbb{R}^{d\times d}$ is a diagonal matrix whose diagonal entries are the non-zero eigenvalues of $\mpp$ in descending order, and the orthonormal columns of $\muu\in\mathcal{O}_{N\times d}$ are the corresponding eigenvectors. We define the latent positions of the entities as $\mx=\muu\mLambda^{1/2}\in\mathbb{R}^{N\times d}$, where the $s$th row $\bm{x}_s$ gives the latent position of entity $s$ in $\mathbb{R}^{d}$. Then we have $\mpp = \mx \mx^{\top}$, and each entry is the inner product of the corresponding latent positions, i.e., $\mpp_{s,t}=\bm{x}_s^\top\bm{x}_t$ for all $s,t \in [N]$.

We assume that $\mpp$ is partially observed. More specifically, suppose that we have $m$ sources, and from each source, we obtain a (corrupted) observation for a subset of the entities. For any $i\in[m]$, we denote the index set of the entities contained in the $i$th source by $\mathcal{U}_i\subseteq [N]$. We denote $n_i  \colonequals  |\mathcal{U}_i|$ and the population matrix for the $i$th source by $\mpp^{(i)}\in\mathbb{R}^{n_i\times n_i}$. We then have
$$\mpp^{(i)}=\mpp_{\mathcal{U}_i,\mathcal{U}_i}=\muu_{\mathcal{U}_i}\mLambda\muu_{\mathcal{U}_i}^\top =\mx_{\mathcal{U}_i}\mx_{\mathcal{U}_i}^\top,$$
where $\mpp_{\mathcal{U}_i,\mathcal{U}_i}$ is the submatrix of $\mpp$ formed from rows and columns in $\mathcal{U}_i$, $\muu_{\mathcal{U}_i}\in\mathbb{R}^{n_i\times d}$ contains the rows of $\muu$ in $\mathcal{U}_i$, and $\mx_{\mathcal{U}_i}\in\mathbb{R}^{n_i\times d}$ contains the latent positions of $\mathcal{U}_i$.
Suppose that from each source $i$, we observe a noisy submatrix
\begin{equation}
\ma^{(i)} = \mpp^{(i)}+\mn^{(i)}, \label{eq:full-model}
\end{equation}
where $\mn^{(i)} \in \mathbb{R}^{n_i \times n_i}$ denotes an unobserved symmetric zero-mean random noise matrix. %

Since each source observes only the entities in its own index set $\mathcal{U}_i$, the collection $\{\ma^{(i)}\}_{i=1}^m$ covers a union of principal submatrices of $\mpp$ that may overlap with one another, forming a block-structured observation pattern such as the one illustrated in Figure~\ref{fig:motivating_examples}(a). %
The goal is to recover the entire population matrix $\mpp$, or equivalently the latent positions in $\mx$, based on the observed $\{\ma^{(i)}\}_{i=1}^m$.
    \label{fig:block_structure}

Our proposed method, Global Synchronized Multiple Matrix Integration (GSMMI), improves upon the Chain-linked Multiple Matrix Integration (CMMI) algorithm proposed by \citet{zheng2026chain}. We first introduce the main idea of CMMI before presenting our improvements.

\subsection{Revisiting Chain-Linked Multiple Matrix Integration}\label{sec:rev_CMMI}

CMMI applies a natural idea: first estimate the latent position matrix $\mx_{\mathcal{U}_i}$ from each observed submatrix $\ma^{(i)}$ by computing $\hat{\mx}^{(i)}\in\mathbb{R}^{n_i\times d}$ such that $\ma^{(i)}\approx \hat{\mx}^{(i)}\hat{\mx}^{(i)\top}$ (for example, using the scaled leading eigenvectors $\hat{\mx}^{(i)}=\hat{\muu}^{(i)}(\hat{\mLambda}^{(i)})^{1/2}$, where $\hat{\mLambda}^{(i)}$ and $\hat{\muu}^{(i)}$ contain the leading eigenvalues and eigenvectors of $\ma^{(i)}$, respectively), and then integrate $\{\hat{\mx}^{(i)}\}$ to recover $\mpp$.
Note that the latent position matrix $\hat{\mx}^{(i)}$ obtained from $\ma^{(i)}$ is unique only up to an orthogonal transformation. Therefore, to integrate any $\hat{\mx}^{(i)}$ and $\hat{\mx}^{(j)}$ for jointly recovering $\mpp$, we need an orthogonal alignment matrix to bridge them.

When two submatrices $\ma^{(i)}$ and $\ma^{(j)}$ involve sufficiently many overlapping entities ($|\mathcal{U}_i\cap\mathcal{U}_j|\geq d$), such an alignment matrix $\hat\mw^{(i,j)}\in\mathcal{O}_d$ can be obtained by matching their latent positions corresponding to the shared entities, for instance by solving the orthogonal Procrustes problem
\begin{equation}\label{eq:procrustes_adjacent}
	\hat\mw^{(i,j)}=\underset{\mo\in \mathcal{O}_d}{\operatorname{argmin}} \|\hat{\mx}^{(i)}_{\langle\mathcal{U}_{i}\cap \mathcal{U}_j\rangle}\mo-\hat{\mx}^{(j)}_{\langle\mathcal{U}_{i}\cap \mathcal{U}_j\rangle}\|_F,
\end{equation}
where $\hat\mx^{(i)}_{\langle\mathcal{U}_i\cap \mathcal{U}_j\rangle}$ and $\hat\mx^{(j)}_{\langle\mathcal{U}_i\cap \mathcal{U}_j\rangle}$ are the rows of $\hat\mx^{(i)}$ and $\hat\mx^{(j)}$ corresponding to the shared entities in $\mathcal{U}_i \cap \mathcal{U}_j$. The problem in \eqref{eq:procrustes_adjacent} admits a closed-form solution via singular value decomposition (SVD). Let $\mb^{(i,j)}  \colonequals  (\hat{\mx}^{(i)}_{\langle\mathcal{U}_i\cap \mathcal{U}_j\rangle})^\top \hat{\mx}^{(j)}_{\langle\mathcal{U}_i\cap \mathcal{U}_j\rangle}$ and let $\mb^{(i,j)} = \muu_{\mb^{(i,j)}} \mSigma_{\mb^{(i,j)}} \mv_{\mb^{(i,j)}}^\top$ be its SVD. Then the optimal alignment matrix is $\hat\mw^{(i,j)} = \muu_{\mb^{(i,j)}} \mv_{\mb^{(i,j)}}^\top$.
With such an alignment matrix, $\hat{\mx}^{(i)}\hat\mw^{(i,j)}$ and $\hat{\mx}^{(j)}$ can be used together to recover $\mpp$; for example, $\mpp_{\mathcal{U}_i,\mathcal{U}_j}$ can be estimated as $\hat\mpp_{\mathcal{U}_i,\mathcal{U}_j}=(\hat\mx^{(i)}\hat\mw^{(i,j)})\hat\mx^{(j)\top}$.

CMMI can also align submatrices through a chain of overlaps. It constructs an undirected graph $\mathcal{G}$ with $m$ vertices $\{v_1, \ldots, v_m\}$, where $v_i$ corresponds to $\ma^{(i)}$ and its estimated latent position matrix $\hat\mx^{(i)}$. Vertices $v_i$ and $v_j$ are adjacent if and only if $|\mathcal{U}_i \cap \mathcal{U}_j| \geq d$, and each edge has an alignment matrix $\hat\mw^{(i,j)}$ computed as above. Two submatrices are \textit{connected} if there is a path between their vertices in $\mathcal{G}$. Composing the alignment matrices along such paths allows CMMI to integrate all connected submatrices.

We assume that $\mathcal{G}$ is connected, since the procedure can otherwise be applied to each connected component separately. If $\mathcal{G}$ contains cycles, different paths between the same pair of vertices may give different alignments. CMMI chooses a spanning tree to give each pair a unique connecting path and consistently align all latent position matrices $\{\hat{\mx}^{(i)}\}$. Figure~\ref{fig:G} shows an example with $m=10$ vertices and a spanning tree in red. For entities appearing in multiple submatrices, CMMI then takes a (weighted) average of their aligned latent position estimates to further reduce the estimation error. See \citet{zheng2026chain} for details on alignment along paths, spanning tree selection, and averaging weights.
\begin{figure}[htbp]
    \centering
    \includegraphics[width=0.25\textwidth]{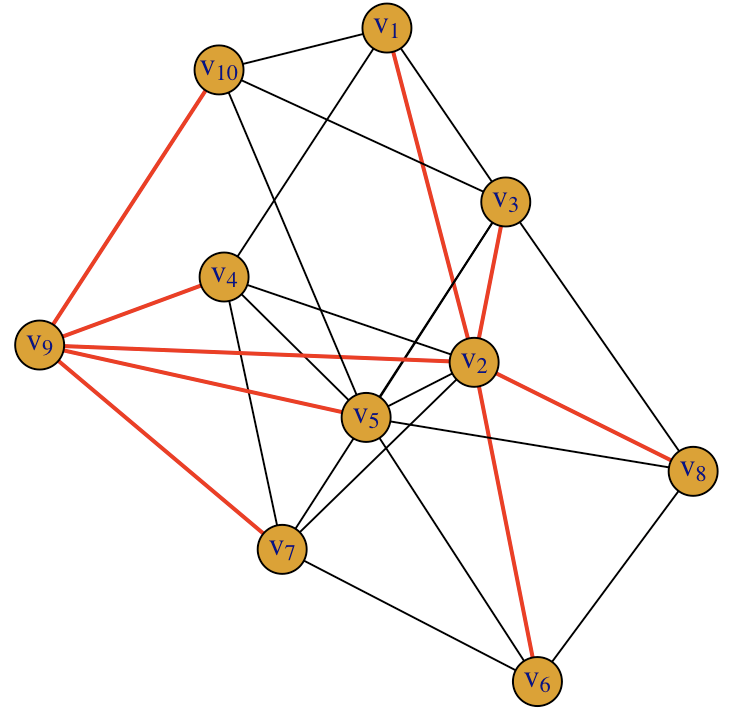}
    \caption{\footnotesize Illustration of a graph $\mathcal{G}$ with $m=10$ vertices containing cycles. Red edges represent a spanning tree that can be used by CMMI for alignment, while black edges indicate unused information about overlapping subsets of entities.}
\label{fig:G}
\end{figure}

The total error has two sources: local estimation error in each $\hat{\mx}^{(i)}$ and cross-submatrix alignment error in integrating $\{\hat{\mx}^{(i)}\}$. Under certain conditions, including a bounded number of sources, \citet{zheng2026chain} prove that the alignment error is negligible compared to the local estimation error (see Remark~5 of \citet{zheng2026chain}). With many sources, however, the spanning tree may have longer paths along which alignment errors accumulate. These errors can then become substantial, making their reduction important for accurate recovery.

CMMI uses only the $m-1$ overlaps corresponding to the spanning tree edges, shown in red in Figure~\ref{fig:G}. The remaining edges, shown in black, provide unused alignment information. With many sources and a rich overlap structure, the graph can have far more than $m-1$ edges, so the tree may discard most of this information.
Using all overlaps may therefore improve the alignment of $\{\hat{\mx}^{(i)}\}$ and the recovery of $\mx$ and $\mpp$, especially in this regime.

CMMI is also computationally efficient. For fixed rank $d$, the SVD solution of \eqref{eq:procrustes_adjacent} takes $O(|\mathcal{U}_i \cap \mathcal{U}_j|)$ operations per alignment. Given $\mathcal{G}$, computing a minimum spanning tree takes $O(E\log m)$ operations, where $E$ is the number of overlapping pairs. The $m-1$ alignments along the tree take $O(m\bar{s})$ operations, where $\bar{s}$ is the average overlap size. The total alignment cost is thus $O(m\bar{s} + E\log m)$, allowing CMMI to scale to large datasets. Experiments in \citet{zheng2026chain} show that, under block-structured missing patterns, CMMI is faster and considerably more accurate than previous methods for general matrix completion applied to matrix integration problems.
Our goal below is to use all overlapping pairs to improve alignment while retaining computational efficiency comparable to CMMI.

\subsection{Our proposed Global Synchronized Multiple Matrix Integration}

We now introduce the proposed GSMMI algorithm, which extends CMMI by exploiting all pairwise overlapping information. To fully utilize the overlapping structure, we formulate the alignment problem as a synchronization-type problem. Specifically, we seek transformation matrices $\{\hat{\mw}^{(i)}\}$ such that all estimated latent position matrices $\{\hat{\mx}^{(i)}\}$ are optimally aligned across \textit{all} their overlapping entities:
\begin{equation}\label{eq:partial MLE}
\{\hat{\mw}^{(i)}\}_{i=1}^m = \argmin_{\mo^{(1)},\ldots,\mo^{(m)}\in \mathcal{O}_d}
\sum_{i < j : \mathcal{U}_i \cap \mathcal{U}_j \neq \varnothing}
\pi_{i,j}\cdot \|\hat{\mx}^{(i)}_{\langle\mathcal{U}_i \cap \mathcal{U}_j\rangle}\mo^{(i)}-\hat{\mx}^{(j)}_{\langle\mathcal{U}_i \cap \mathcal{U}_j\rangle}\mo^{(j)}\|_F^2,
\end{equation}
where $\pi_{i,j}$ are weights for each overlapping pair, to be specified in Remark~\ref{rem:weights}. 

This is a challenging non-convex optimization problem over the product space of orthogonal matrices $(\mathcal{O}_d)^m$. We propose to solve it via an alternating Procrustes optimization approach. We fix $\hat{\mw}^{(1)} = \mi_d$ as the reference and iteratively update each $\hat{\mw}^{(i)}$ for $i = 2, \ldots, m$. For a fixed set of transformation matrices $\{\hat{\mw}^{(j)}\}_{j \neq i}$, the optimal update for $\hat{\mw}^{(i)}$ that minimizes \eqref{eq:partial MLE} is
\begin{equation}\label{eq:procrustes_update}
\hat{\mw}^{(i)} = \argmin_{\mo \in \mathcal{O}_d} \sum_{j \neq i : \mathcal{U}_i \cap \mathcal{U}_j \neq \varnothing} \pi_{i,j} \|\hat{\mx}^{(i)}_{\langle\mathcal{U}_i \cap \mathcal{U}_j\rangle} \mo - \hat{\mx}^{(j)}_{\langle\mathcal{U}_i \cap \mathcal{U}_j\rangle} \hat{\mw}^{(j)}\|_F^2.
\end{equation}
The solution to \eqref{eq:procrustes_update} admits a closed-form expression via SVD. Let
$\mb^{(i)}  \colonequals  \sum_{j \neq i : \mathcal{U}_i \cap \mathcal{U}_j \neq \varnothing} \pi_{i,j} (\hat{\mx}^{(i)}_{\langle\mathcal{U}_i \cap \mathcal{U}_j\rangle})^\top \hat{\mx}^{(j)}_{\langle\mathcal{U}_i \cap \mathcal{U}_j\rangle} \hat{\mw}^{(j)},$
and let $\mb^{(i)} = \muu_{\mb^{(i)}} \mSigma_{\mb^{(i)}} \mv_{\mb^{(i)}}^\top$ be its SVD. Then $\hat{\mw}^{(i)} = \muu_{\mb^{(i)}} \mv_{\mb^{(i)}}^\top$.
We cycle through all submatrices $i = 2, \ldots, m$ and update each $\hat{\mw}^{(i)}$ until convergence, which occurs when $\sum_{i=2}^m \|\hat{\mw}^{(i)}_{\text{new}} - \hat{\mw}^{(i)}_{\text{old}}\|_F^2 < \varepsilon^2$ for some tolerance $\varepsilon > 0$ (e.g., $\varepsilon = 10^{-6}$). This alternating optimization decreases the objective in \eqref{eq:partial MLE} and converges to a stationary point; when initialized well, it typically attains a near-optimal solution in practice.

After obtaining $\{\hat{\mw}^{(i)}\}$ through alternating optimization, we follow the same aggregation strategy as CMMI. For each entity $k \in [N]$, we average its aligned latent positions across all submatrices containing it by
\begin{equation}\label{eq:aggregation}
\hat{\bm{x}}_k=\frac{\sum_{i: k\in\mathcal{U}_i}
\tau_i \cdot (\hat{\mx}^{(i)}\hat{\mw}^{(i)})_{\langle k\rangle}}{\sum_{i :  k\in\mathcal{U}_i}
\tau_i},
\end{equation}
where the $\tau_i$ are weights for each source, to be specified in Remark~\ref{rem:weights}.
We then form $\hat{\bm{X}} = [\hat{\bm{x}}_1|\cdots|\hat{\bm{x}}_N]^\top$, whose rows are the estimated latent positions of all $N$ entities and can be used to support downstream tasks such as clustering, visualization, and community detection. The estimate of the population matrix $\mpp$ is then $\hat{\mpp}=\hat{\mx}\hat{\mx}^\top$.

Note that the quality of the initialization is crucial for the alternating optimization, affecting both its computational efficiency and the proximity of the converged solution to the global optimum. We therefore adopt the CMMI approach described above to first obtain an initial alignment $\{\hat{\mw}^{(i)}_{\text{init}}\}$, which we then refine with the above GSMMI. Specifically, CMMI selects a spanning tree of the graph $\mathcal{G}$ and obtains the initial alignment for each submatrix by composing the pairwise alignments along the path from the root, which is taken to be the first submatrix with $\hat{\mw}^{(1)}_{\text{init}} = \mi_d$. With this high-quality initialization, the alternating GSMMI optimization can converge rapidly to a stationary point that improves on the initial CMMI estimate.
The complete procedure is summarized in Algorithm~\ref{alg:gsmmi}.

\begin{algorithm}[htbp]
\caption{GSMMI for Positive Semidefinite Matrices}
\label{alg:gsmmi}
\begingroup
\small
\setlength{\baselineskip}{10pt}
\setlength{\itemsep}{0pt}
\setlength{\parskip}{0pt}
\setlength{\abovedisplayskip}{4pt}
\setlength{\belowdisplayskip}{4pt}
\begin{algorithmic}
\REQUIRE Observed submatrices $\{\bm{A}^{(i)}\}_{i\in[m]}$ with their index sets $\{\mathcal{U}_i\}_{i\in[m]}\subseteq [N]$, embedding dimension $d$, pairwise alignment weights $\{\pi_{i,j}\}_{i\neq j\in[m]}$ and aggregation weights $\{\tau_i\}_{i\in[m]}$ (see Remark~\ref{rem:weights} for choices).

\STATE \textbf{Step 1: Obtain local latent position estimations $\{\hat{\bm{X}}^{(i)}\}_{i\in[m]}$}
\STATE Obtain $\hat{\bm{X}}^{(i)}\in\mathbb{R}^{n_i\times d}$ such that $\bm{A}^{(i)}\approx \hat{\bm{X}}^{(i)}\hat{\bm{X}}^{(i)\top}$ for each $i\in[m]$, where $n_i=|\mathcal{U}_i|$ (see Remark~\ref{rem:individual est} for details).

\STATE \textbf{Step 2: Obtain alignment matrices $\{\hat{\bm{W}}^{(i)}\}_{i\in[m]}$ via synchronization}
\STATE Solve for transformation matrices $\{\hat{\bm{W}}^{(i)}\}_{i\in[m]}$ by
\[
\{\hat{\bm{W}}^{(i)}\}_{i=1}^m = \argmin_{\bm{O}^{(1)},\ldots,\bm{O}^{(m)}\in \mathcal{O}_d}
\sum_{\substack{(i,j): \\\mathcal{U}_i \cap \mathcal{U}_j \neq \varnothing}}
\pi_{i,j}\cdot \|\hat{\bm{X}}^{(i)}_{\langle\mathcal{U}_i \cap \mathcal{U}_j\rangle}\bm{O}^{(i)}-\hat{\bm{X}}^{(j)}_{\langle\mathcal{U}_i \cap \mathcal{U}_j\rangle}\bm{O}^{(j)}\|_F^2,
\]
using alternating optimization as in \eqref{eq:procrustes_update} with $\hat{\bm{W}}^{(1)} = \mi_d$ fixed and initialization from CMMI.

\STATE \textbf{Step 3: Aggregate $\{\hat{\bm{X}}^{(i)}\hat{\bm{W}}^{(i)}\}_{i\in[m]}$}
\STATE For each entity $k\in[N]$, compute
    $\hat{\bm{x}}_k = \frac{\sum_{i: k\in\mathcal{U}_i} \tau_i \cdot (\hat{\bm{X}}^{(i)}\hat{\bm{W}}^{(i)})_{\langle k\rangle}}{\sum_{i: k\in\mathcal{U}_i} \tau_i}.$
\STATE Form $\hat{\bm{X}} = [\hat{\bm{x}}_1|\cdots|\hat{\bm{x}}_N]^\top$ and compute $\hat{\bm{P}} = \hat{\bm{X}}\hat{\bm{X}}^\top$.

\ENSURE Estimated latent position matrix $\hat{\bm{X}}$ and estimated complete matrix $\hat{\bm{P}}$.
\end{algorithmic}
\endgroup
\end{algorithm}

\begin{Remark}\label{rem:individual est}
To obtain $\hat{\bm{X}}^{(i)}\in\mathbb{R}^{n_i\times d}$ such that $\bm{A}^{(i)}\approx \hat{\bm{X}}^{(i)}\hat{\bm{X}}^{(i )\top}$ for each $i\in[m]$, various methods can be employed depending on the scenario. For example, if $\mathbb{E}(\bm{A}^{(i)})=\bm{P}^{(i)}$ with a relatively low noise level, one can simply use the scaled leading eigenvectors $\hat{\bm{X}}^{(i)}=\hat{\bm{U}}^{(i)}(\hat{\bm{\Lambda}}^{(i)})^{1/2}$, where $\hat{\bm{\Lambda}}^{(i)}=\diag(\hat{\lambda}^{(i)}_1,\dots,\hat{\lambda}^{(i)}_d)$ and $\hat{\bm{U}}^{(i)}=[\hat{\bm{u}}^{(i)}_1|\dots|\hat{\bm{u}}^{(i)}_d]$ contain the top-$d$ eigenvalues and eigenvectors of $\bm{A}^{(i)}$, respectively.
When the noise level is more substantial, with $\bm{N}^{(i)}$ having approximately independent upper-triangular entries satisfying certain moment conditions (e.g., sub-Gaussian), a debiasing transformation can further improve accuracy.
In additive spiked Wigner models, supercritical signal eigenvalues produce outliers whose locations and eigenvector overlaps have explicit limits \citep{benaych2011eigenvalues}.
The transformations
\begin{align}\label{eq:eigen-debiasing}
    \phi(\hat{\lambda}^{(i)}_k) \colonequals \frac{\hat{\lambda}^{(i)}_k+\sqrt{(\hat{\lambda}^{(i)}_k)^2-4\hat{\sigma}_i^2n_i}}{2},
    \quad
    \psi(\hat{\bm{u}}_k^{(i)}) \colonequals \frac{1}{\sqrt{1 - \frac{\hat{\sigma}_i^2 n_i}{(\phi(\hat{\lambda}^{(i)}_k))^2}}}\cdot \hat{\bm{u}}_k^{(i)}
\end{align}
invert the limiting eigenvalue shift and compensate for the attenuation of the corresponding eigenvectors along the signal directions, respectively.
One can then form
$
    \hat{\bm{X}}^{(i)} = [\hat{\bm{x}}^{(i)}_1|\cdots|\hat{\bm{x}}^{(i)}_d]$, $ 
    \hat{\bm{x}}^{(i)}_k = \phi^{1/2}(\hat{\lambda}^{(i)}_k)\cdot \psi(\hat{\bm{u}}_k^{(i)}).
$
The noise variance can be estimated by
$
    \hat{\sigma}_i^2 \colonequals \|\bm{A}^{(i)}-\hat{\bm{P}}^{(i)}\|_F^2/n_i^2$, $
    \hat{\bm{P}}^{(i)} = \hat{\bm{U}}^{(i)} \hat{\bm{\Lambda}}^{(i)}\hat{\bm{U}}^{(i)\top}.
$
Additionally, when $\bm{A}^{(i)}$ has close to independent and uniformly distributed missing entries, low-rank matrix completion methods (e.g., \citet{hastie2015matrix, chatterjee2015matrix, troyanskaya2001missing}) can be applied to obtain $\hat{\bm{X}}^{(i)}$.
\end{Remark}


\begin{Remark}
\label{rem:weights}
The weights $\{\pi_{i,j}\}$ and $\{\tau_i\}$ are left as inputs to Algorithm~\ref{alg:gsmmi}. The simplest choice is uniform weighting, $\pi_{i,j} \equiv 1$ and $\tau_i \equiv 1$, which treats all overlapping pairs and all submatrices equally.

When the submatrices differ substantially in estimation quality, it is advantageous to downweight less reliable sources. A natural choice is inverse-variance weighting,
$
\pi_{i,j} = (c_i^2 + c_j^2)^{-1}$,
$
\tau_i = c_i^{-2},
$
where $c_i$ measures the average estimation error of the latent positions $\hat{\bm{X}}^{(i)}$. This gives greater weight to overlaps between accurately estimated submatrices and to accurately estimated submatrices themselves.

The error level $c_i$ can be estimated from the data. Results in the literature (e.g., Lemma~D.1 in \citet{zheng2026chain}, Eq.~(4) in \citet{xie2024entrywise}) show that, under mild conditions,
$
\hat{\bm{X}}^{(i)}\bm{O}^* - \bm{X}_{\mathcal{U}_i} \approx (\bm{A}^{(i)}-\bm{P}^{(i)})\bm{X}_{\mathcal{U}_i}(\bm{X}_{\mathcal{U}_i}^\top\bm{X}_{\mathcal{U}_i})^{-1},
$
where $\bm{O}^*$ is the optimal orthogonal alignment, so that each row of the estimation error has covariance approximately $\sigma_i^2(\bm{X}_{\mathcal{U}_i}^\top\bm{X}_{\mathcal{U}_i})^{-1}$. This motivates the plug-in estimate $c_i^2 = \hat{\sigma}_i^2\,\operatorname{tr}\big((\hat{\bm{X}}^{(i)\top}\hat{\bm{X}}^{(i)})^{-1}\big)$, where $\hat{\sigma}_i^2 = \|\bm{A}^{(i)}-\hat{\bm{P}}^{(i)}\|_F^2/n_i^2$ and $\hat{\bm{P}}^{(i)} \colonequals \hat{\bm{X}}^{(i)}\hat{\bm{X}}^{(i)\top}$. In practice, we can clip the $c_i$ so that $\max_i c_i / \min_i c_i \le 10$ to avoid extreme weights.
\end{Remark}

\begin{Remark}\label{rem:comp}
The proposed GSMMI algorithm maintains the computational efficiency of CMMI. Specifically, each iteration of the alternating optimization requires $O(E\bar{s})$ operations, where $E\bar{s}$ is the total size of all pairwise overlaps.
With CMMI initialization, GSMMI typically converges within tens of iterations in our experiments 
(see Sections~\ref{sec:simu} and \ref{sec:real}).

In Appendix~\ref{sec:init CMMI}, we compare GSMMI with a cold-start variant (GSMMI-cold) whose initialization is determined by the CVXR solver without CMMI information.
GSMMI is tens of times faster in this comparison, with nearly identical accuracy.
These results suggest that CMMI initialization can reduce the computational cost.

For positive semidefinite matrices, the objective in \eqref{eq:partial MLE} is a little Grothendieck problem over the orthogonal group (see Appendix~\ref{sec:init CMMI}).
For such problems, \citet{bandeira2016approximating} propose an approximation algorithm (SDP-approx) based on semidefinite programming relaxation and Gaussian random projections that recover feasible orthogonal matrices.
In our comparison, GSMMI is several times faster than SDP-approx and achieves higher accuracy.
Similar approximation algorithms are not readily available for indefinite or rectangular matrices, where CMMI initialization also provides a practical starting point.
\end{Remark}

Overall, GSMMI improves upon CMMI by exploiting all pairwise overlaps through global synchronization, while retaining comparable computational efficiency. As demonstrated in our later experiments, including both simulations and real-data analyses, GSMMI consistently achieves higher accuracy than CMMI, with significant improvement in some typical settings, at only a modest additional computational cost.

\section{Theoretical results}\label{sec:theory}

Since GSMMI and CMMI differ only in their alignment strategy, our comparison of the two methods focuses on the difference in alignment accuracy.
As discussed above, the more comprehensive use of the overlapping structure in GSMMI intuitively suggests improved alignment accuracy compared with CMMI, particularly when the number of data sources is large and the overlapping patterns are rich.
To formalize this intuition, we present in this section a theoretical analysis comparing the two methods in a simplified setting that still retains the essential structure of the problem.

\subsection{Simplified synchronization model}

Specifically, we propose a model for the inputs into the synchronization subproblem, which approximately describes the results after principal component analysis is applied to the observations $\ma^{(i)}$.

Recall that the true signal matrix $\mpp$ is assumed to be symmetric positive semidefinite of size $N\times N$ with rank $d$, with eigendecomposition $\mpp=\muu\mLambda\muu^\top$, where $\mLambda=\diag(\lambda_1,\dots,\lambda_d)$ contains the nonzero eigenvalues of $\mpp$ in descending order and the orthonormal columns of $\muu=[\bm{u}_1|\dots|\bm{u}_d]\in\mathcal{O}_{N\times d}$ are the corresponding eigenvectors. We assume that the eigenvalues are distinct. 
We write $\mx= \muu\mLambda^{1/2}$ for the latent positions, so that $\mpp=\mx \mx^{\top}$.

For each source $i$, consider the population submatrix $\mpp^{(i)}=\mpp_{\mathcal{U}_i,\mathcal{U}_i}=\mx_{\mathcal{U}_i}\mx_{\mathcal{U}_i}^\top$ of size $n_i\times n_i$, with eigendecomposition $\mpp^{(i)}=\muu^{(i)}\mLambda^{(i)}\muu^{(i)\top}$, where $\mLambda^{(i)}=\operatorname{diag}(\lambda_1^{(i)},\dots,\lambda_d^{(i)})$ contains its eigenvalues in descending order and $\muu^{(i)}=[\bm{u}_1^{(i)}|\dots|\bm{u}_d^{(i)}]\in\mathcal{O}_{n_i\times d}$ the corresponding eigenvectors. We write $\mx^{(i)}=\muu^{(i)}(\mLambda^{(i)})^{1/2}$, so that $\mpp^{(i)}=\mx^{(i)}\mx^{(i)\top}$. Note that $\mx_{\mathcal{U}_i}$ and $\mx^{(i)}$ are equal up to multiplication on the right by some $d\times d$ orthogonal matrix.

We consider the setting in which all $m$ observed blocks have the same size: each entity set $\mathcal{U}_i$ has size $n_i\equiv n = \lfloor\alpha N\rfloor$ for some $\alpha \in (0,1]$. Each observed matrix $\ma^{(i)} \in \mathbb{R}^{n\times n}$ is given by $\ma^{(i)} = \mpp^{(i)} + \mn^{(i)}$, where the symmetric noise matrix $\mn^{(i)}$ is assumed to have independent off-diagonal entries $\mn_{s,t}^{(i)} \sim \mathcal{N}(0, \sigma_i^2)$ for $s < t$ (with $\mn_{t,s}^{(i)} = \mn_{s,t}^{(i)}$) and independent diagonal entries $\mn_{s,s}^{(i)} \sim \mathcal{N}(0, 2\sigma_i^2)$. We further take the noise level to be constant across blocks, $\sigma_i \equiv \sigma$.

As noted in Remark~\ref{rem:individual est}, a debiasing transformation can further improve the accuracy of the estimate $\hat{\bm{X}}^{(i)}$ in the Gaussian noise setting, which we adopt throughout our analysis. Specifically, let $\hat{\bm{\Lambda}}^{(i)}=\diag(\hat{\lambda}^{(i)}_1,\dots,\hat{\lambda}^{(i)}_d)$ contain the top-$d$ eigenvalues of $\bm{A}^{(i)}$ in descending order and $\hat{\bm{U}}^{(i)}=[\hat{\bm{u}}^{(i)}_1|\dots|\hat{\bm{u}}^{(i)}_d]$ the corresponding eigenvectors. We estimate the noise variance by $\hat{\sigma}_i^2  \colonequals  \|\bm{A}^{(i)}-\hat{\bm{P}}^{(i)}\|_F^2/n_i^2$, where $\hat{\bm{P}}^{(i)}=\hat{\bm{U}}^{(i)} \hat{\bm{\Lambda}}^{(i)}\hat{\bm{U}}^{(i)\top}$, and apply the debiasing transformations in \eqref{eq:eigen-debiasing} to form the debiased estimates $\hat{\bm{x}}^{(i)}_k = \phi^{1/2}(\hat{\lambda}^{(i)}_k)\cdot \psi(\hat{\bm{u}}_k^{(i)})$ for $k\in[d]$, and set $\hat{\bm{X}}^{(i)} = [\hat{\bm{x}}^{(i)}_1|\cdots|\hat{\bm{x}}^{(i)}_d]$. Throughout this section, we write ${\bm{x}}^{(i)}_k$ for the $k$th column of $\mx^{(i)}$, corresponding to the $k$th latent dimension.
The following theorem characterizes the asymptotic distribution of the debiased estimator $\hat{\bm{x}}^{(i)}_k$ around ${\bm{x}}^{(i)}_k$; it is in the spirit of various similar theorems in the literature about Gaussian fluctuations of eigenvectors of supercritical spiked matrix models \citep{LCC-2024-GaussianFluctuationsEigenvectors,CK-2025-EigenvectorFluctuationsSpikedMatrix}.

\begin{Theorem}\label{thm:xi}
Fix a source $i\in[m]$ and take $d$ fixed. Suppose $n\to\infty$ and $\lambda_k^{(i)}/(\alpha\lambda_k)\xrightarrow{\mathbb{P}}1$ for each $k\in[d]$. For each $k\in[d]$, define $\rho_k \colonequals \sqrt{\frac{\alpha\lambda_k^2}{\sigma^2N}}$ and assume $\rho_k\to\rho_k^*$, where the constants $\rho_1^*,\ldots,\rho_d^*>0$ are pairwise distinct. Also assume the incoherence bound $\|\bm{U}^{(i)}\|_{2\to\infty}\lesssim d^{1/2}n^{-1/2}$. Then, for any $k\in[d]$ with $\rho_k^*> 1$ and any finite ordered index set $\mathcal{I}\subset[n]$, we have
\begin{equation*}
\frac{\sqrt{\alpha\lambda_k}}{\sigma}\sqrt{1-\frac{\sigma^2N}{\alpha\lambda_k^2}}
\left[\hat{\bm{x}}_k^{(i)}-\bm{x}_k^{(i)}\right]_{\mathcal{I}}
\xrightarrow{\mathcal{D}}
\mathcal{N}(\bm{0},\bm{I}_{|\mathcal{I}|}),
\end{equation*}
where we resolve the sign ambiguity in $\hat{\bm{x}}_k^{(i)}$ by making it have positive correlation with $\bm x_k^{(i)}$.
Set $\hat{\bm x}_k^{(i)}=\bm0$ when the debiasing factors are not real and positive.
Here the limit is as $N \to \infty$ and in distribution.
\end{Theorem}

\begin{Remark}\label{rmk:cond_xi}
Theorem~\ref{thm:xi} concerns each source $i$ individually, but we express its assumptions and conclusion through the global parameters $\lambda_k$ and $N$ to simplify the subsequent analysis. The condition $\lambda_k^{(i)}/(\alpha\lambda_k)\xrightarrow{\mathbb{P}}1$ connects the block and global scales. 
This condition is mild. For instance, for fixed $\alpha$, it holds almost surely when the eigenvectors $\muu$ are sufficiently incoherent and each index in $\mathcal{U}_i$ is sampled uniformly at random from $[N]$; see Lemma~\ref{lemma:lambdai and lambda} for details.
The incoherence bound of $\bm{U}^{(i)}$ controls the magnitude of each coordinate of the population eigenvectors, ensuring that no single entity dominates the latent structure. It is also mild, since the incoherence of $\bm{U}^{(i)}$ is implied, with high probability, by incoherence of the global eigenvectors $\muu$ together with uniform random sampling of the indices in $\mathcal{U}_i$ (see Remark~1 in \cite{zheng2026chain}).
The distinct limits $\rho_k^*$ keep the corresponding sample eigenvectors from mixing within a nearly degenerate eigenspace.
\end{Remark}

Motivated by Theorem~\ref{thm:xi}, in our model of the synchronization subproblem, we propose that the $\widehat{\mx}^{(i)}$ are formed as the submatrices $\mx_{\mathcal{U}_i}$ of the signal, subject to a random orthogonal transformation in $\mathcal{O}_d$, and further corrupted by Gaussian noise.
In notation, for each source $i \in [m]$, we suppose that we observe
\begin{equation}\label{eq:simple-sync}
	\hat{\mx}^{(i)} = \mx_{\mathcal{U}_i}\mw^{(i){\top}} + \bm\Delta^{(i)}\bm\Gamma,
\end{equation}
where $\mx \in \RR^{N \times d}$ is the full signal matrix, $\mx_{\mathcal{U}_i}$ is its restriction to the entities indexed by set $\mathcal{U}_i \subseteq [N]$ observed by source $i$, $\mw^{(i)} \in \mathcal{O}_d$ is an orthogonal matrix, $\bm\Delta^{(i)} \in \RR^{n \times d}$ has i.i.d.\ $\mathcal{N}(0, 1)$ entries, and $\bm\Gamma = \mathrm{diag}(\gamma_1, \dots, \gamma_d)$ with $\gamma_i > 0$ gives the scaling of the noise level in each dimension.
To match the asymptotic behavior in Theorem~\ref{thm:xi}, we set these scalings in terms of the parameters of the original matrix integration model as
\begin{equation} \label{eq:gamma-lambda}
\gamma_k = \sigma\sqrt{\frac{\lambda_k}{\alpha\lambda_k^2 - \sigma^2 N}} \quad  \text{for } k \in [d].
\end{equation}

Below, we first consider the performance of CMMI and GSMMI on this model of synchronization in the low-rank cases $d = 1$ (in Section~\ref{sec:d=1}), which reduces the orthogonal degree of freedom to a sign flip, and $d = 2$ (in Section~\ref{sec:d=2}), in which case we may still combine the two dimensions into a complex vector and view the orthogonal degree of freedom as a complex phase shift.
Since the analysis of the iterative GSMMI optimization algorithm appears challenging even in these cases, we further consider a proxy spectral algorithm that we believe to behave similarly to GSMMI, as detailed below.
This algorithm is different from GSMMI, but shares the key global behavior of taking into account all pairwise alignment relations at the same time.
We then turn from these algorithm-specific results to a broader discussion of the detection problem for this class of integration tasks, proving information-theoretic and low-degree polynomial lower bounds for detection in the case $d = 1$ both for the simplified synchronization model (in Section~\ref{sec:low_syn}) and for a special case of the original matrix integration model when the noise matrices $\mn^{(i)}$ have i.i.d.\ Gaussian entries up to symmetry (in Section~\ref{sec:low_mat}).

\subsection{Alignment analysis: rank $d=1$}\label{sec:d=1}

We begin with the case $d=1$.
In this case, the global latent position matrix $\bm{X}$ reduces to an $N$-dimensional vector $\bm{x}$, and similarly the local latent position matrices $\bm{X}_{\mathcal{U}_i}$, $\bm{X}^{(i)}$, and $\hat{\bm{X}}^{(i)}$ reduce to $n$-dimensional vectors $\bm{x}_{\mathcal{U}_i}$, $\bm{x}^{(i)}$, and $\hat{\bm{x}}^{(i)}$, respectively.
The orthogonal transformation $\mw^{(i)}$ for $d=1$ reduces to a sign $w^{(i)} \in \{-1,+1\}$.

The simplified model proposed in \eqref{eq:simple-sync} then becomes, for $i \in [m]$, the collection of observations
\begin{equation}\label{eq:simple}
	\hat{\bm{x}}^{(i)} = w^{(i)} \cdot \bm{x}_{\mathcal{U}_i} + \gamma \cdot  \bm{\delta}_i,
\end{equation}
where $\bm{x}\in\mathbb{R}^N$ is the full signal vector, $\mathcal{U}_i \subset [N]$ is a size-$n$ subset with $n = \alpha N$ indexing the entities observed by source $i$, $\bm{x}_{\mathcal{U}_i}$ is the corresponding subvector, $w^{(i)}\in\{-1,+1\}$ is the sign relating $\bm{x}_{\mathcal{U}_i}$ and $\hat{\bm{x}}^{(i)}$, $\bm{\delta}_i \sim \mathcal{N}(\bm{0}, \bm{I}_n)$ is standard Gaussian noise, and $\gamma > 0$ is the noise level, given in terms of the matrix integration problem parameters in \eqref{eq:gamma-lambda}.

\paragraph{Analysis of CMMI.}
For any pair $(i, j)$ with overlapping coordinates $\mathcal{S}_{i,j} \colonequals \mathcal{U}_i \cap \mathcal{U}_j \neq \varnothing$, CMMI estimates the relative sign between sources $i$ and $j$ by %
$$
\hat{w}^{(i,j)} = \argmin_{o \in \{-1,+1\}} \left\| \hat{\bm{x}}^{(i)}_{\langle\mathcal{S}_{i,j}\rangle}\cdot o -  \hat{\bm{x}}^{(j)}_{\langle\mathcal{S}_{i,j}\rangle} \right\|^2
= \operatorname{sgn}\!\left(\langle \hat{\bm{x}}^{(i)}_{\langle\mathcal{S}_{i,j}\rangle}, \hat{\bm{x}}^{(j)}_{\langle\mathcal{S}_{i,j}\rangle} \rangle\right),
$$
where the second equality follows by expanding the squared norm $
\left\| \hat{\bm{x}}^{(i)}_{\langle\mathcal{S}_{i,j}\rangle}\cdot o - \hat{\bm{x}}^{(j)}_{\langle\mathcal{S}_{i,j}\rangle} \right\|^2 = \left\|\hat{\bm{x}}^{(i)}_{\langle\mathcal{S}_{i,j}\rangle}\right\|^2 + \left\|\hat{\bm{x}}^{(j)}_{\langle\mathcal{S}_{i,j}\rangle}\right\|^2 - 2o \langle \hat{\bm{x}}^{(i)}_{\langle\mathcal{S}_{i,j}\rangle}, \hat{\bm{x}}^{(j)}_{\langle\mathcal{S}_{i,j}\rangle} \rangle
$. Based on these pairwise estimates, CMMI builds a graph on the $m$ sources, connecting $i$ and $j$ whenever $\mathcal{S}_{i,j} \neq \varnothing$ and labeling the edge by $\hat{w}^{(i,j)}$. It then selects a spanning tree and propagates the signs along it to align all $\hat{\bm{x}}^{(i)}$ to a common orientation (up to an overall $\pm1$ ambiguity): it fixes $\hat{w}^{(1)} = +1$ at the root, and for each tree edge $(i,j)$ sets $\hat{w}^{(j)} = \hat{w}^{(i,j)} \hat{w}^{(i)}$.

The following theorem establishes a sufficient condition under which CMMI achieves perfect alignment under model~\eqref{eq:simple}.

\begin{Theorem}[(Analysis of CMMI for $d=1$)]\label{thm:cmmi_error}
Consider model~\eqref{eq:simple}. Assume $\rho_1=\sqrt{\frac{\alpha\lambda_1^2}{\sigma^2N}}\to\rho_1^*> 1$, $|\mathcal{S}_{i,j}|\asymp \alpha^2 N$ and $\|\bm{x}_{\mathcal{S}_{i,j}}\|^2\asymp\alpha^2\lambda_1$ uniformly for the overlaps along the spanning tree, and $\max_\ell x_\ell^2 = o(\lambda_1\alpha^2)$. Then CMMI achieves perfect alignment with high probability provided
\[
\alpha^2N=\omega(\log m + \log N).
\]
\end{Theorem}

\begin{Remark}
The detectability threshold condition $\rho_1^*>1$ comes from Theorem~\ref{thm:xi}. %
The conditions $|\mathcal{S}_{i,j}|\asymp\alpha^2 N$ and $\|\bm{x}_{\mathcal{S}_{i,j}}\|^2\asymp\alpha^2\lambda_1$ specify a reasonable size and signal strength for the overlaps used by CMMI, and hold with high probability when the $\{\mathcal{U}_i\}$ are sampled uniformly at random from $[N]$ and independently across $i$.
Finally, the delocalization condition $\max_\ell x_\ell^2 = o(\lambda_1\alpha^2)$ is a mild assumption on the distribution of the signal across coordinates, ensuring that no single coordinate carries too much signal.
\end{Remark}

\paragraph{Analysis of GSMMI.}
Unlike CMMI, GSMMI formulates the alignment problem as a global synchronization problem, minimizing the total disagreement across all pairs:
\begin{equation*}
	\argmin_{o^{(1)},\dots,o^{(m)}\in \{\pm 1\}}
\sum_{i,j\in[m]}
\left\| \hat{\bm{x}}^{(i)}_{\langle\mathcal{S}_{i,j}\rangle} o^{(i)} - \hat{\bm{x}}^{(j)}_{\langle\mathcal{S}_{i,j}\rangle} o^{(j)} \right\|^2,
\end{equation*}
where we take equal weights $\pi_{i,j} \equiv 1$, which is justified in our setting, as all overlapping pairs are of comparable quality.
Letting $A_{i,j}  \colonequals  \langle \hat{\bm{x}}^{(i)}_{\langle\mathcal{S}_{i,j}\rangle}, \hat{\bm{x}}^{(j)}_{\langle\mathcal{S}_{i,j}\rangle} \rangle$ for $i\neq j$, this is equivalent to
$$
\argmax_{o^{(i)} \in \{\pm 1\}} \sum_{i<j} A_{i,j}\, o^{(i)} o^{(j)}.
$$
As a proxy for GSMMI, we consider a spectral relaxation of this problem. Let $\bm{A} = (A_{i,j})_{i,j \in [m]}$ be the pairwise correlation matrix with zero diagonal, and let $\hat{\bm{v}}$ be its leading eigenvector. The estimated signs are then $\hat{w}^{(i)} = \operatorname{sgn}(\hat{\bm{v}}_i)$ for $i \in [m]$.

The behavior of this spectral relaxation depends on the structure of the overlaps. We define the symmetric {overlap matrix} $\bm{S} = (S_{i,j})_{i,j\in[m]}$ by
\begin{equation}\label{eq:S}
S_{i,j}  \colonequals  \|\bm{x}_{\mathcal{S}_{i,j}}\|^2 \ \text{ for } i \neq j, \qquad S_{i,i}  \colonequals  0,
\end{equation}
and let $\lambda_1(\bm{S})$ and $\bm{s}$ %
denote its leading eigenvalue and eigenvector. Since $S_{i,j} = \|\bm{x}_{\mathcal{S}_{i,j}}\|^2 \geq 0$, the matrix $\bm{S}$ is nonnegative; assuming further that $\bm{S}$ is irreducible, i.e., the overlap graph is connected so that all sources can be integrated, the Perron--Frobenius theorem guarantees that $\lambda_1(\bm{S}) > 0$ is simple and admits a corresponding leading eigenvector $\bm{s}$ whose entries are all positive. Beyond these intrinsic properties, we impose the following mild condition on the spectrum of $\bm{S}$.

\begin{Assumption}%
\label{asmp:overlap}
The overlap matrix $\bm{S}$ in \eqref{eq:S} is irreducible, and its eigenvalues and leading eigenvector $\bm{s}$ satisfy:
\begin{enumerate}
\item[(i)] $\lambda_1(\bm{S}) \asymp \lambda_1 \alpha^2 m$ and $\lambda_2(\bm{S}) = o(\lambda_1(\bm{S}))$, where $\lambda_1(\bm{S}), \lambda_2(\bm{S})$ are the largest and second-largest eigenvalues in magnitude;
\item[(ii)] $\min_{i\in[m]} s_i \gtrsim 1/\sqrt{m}$.
\end{enumerate}
\end{Assumption}

\begin{Remark}\label{rem:overlap-examples}
Assumption~\ref{asmp:overlap} requires the overlap structure to be sufficiently connected and homogeneous: (i) ensures a spectral gap with $\lambda_1(\bm{S})$ growing linearly in $m$, and (ii) ensures that every source participates in the overlaps in a balanced way. We give two example settings in which this holds. {(a) Balanced designs:} Suppose the overlaps are balanced, i.e., $\|\bm{x}_{\mathcal{S}_{i,j}}\|^2 = \lambda_1 \alpha^2 (1+o(1))$ uniformly for all $i \neq j$. Then $\bm{S} = \lambda_1 \alpha^2 (\bm{1}\bm{1}^\top - \bm{I}_m)(1+o(1))$, whose leading eigenvalue is $\lambda_1(\bm{S}) = \lambda_1 \alpha^2 (m-1)(1+o(1))$ with leading eigenvector $\bm{s} = \bm{1}/\sqrt{m}$, while all remaining eigenvalues have magnitude $\lambda_1 \alpha^2 (1+o(1)) = o(\lambda_1(\bm{S}))$; hence Assumption~\ref{asmp:overlap} holds. {(b) Uniform random sampling:} Suppose each $\mathcal{U}_i$ is an independent, uniformly random size-$n$ subset of $[N]$ with $n = \alpha N$. Then each coordinate belongs to $\mathcal{S}_{i,j}$ with probability $\alpha^2$, so that $\mathbb{E} (\bm{S}) = \lambda_1 \alpha^2 (\bm{1}\bm{1}^\top - \bm{I}_m)$ (using $\|\bm{x}\|^2 = \lambda_1$ for $d=1$). Provided the signal is delocalized, i.e., $\max_{k} x_k^2 \lesssim \lambda_1/N$, a matrix concentration argument yields $\|\bm{S} - \mathbb{E} (\bm{S})\| = o(\lambda_1 \alpha^2 m)$ with high probability. By Weyl's inequality and an eigenvector perturbation argument, this gives $\lambda_1(\bm{S}) \asymp \lambda_1 \alpha^2 m$, $\lambda_2(\bm{S}) = o(\lambda_1 \alpha^2 m)$, and $\bm{s} \approx \bm{1}/\sqrt{m}$, so Assumption~\ref{asmp:overlap} again holds.
\end{Remark}

The following theorem establishes the condition under which the spectral relaxation of GSMMI achieves perfect alignment under model~\eqref{eq:simple}.

\begin{Theorem}[(Analysis of GSMMI for $d=1$)]\label{thm:gsmmi_error}
Consider model~\eqref{eq:simple}. Assume $\rho_1=\sqrt{\frac{\alpha\lambda_1^2}{\sigma^2N}}\to\rho_1^*> 1$, $\max_{i,j}|\mathcal{S}_{i,j}|\lesssim\alpha^2 N$ and $\max_{i,j}\|\bm{x}_{\mathcal{S}_{i,j}}\|^2\lesssim\lambda_1\alpha^2$, that the overlap matrix $\bm{S}$ defined in \eqref{eq:S} satisfies Assumption~\ref{asmp:overlap}, and that the overlaps $\{\mathcal{S}_{i,j}\}$ are pairwise disjoint. Then the spectral relaxation of GSMMI achieves perfect alignment with high probability provided
$$
\alpha^2 N = \omega(1) \quad \text{and} \quad m \alpha^2 N = \omega(\log^2 N).
$$
\end{Theorem}

\begin{Remark}
The conditions $\max_{i,j}|\mathcal{S}_{i,j}|\lesssim\alpha^2 N$ and $\max_{i,j}\|\bm{x}_{\mathcal{S}_{i,j}}\|^2\lesssim\lambda_1\alpha^2$ require that no overlap is unusually large or carries an unusually large share of the signal. They hold with high probability under uniform random sampling with a delocalized signal.
The pairwise-disjointness assumption on $\{\mathcal{S}_{i,j}\}$ is a technical simplification. Relaxing it requires a more involved analysis, which we leave to future work.
\end{Remark}

Comparing Theorems~\ref{thm:cmmi_error} and~\ref{thm:gsmmi_error}, we see that GSMMI improves upon CMMI when the number of sources $m$ is large. Consider $m = N^{\delta}$ for some $\delta > 0$. The CMMI condition $\alpha^2 N = \omega(\log m + \log N)$ becomes
$
\alpha^2 N = \omega(\log N),$ i.e., $\alpha = \omega\!\left(N^{-1/2}\log^{1/2}N\right),
$
while the two GSMMI conditions reduce to
$
\alpha = \omega\!\left(\max\{N^{-1/2},\, N^{-(1+\delta)/2}\log N\}\right) = \omega(N^{-1/2}).
$
Thus, for perfect alignment, GSMMI removes the $\sqrt{\log N}$ factor in the required sampling fraction, reflecting the benefit of aggregating alignment information globally rather than propagating pairwise estimates along a spanning tree. The advantage of global synchronization is even clearer for $d=2$, considered in the next subsection, where we compare the alignment error.

\subsection{Alignment analysis: rank $d=2$}\label{sec:d=2}

We next consider the case $d = 2$ of the simplified model in \eqref{eq:simple-sync}.
Recall that this means that, for each source $i \in [m]$, we observe an $n \times 2$ matrix
\begin{equation}\label{eq:simple_d2}
	\hat{\bm{X}}^{(i)} = \bm{X}_{\mathcal{U}_i} \bm{W}^{(i)\top} + \bm{\Delta}^{(i)} \bm{\Gamma},
\end{equation}
where $\bm{X} \in \mathbb{R}^{N \times 2}$ is the full signal matrix, $\bm{X}_{\mathcal{U}_i}$ is its restriction to the entities observed by source $i$, $\bm{W}^{(i)} \in \mathcal{O}_2$ is the orthogonal transformation relating $\bm{X}_{\mathcal{U}_i}$ and $\bm{X}^{(i)}$, $\bm{\Delta}^{(i)} \in \mathbb{R}^{n \times 2}$ has i.i.d. $\mathcal{N}(0,1)$ entries, and $\bm{\Gamma} = \mathrm{diag}(\gamma_1, \gamma_2)$ with $\gamma_k > 0$ carries the noise level for each dimension.
As before, the $\gamma_k$ are given in \eqref{eq:gamma-lambda} in terms of the matrix integration problem parameters.

\paragraph{Analysis of CMMI.}
For any pair $(i,j)$ with $|\mathcal{S}_{i,j}| \geq 2$, CMMI estimates the pairwise alignment by
$$
	\hat{\bm{W}}^{(i,j)} = \argmin_{\bm{O} \in \mathcal{O}_d}
	\left\| \hat{\bm{X}}^{(i)}_{\langle\mathcal{S}_{i,j}\rangle} \bm{O} - \hat{\bm{X}}^{(j)}_{\langle\mathcal{S}_{i,j}\rangle} \right\|_F.
$$
CMMI then constructs a spanning tree rooted at source $1$ and obtains each $\hat{\bm{W}}^{(i)}$ by composing these pairwise estimates along the unique root-to-node path, starting from $\hat{\bm{W}}^{(1)} = \bm{I}$. As in the $d=1$ case, the alignments $\{\hat{\bm{W}}^{(i)}\}$ are determined only up to a common orthogonal transformation, so here in the analysis we fix the reference coordinate system by setting $\bm{W}^{(1)} = \bm{I}$.

\begin{Theorem}[(Analysis of CMMI for $d=2$)]\label{thm:cmmi_error_d2}
Consider model~\eqref{eq:simple_d2} with $\bm{X}=\bm{U}\bm{\Lambda}^{1/2}$, where $\bm{U}\in\mathcal{O}_{N\times 2}$ and $\bm{\Lambda}=\operatorname{diag}(\lambda_1,\lambda_2)$. Suppose $\|\bm{U}\|_{2\to\infty}\lesssim d^{1/2}N^{-1/2}$, $\lambda_1\asymp\lambda_2$, and $\rho_k=\sqrt{\frac{\alpha\lambda_k^2}{\sigma^2 N}}\to\rho_k^*>1$ for $k=1,2$, and that each overlap $\mathcal{S}_{i,j}$ along the spanning tree satisfies $|\mathcal{S}_{i,j}|\asymp\alpha^2 N$ with all eigenvalues of $\bm{U}_{\mathcal{S}_{i,j}}^\top\bm{U}_{\mathcal{S}_{i,j}}$ of order $\alpha^2$. If $\alpha^2N=\omega(\log^2 N)$, then for any source $i$ at depth $D_i$ in the tree (with $D_i\leq D\leq m-1$, where $D$ is the tree depth),
$$
\|\hat{\bm{W}}^{(i)} - \bm{W}^{(i)}\|_F \lesssim D_i \cdot \frac{\log N}{\alpha N^{1/2}}
$$
with high probability.
\end{Theorem}

\begin{Remark}
The conditions for $\|\bm{U}\|_{2\to\infty}$ and $\lambda_1,\lambda_2$ are satisfied whenever $\bm{P}=\bm{X}\bm{X}^\top$ has bounded condition number and bounded coherence; see, e.g., \cite{abbe2020entrywise,chen2021spectral,recht2011simpler}. The condition for $\rho_k$ comes from Theorem~\ref{thm:xi}. The overlap conditions are the $d=2$ analogues of those in Theorem~\ref{thm:cmmi_error}, requiring each overlap to be sufficiently large and well-conditioned; they hold, for instance, under uniform random sampling with incoherent eigenvectors, where in fact $\bm{U}_{\mathcal{S}_{i,j}}^\top\bm{U}_{\mathcal{S}_{i,j}}\to\alpha^2\bm{I}_d$ (see Lemma~\ref{lemma:U^TU to alpha I}).
\end{Remark}

Notice that a key feature of the error bound is its linear growth in the tree depth $D_i$. Because $\hat{\bm{W}}^{(i)}$ is obtained by composing pairwise estimates along the root-to-node path, the alignment errors accumulate with the path length. Although Theorem~\ref{thm:cmmi_error_d2} is stated for $d=2$, the result also holds for general $d$, as the proof applies with only minor modifications.

\paragraph{Analysis of GSMMI.}

For the analysis of GSMMI for $d=2$, we exploit the structure of $2$-dimensional orthogonal transformations to embed them as complex numbers. Any rotation $\bm{W}^{(i)} \in \mathrm{SO}(2)$, where $\mathrm{SO}(2)$ denotes the group of $2\times 2$ rotation matrices, can be written as
$$
\bm{W}^{(i)} = \begin{bmatrix} \cos\theta_i & -\sin\theta_i \\ \sin\theta_i & \cos\theta_i \end{bmatrix}
$$
for some rotation angle $\theta_i \in [0, 2\pi)$. Setting $\bm{c} = [1, \imath]^\top$ with $\imath=\sqrt{-1}$, and writing $w^{(i)}  \colonequals  e^{-\imath\theta_i} \in \mathbb{S}^1$, where $\mathbb{S}^1  \colonequals  \{z\in\mathbb{C} : |z|=1\}$ is the unit circle, we have
$$
\bm{W}^{(i)}\bm{c} = e^{-\imath\theta_i}\bm{c} = w^{(i)}\bm{c}, \qquad \bm{W}^{(i)\top}\bm{c} = e^{\imath\theta_i}\bm{c} = \overline{w^{(i)}}\,\bm{c}.
$$
We focus on rotations $\mathrm{SO}(2)$ in the main analysis; reflections, which form the rest of $\mathcal{O}_2$ and take the form $\bm{W}^{(i)} = \begin{bmatrix} \cos\theta_i & \sin\theta_i \\ \sin\theta_i & -\cos\theta_i \end{bmatrix}$, satisfy $\bm{W}^{(i)}\bm{c} = e^{\imath\theta_i}\overline{\bm{c}}$ with $\overline{\bm{c}} = [1, -\imath]^\top$ and can be handled analogously with the same guarantees.

Define the complex signal $\bm{z}  \colonequals  \bm{X}\bm{c} = \bm{x}_{1} + \imath\bm{x}_{2} \in \mathbb{C}^N$ and the complex-valued observations $\hat{\bm{z}}^{(i)}  \colonequals  \hat{\bm{X}}^{(i)} \bm{c} = \hat{\bm{x}}^{(i)}_{1} + \imath\hat{\bm{x}}^{(i)}_{2} \in \mathbb{C}^n$, where ${\bm{x}}_{k}$ and $\hat{\bm{x}}^{(i)}_{k}$ denote the $k$th columns of ${\bm{X}}$ and $\hat{\bm{X}}^{(i)}$, respectively. Right-multiplying model~\eqref{eq:simple_d2} by $\bm{c}$ gives
$$
\hat{\bm{z}}^{(i)} = \bm{X}_{\mathcal{U}_i}\bm{W}^{(i)\top}\bm{c} + \bm{\Delta}^{(i)}\bm{\Gamma}\bm{c}.
$$
Using $\bm{W}^{(i)\top}\bm{c} = \overline{w^{(i)}}\,\bm{c}$ and $\bm{X}_{\mathcal{U}_i}\bm{c} = \bm{z}_{\mathcal{U}_i}$, we obtain
\begin{equation}\label{eq:complex_model}
	\hat{\bm{z}}^{(i)} = \overline{w^{(i)}} \cdot \bm{z}_{\mathcal{U}_i} + \bm{\epsilon}_i,
\end{equation}
where $\bm{\epsilon}_i  \colonequals  \bm{\Delta}^{(i)}\bm{\Gamma}\bm{c} \in \mathbb{C}^n$ is complex Gaussian noise with independent entries
$$
\epsilon_{i,j} = \gamma_1 \Delta^{(i)}_{j,1} + \imath\gamma_2 \Delta^{(i)}_{j,2}.
$$

GSMMI formulates the alignment as the global synchronization problem
\begin{equation*}
	\argmin_{\bm{O}^{(1)},\dots,\bm{O}^{(m)}\in \mathrm{SO}(2)}
\sum_{i,j\in[m]: \mathcal{S}_{i,j} \neq \varnothing}
\left\| \hat{\bm{X}}^{(i)}_{\langle\mathcal{S}_{i,j}\rangle} \bm{O}^{(i)} - \hat{\bm{X}}^{(j)}_{\langle\mathcal{S}_{i,j}\rangle} \bm{O}^{(j)} \right\|_F^2.
\end{equation*}
Since $\|\bm{M}\|_F^2 = \|\bm{M}\bm{c}\|^2$ for any real $n\times 2$ matrix $\bm{M}$, and using $\hat{\bm{z}}^{(i)}=\hat{\bm{X}}^{(i)}\bm{c}$ together with the fact that any $\bm{O}^{(i)}\in\mathrm{SO}(2)$ satisfies $\bm{O}^{(i)}\bm{c}=o^{(i)}\bm{c}$ for some $o^{(i)}\in\mathbb{S}^1$, the problem is equivalent to
$$
\argmin_{o^{(1)},\dots,o^{(m)}\in \mathbb{S}^1}
\sum_{i,j\in[m]: \mathcal{S}_{i,j} \neq \varnothing}
\left\| \hat{\bm{z}}^{(i)}_{\langle\mathcal{S}_{i,j}\rangle} o^{(i)} - \hat{\bm{z}}^{(j)}_{\langle\mathcal{S}_{i,j}\rangle} o^{(j)} \right\|^2.
$$
Expanding the objective with $\|\bm{m}\|^2 = \bm{m}^\dagger \bm{m}$ for $\bm{m} \in \mathbb{C}^n$, where $(\cdot)^\dagger$ denotes the conjugate transpose, this is equivalent to
$$
\max_{o^{(i)} \in \mathbb{S}^1} \sum_{i<j: \mathcal{S}_{i,j} \neq \varnothing} \Re\left( A_{i,j}\, \overline{o^{(i)}} o^{(j)} \right),
$$
where $\Re(\cdot)$ denotes the real part and $A_{i,j}  \colonequals  (\hat{\bm{z}}^{(i)}_{\langle\mathcal{S}_{i,j}\rangle})^\dagger \hat{\bm{z}}^{(j)}_{\langle\mathcal{S}_{i,j}\rangle} \in \mathbb{C}$ for $i \neq j$. To analyze GSMMI, we consider a spectral relaxation. Let $\bm{A} = (A_{i,j})_{i,j \in [m]}$ be the Hermitian matrix with $A_{i,i} = 0$ and $A_{j,i} = \overline{A_{i,j}}$, and let $\hat{\bm{v}} \in \mathbb{C}^m$ be its leading eigenvector. The estimated phases are then $\hat{w}^{(i)} = \hat{\bm{v}}_i / |\hat{\bm{v}}_i|$ for $i \in [m]$.

As in the $d=1$ case, the behavior of this relaxation depends on the overlaps, and we define the overlap matrix $\bm{S} = (S_{i,j})_{i,j\in[m]}$ by
\begin{equation}\label{eq:S_d2}
	S_{i,j}  \colonequals  \|\bm{z}_{\mathcal{S}_{i,j}}\|^2 \ \text{ for } i \neq j, \qquad S_{i,i}  \colonequals  0.
\end{equation}
Since $\bm{S}$ is a nonnegative symmetric matrix with the same structure as in the $d=1$ case, its spectral properties are analogous, and we again impose Assumption~\ref{asmp:overlap}. The following theorem establishes the error guarantee for the spectral relaxation of GSMMI.

\begin{Theorem}[(Analysis of GSMMI for $d=2$)]\label{thm:gsmmi_error_d2}
Consider model~\eqref{eq:complex_model}, where $\bm{z}=\bm{X}\bm{c}$ with $\bm{X}=\bm{U}\bm{\Lambda}^{1/2}$, $\bm{\Lambda}=\operatorname{diag}(\lambda_1,\lambda_2)$, and $\bm{c}=[1,\imath]^\top$, so that $\|\bm{z}\|^2=\lambda_1+\lambda_2$.
Suppose $\lambda_1\asymp\lambda_2$, $\rho_k=\sqrt{\frac{\alpha\lambda_k^2}{\sigma^2 N}}\to\rho_k^*>1$ for $k=1,2$, $\max_{i,j}|\mathcal{S}_{i,j}|\lesssim\alpha^2 N$, $\max_{i,j}\|\bm{z}_{\mathcal{S}_{i,j}}\|^2\lesssim (\lambda_1+\lambda_2)\alpha^2$, that the overlap matrix $\bm{S}$ defined in \eqref{eq:S_d2} satisfies Assumption~\ref{asmp:overlap}, and that the overlaps $\{\mathcal{S}_{i,j}\}$ are pairwise disjoint. If $\alpha^2 N=\omega(1)$ and $m\alpha^2 N=\omega(\log^2 N)$, then the spectral relaxation of GSMMI satisfies, for all $i\in[m]$,
$$
\|\hat{\bm{W}}^{(i)} - \bm{W}^{(i)}\|_F
\lesssim \frac{1}{\alpha N^{1/2}} + \frac{\log N}{\alpha N^{1/2} m^{1/2}}
$$
with high probability.
\end{Theorem}

Comparing Theorems~\ref{thm:cmmi_error_d2} and~\ref{thm:gsmmi_error_d2} reveals a clear advantage of GSMMI over CMMI, especially when the number of sources $m$ is large  under suitable assumptions. For CMMI, the alignment error for each source $i$ depends on its depth $D_i$ in the spanning tree and is of order $\mathcal{O}_p(D_i \log N / (\alpha N^{1/2}))$, where $D_i \leq D \leq m-1$ and $D$ denotes the maximum tree depth; that is, sources deep in the tree may incur substantial errors from the sequential propagation of pairwise alignment errors. In contrast, GSMMI achieves the uniform bound $\mathcal{O}_p((\alpha N^{1/2})^{-1}+\log N/(\alpha N^{1/2}m^{1/2}))$ simultaneously for all sources. By jointly estimating all alignments, GSMMI avoids this error accumulation and yields a sharper bound as $m$ grows.

This confirms the following intuition. Under a suitably rich overlap structure, as $m$ increases, CMMI cannot fully benefit from the additional overlap information, since its alignment procedure uses only the edges of a spanning tree, and the tree depth may even increase with $m$, causing alignment errors to accumulate along tree paths and thereby degrading accuracy. In contrast, GSMMI exploits the increasing redundancy in the overlap structure through global synchronization, so that its alignment error can decrease as $m$ grows. This underlies the growing advantage of global synchronization over spanning-tree-based alignment when integrating a larger number of sources.

\subsection{Lower bounds on detection: rank $d=1$, synchronization model}\label{sec:low_syn}

We next turn to lower bounds for the related task of \emph{detecting} the presence of signal in our simplified synchronization model.
While in general detection and alignment or recovery are distinct tasks, they share the same fundamental thresholds in many natural problems (see, e.g., discussion of detection-recovery gaps in \cite{Wein-2025-LowDegreeSurvey}), and so we expect that these results should also give thresholds below which the integration problem is intractable for any method.
In particular, below in our experiments we will compare the performance of CMMI and GSMMI to such thresholds as a way of understanding how close they come to saturating these limitations on arbitrary computations.

We give these results in terms of the following notion of asymptotic hypothesis testing.
For a sequence of probability measures $\mathbb{Q}_N$ and $\mathbb{P}_N$ over some spaces $\Omega_N$ and a sequence of functions $f_N: \Omega_N \to \{0, 1\}$, which we call \emph{tests} in this context, we say that these tests achieve \emph{strong detection} between the sequences if
\[ \lim_{N \to \infty} \mathbb{P}_N[f_N = 1] = \lim_{N \to \infty} \mathbb{Q}_N[f_N = 0] = 1, \]
and achieve \emph{weak detection} if there is some $\epsilon > 0$ such that
\[ \liminf_{N \to \infty} \mathbb{P}_N[f_N = 1] > \frac{1}{2} + \epsilon, \qquad \liminf_{N \to \infty} \mathbb{Q}_N[f_N = 0] > \frac{1}{2} + \epsilon.
\]

\begin{Theorem}
    \label{thm:lb-sync}
    Consider the model \eqref{eq:simple-sync} in the case $d = 1$, as given in \eqref{eq:simple}.
    Further suppose in this case that we draw the full signal vector $\bm x \sim \mathcal{N}(\bm 0, \bm I_N)$ and draw the signs $w^{(i)} \sim \mathrm{Unif}(\{\pm 1\})$ i.i.d.
    Denote the resulting distribution of the observations $(\hat{\bm{x}}^{(1)}, \dots, \hat{\bm{x}}^{(m)})$ by $\mathbb{P}_{\gamma}$ (leaving the dependence on $N$ and the $\mathcal{U}_i$ implicit).
    Let $\mathbb{Q}$ denote the null distribution under which the vectors $\hat{\bm{x}}^{(i)} \sim \mathcal{N}(\bm{0}, (1+\gamma^2)\bm{I}_{|\mathcal{U}_i|})$, $i\in[m]$, are mutually independent, so that all observations are independent and have the same one-dimensional marginals as under $\mathbb{P}_{\gamma}$.
    We leave the dependence of $\mathbb{Q}$ on $N$, $\gamma$, and the $\mathcal{U}_i$ implicit.
    Define
\begin{align*}
    B_{ij}
    \colonequals |\mathcal{U}_i\cap\mathcal{U}_j|\,
      \bm{1}_{\{i\neq j\}}, \quad
    \eta \colonequals \left(\frac{1}{1 + \gamma^2}\right)^2\|\bm{B}\|_F, \quad
    \zeta \colonequals \lim_{N \to \infty} \frac{1}{2}\log\left(1 + \frac{2}{\gamma^2(\gamma^2 + 2)}\right) \lambda_1(\bm{B}).
\end{align*}
Assume that the sets
$\{\mathcal{U}_i\cap\mathcal{U}_j:1\leq i<j\leq m\}$
are pairwise disjoint. Then the following statements hold
as $N\to\infty$.
\begin{enumerate}
    \item If $\eta=o(1)$, then no sequence of tests achieves weak detection between $\mathbb{P}_{\gamma}$ and $\mathbb{Q}$.
    \item Suppose that $\eta = \Theta(1)$ and the limit defining $\zeta$ exists.
    If $\zeta < 1$, then no sequence of tests achieves strong detection between $\mathbb{P}_{\gamma}$ and $\mathbb{Q}$.
\end{enumerate}
\end{Theorem}

In Theorem~\ref{thm:ub-sync} we also give related results showing that these lower bounds are likely to be tight.
We note that, modulo the assumption on disjoint overlap sets $\mathcal{U}_i \cap \mathcal{U}_j$ and the extra conditions on $\bm{B}$ in Theorem~\ref{thm:ub-sync}, these results taken together cover all possible parameter scalings provided that $\lambda_1(\bm{B}) = \|\bm{B}\|$ and $\|\bm{B}\|_F$ are of the same order as $N \to \infty$.
One natural situation in which this happens is if the $B_{ij} = |\mathcal{U}_i \cap \mathcal{U}_j|$ are all of the same order, in which case both norms are of this same order times $m$.

\subsection{Lower bounds on detection: rank $d=1$, matrix integration model}\label{sec:low_mat}

We also give comparable results for the original model of matrix integration, specialized with distributional assumptions as follows.

\begin{Theorem}
    \label{thm:lb-mx}
    Consider the model \eqref{eq:full-model} in the case $d = 1$.
    Further suppose in this case that we draw the noise matrices $\mn^{(i)}$ to have entries $N^{(i)}_{ab} = N^{(i)}_{ba} \sim \mathcal{N}(0, 1)$ i.i.d.\ for each $a, b \in \mathcal{U}_i$ and $i \in[m]$.
    Also suppose that we draw the full signal vector $\bm x$ to have i.i.d.\ entries $x_i = \frac{\tilde\lambda^{1/2}}{N^{1/4}}x_i^{(0)}$, so that $\mpp = \bm x \bm x^\top = \frac{\tilde\lambda}{\sqrt{N}} \bm x^{(0)} \bm x^{(0)\top}$, for $x_i^{(0)} \sim \nu$ drawn i.i.d.\ for some probability measure $\nu$ on $\RR$ having
    \begin{enumerate}
    \item[(a)] $\E_{x \sim \nu}[x] = 0$;
    \item[(b)] $\E_{x \sim \nu}[x^2] = 1$;
    \item[(c)] $\nu$ is symmetric, meaning that $x$ and $-x$ have the
    same distribution when $x \sim \nu$.
\end{enumerate}
Note that in this case, under mild further regularity conditions on $\nu$ as we assume below, we indeed have $\frac{1}{\sqrt{N}}\lambda_1(\bm P) = \tilde\lambda \cdot \frac{\|\bm x^{(0)}\|^2}{N} \to \tilde\lambda$ as $N \to \infty$ in probability, so that $\tilde\lambda$ is the leading eigenvalue of $\mpp$ normalized by $\sqrt{N}$, which we distinguish by the tilde from the unnormalized eigenvalues $\lambda_k$ of $\mpp$ used elsewhere.
The above assumptions fully specify the distribution of $(\bm A^{(1)}, \dots, \bm A^{(m)})$, which we denote $\mathbb{P}_{\tilde\lambda}$.

Assume further that each coordinate of $[N]$ belongs to a uniformly bounded number of
sets:
\[
    \max_{a \in [N]}
    \#\{i : a \in \mathcal{U}_i\} = O(1).
\]
Define the matrix $\widetilde{\bm B}$, a small variation on $\bm{B}$ defined above, by
$
    \widetilde{B}_{ij} = |\mathcal{U}_i \cap \mathcal{U}_j|$,
for $i,j \in [m].
$
Suppose that
\[
    \tilde\lambda <
    \lim_{N \to \infty}
    \sqrt{\frac{N}{\lambda_1(\widetilde{\bm B})}},
\]
where the limit on the right-hand side exists.
Then, the following hold:
\begin{enumerate}
    \item Suppose, in addition, that for independent $x,x' \sim \nu$, the product
$xx'$ is $\sigma^2$-subgaussian for some constant $\sigma^2 > 0$.
Then, for any sequence $D(N) \in \mathbb{N}$ satisfying $D = o(\frac{\lambda_1(\widetilde{\bm B})^3}{N^2})$,
there is no sequence of polynomials $f_N$ in the entries of $(\bm A^{(1)}, \dots, \bm A^{(m)})$ with
$\deg(f_N) \leq D(N)$ that achieves strong separation (see \eqref{eq:strong-sep} for this notion) between $\mathbb{P}_{\tilde\lambda}$ and $\mathbb{P}_0$.
    \item Suppose, in addition, that for independent
$x,x' \sim \nu$, the product $xx'$ is $1$-moment-subgaussian in the
sense of Definition~3.13 of~\cite{de2026computational}.
Then, no sequence of tests achieves strong detection between $\mathbb{P}_{\tilde\lambda}$ and $\mathbb{P}_0$.
\end{enumerate}
\end{Theorem}

\section{Experiments on synthetic data}\label{sec:simu}

In this section, we investigate the empirical performance of GSMMI through simulation studies with synthetic datasets.
We first compare the estimation accuracy and computational cost of GSMMI and CMMI across a variety of configurations, showing that GSMMI achieves uniformly improved accuracy with comparable computational efficiency (Section~\ref{sec:simu_error}).
We then examine how the advantage of GSMMI over CMMI grows as the number of sources $m$ increases, empirically confirming the theoretical advantages of GSMMI established in Section~\ref{sec:theory} (Section~\ref{sec:simu_vary_m}).
We also consider a setting that mimics realistic multi-source data integration, involving a blockwise missing pattern and a suitably rich overlap structure.
In this setting, we compare GSMMI and CMMI against several existing matrix-completion methods to further illustrate the advantages of GSMMI (Section~\ref{sec:simu_block}).
Finally, we compare the empirical behavior of both methods with the detection threshold derived in Section~\ref{sec:low_mat} (Section~\ref{sec:simu_threshold}).

\subsection{Improved accuracy with comparable efficiency}\label{sec:simu_error}

We conduct simulation experiments under a variety of configurations, specifically across a broad range of signal strengths and block sizes, to comprehensively compare GSMMI and CMMI in terms of estimation accuracy and computational cost.

We generate the population positive semidefinite matrix $\bm{P} = \bm{U}\bm{\Lambda}\bm{U}^\top$ of size $N = 200$ and rank $d = 3$, where $\bm{U} \in \mathcal{O}_{N \times d}$ is randomly generated and $\bm{\Lambda} = \diag(\lambda, \frac{3}{4}\lambda, \frac{1}{2}\lambda)$ with signal strength parameter $\lambda$. We observe $m = 20$ submatrices, each of size $n = \lfloor\alpha N \rfloor$, where the ratio $\alpha$ controls the size of each observed submatrix relative to the entire matrix. For each submatrix $i \in [m]$, we randomly select $n$ entities from $[N]$ to form the index set $\mathcal{U}_i$ and observe the noisy matrix $\bm{A}^{(i)} = \bm{P}_{\mathcal{U}_i, \mathcal{U}_i} + \bm{N}^{(i)}$, where $\bm{N}^{(i)}$ is a symmetric Gaussian noise matrix with off-diagonal entries $N^{(i)}_{j,k} \sim N(0, \sigma^2)$ for $j < k$ and diagonal entries $N^{(i)}_{j,j} \sim N(0, 2\sigma^2)$, with noise level $\sigma = 0.1$.
We vary the signal strength $\lambda$ and the submatrix size ratio $\alpha$ and apply both GSMMI and CMMI to compute the recovered matrix $\hat{\bm{P}}$.

We report matrix correlation $\langle \bm{P}, \hat{\bm{P}} \rangle_F / (\|\bm{P}\|_F \|\hat{\bm{P}}\|_F)$ and relative Frobenius error $\|\bm{P} - \hat{\bm{P}}\|_F / \|\bm{P}\|_F$.
Figure~\ref{fig:simu_accuracy} reports the results, averaged over $20$ Monte Carlo replicates.
GSMMI outperforms CMMI across these settings, especially when $\alpha$ is small and $\lambda$ is moderate.
Summed over all settings and replicates, GSMMI takes $960.6$ seconds in total compared with $286.8$ seconds for CMMI, i.e., about $3.4$ times the running time of CMMI.

\begin{figure}[htbp]
    \centering
    \includegraphics[width=\textwidth]{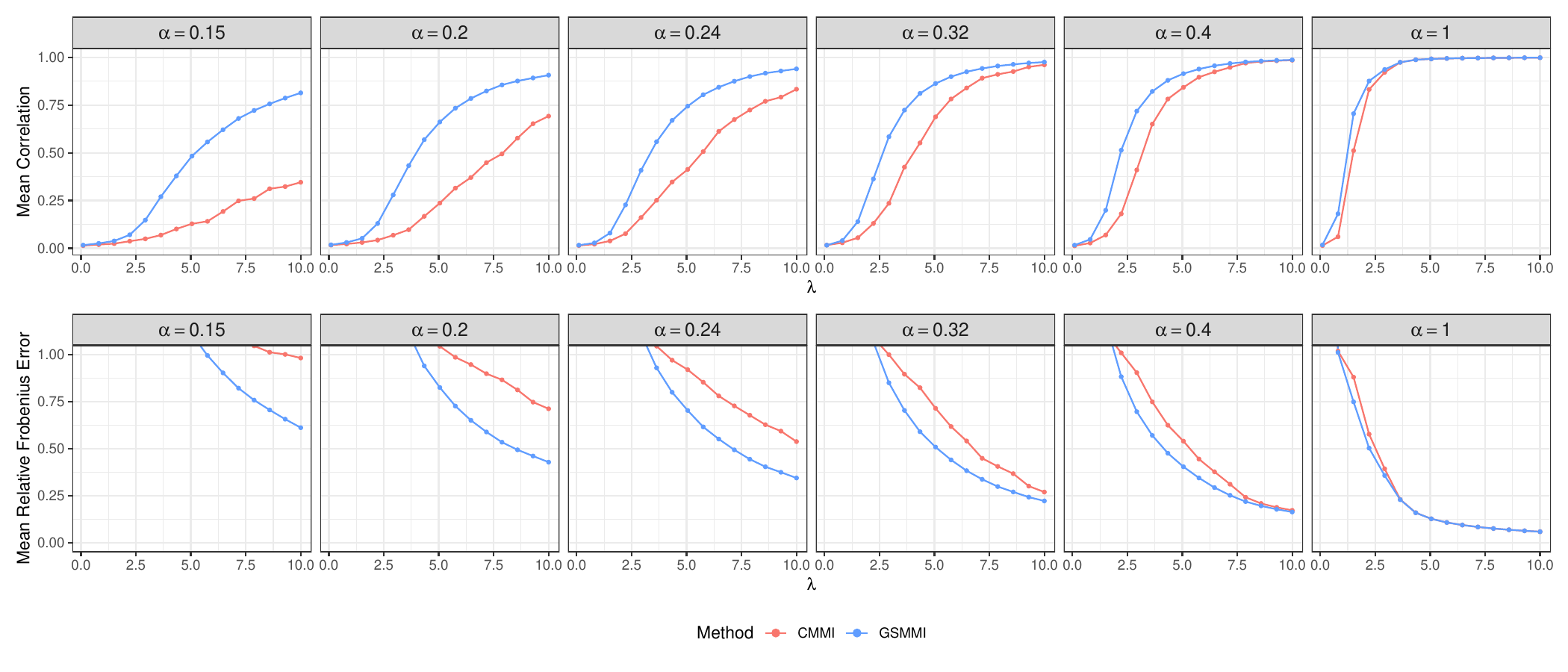}
    \caption{\footnotesize Estimation accuracy of GSMMI and CMMI as we vary the signal strength $\lambda \in [0.1, 10.0]$ (15 equally spaced values) and the submatrix size ratio $\alpha \in \{0.15, 0.2, 0.24, 0.32, 0.4, 1\}$, with $N=200$, $m=20$, $\sigma=0.1$. Top: correlation $\langle \bm{P}, \hat{\bm{P}} \rangle_F / (\|\bm{P}\|_F \|\hat{\bm{P}}\|_F)$. Bottom: relative Frobenius error $\|\bm{P} - \hat{\bm{P}}\|_F / \|\bm{P}\|_F$. Results are averaged over $20$ independent Monte Carlo replicates.}
    \label{fig:simu_accuracy}
\end{figure}


\subsection{Effect of the number of sources}\label{sec:simu_vary_m}

Our theoretical results suggest that, under suitable assumptions, the advantage of GSMMI over CMMI in alignment accuracy increases with the number of sources $m$. To empirically verify this, we conduct an experiment in which we vary $m$.

We generate the population matrix $\bm{P}$ as described in Section~\ref{sec:simu_error}, with $N = 200$, rank $d = 3$, $\bm{\Lambda} = \diag(\lambda, \tfrac{3}{4}\lambda, \tfrac{1}{2}\lambda)$, signal strength $\lambda = 4$, and noise level $\sigma = 0.1$. We fix the submatrix-size ratio at $\alpha = 0.3$, so that each block contains $n = \lfloor \alpha N \rfloor = 60$ entities, and vary the number of sources over $m \in \{5, 8, 12, 20, 30, 50, 80\}$. For each $m$, we apply both GSMMI and CMMI and report the matrix correlation, the relative Frobenius error, and the running time, averaged over $100$ Monte Carlo replicates.

Figure~\ref{fig:simu_varym} shows that CMMI's correlation plateaus at approximately $0.45$ and decreases slightly for larger $m$.
GSMMI's correlation rises from $0.38$ at $m = 5$ to $0.91$ at $m = 80$, so the accuracy gap widens as the number of sources grows.

\begin{figure}[htbp]
    \centering
    \includegraphics[width=0.8\textwidth]{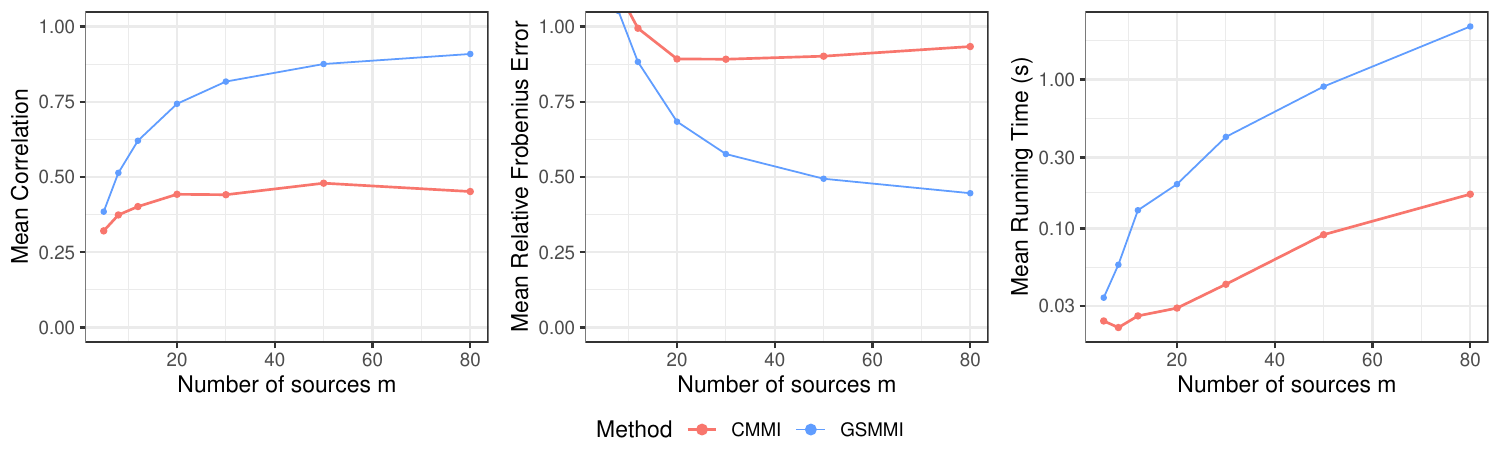}
    \caption{\footnotesize Performance of GSMMI and CMMI as the number of sources $m$ varies over $\{5, 8, 12, 20, 30, 50, 80\}$, with the block-size ratio fixed at $\alpha = 0.3$, $N = 200$, $d = 3$, $\lambda = 4$, and $\sigma = 0.1$. Left: matrix correlation $\langle \bm{P}, \hat{\bm{P}} \rangle_F / (\|\bm{P}\|_F \|\hat{\bm{P}}\|_F)$. Middle: relative Frobenius error $\|\bm{P} - \hat{\bm{P}}\|_F / \|\bm{P}\|_F$. Right: running time in seconds (log scale). Results are averaged over $100$ independent Monte Carlo replicates.}
    \label{fig:simu_varym}
\end{figure}

The running time of GSMMI grows faster than that of CMMI as $m$ increases.
One possible explanation is that a less accurate CMMI initialization requires more iterations to improve the alignment.
At $m = 80$, GSMMI takes under $2.5$ seconds, about $13$ times the running time of CMMI, while roughly doubling its correlation (from $0.45$ to $0.91$).

\subsection{Advantages over existing imputation methods for blockwise missing data}
\label{sec:simu_block}

We now consider a setting that mimics realistic multi-source data integration involving blockwise missing patterns. The experiments in \citet{zheng2026chain} demonstrate the superiority of CMMI over several existing matrix completion methods for blockwise missing data, but they focus on the case where the overlap graph $\mathcal{G}$ of the observed blocks (the graph encoding overlaps among observed blocks; see Section~\ref{sec:rev_CMMI} for its precise definition) forms a simple chain; that is, the observed blocks overlap only sequentially, with no more complex overlap structure (see Section~4.2 of \citet{zheng2026chain} for details). In practice, however, observed blocks from multiple sources need not follow a purely chain-type structure and may instead exhibit a somewhat richer overlap pattern. To mimic this situation, we simulate observed blocks whose overlap graph contains some cycles, and compare GSMMI, CMMI, and several existing matrix completion methods, including generalized spectral regularization (GSR) \citep{mazumder2010spectral}, fast alternating least squares (FALS) \citep{hastie2015matrix}, singular value thresholding (SVT) \citep{cai2010singular}, and universal singular value thresholding (USVT) \citep{chatterjee2015matrix}.

We construct the observed submatrices as follows. The $m$ blocks of size $n$ are arranged in a ring: adjacent blocks share $s$ entities, so that the last $s$ entities of block $i$ coincide with the first $s$ entities of block $i+1$, and the ring is closed by identifying the last $s$ entities of block $m$ with the first $s$ entities of block $1$. This ring structure yields a total of $N = mn - ms$ entities. On top of the ring, to enrich the overlap structure, we add $m$ random cross-block shortcuts: for each shortcut, we pick a block at random and append $s$ of its entities to another randomly chosen non-adjacent block. The resulting design retains a clear block structure while endowing the overlap graph with a moderately rich cyclic connectivity. The left panel of Figure~\ref{fig:simu_block} illustrates the resulting missing pattern for $m = 5$: the ring links produce the block-diagonal structure together with the off-diagonal blocks connecting the first and last sources, while the random shortcuts appear as additional off-diagonal overlaps between non-adjacent blocks. Note that we take $s$ much smaller than $n$ (specifically $s = 0.1n$), to examine whether GSMMI and CMMI can succeed even when the number of entities shared between overlapping blocks is small. Moreover, notice that we add only $m$ random shortcuts, a small number compared with the $m(m-3)/2$ non-adjacent block pairs, especially when $m$ is large, so the enrichment of the overlap structure is deliberately modest; our goal is to explore whether GSMMI can improve over CMMI using only a moderately rich overlap structure.

We generate the population matrix $\bm{P} = \bm{U}\bm{\Lambda}\bm{U}^\top$ of size $N$ and rank $d = 3$, where $\bm{U} \in \mathcal{O}_{N \times d}$ is randomly generated, $\bm{\Lambda} = \diag(\lambda, \tfrac{3}{4}\lambda, \tfrac{1}{2}\lambda)$, and $\lambda = N$. For each block $i$ we observe the noisy submatrix $\bm{A}^{(i)} = \bm{P}_{\mathcal{U}_i, \mathcal{U}_i} + \bm{N}^{(i)}$, where $\bm{N}^{(i)}$ is a symmetric Gaussian noise matrix generated as in Section~\ref{sec:simu_error} with noise level $\sigma = 1$. We fix the overlap size $s = 0.1 n$, corresponding to an overlap ratio of about $0.1$, and vary the number of sources $m \in \{5, 10, 15, 20, 25\}$, choosing $n$ so that the total size $N$ remains approximately $1000$ across all settings. Thus, as $m$ increases, the fraction of observed entries decreases, making the recovery task progressively more challenging.

The blockwise missing pattern leaves a large fraction of the entries of $\bm{P}$ never observed by any source, and we focus on the accuracy of recovering these unobserved entries. Specifically, we report the {unobserved relative error}, defined as $\|(\hat{\bm{P}} - \bm{P}) \circ \mathbb{I}_{\Omega^c}\|_F / \|\bm{P} \circ \mathbb{I}_{\Omega^c}\|_F$, where $\Omega^c$ denotes the set of entries not observed by any submatrix, together with the running time. Figure~\ref{fig:simu_block} reports both quantities, averaged over $100$ Monte Carlo replicates.

\begin{figure}[htbp]
    \centering
    \includegraphics[width=0.32\textwidth]{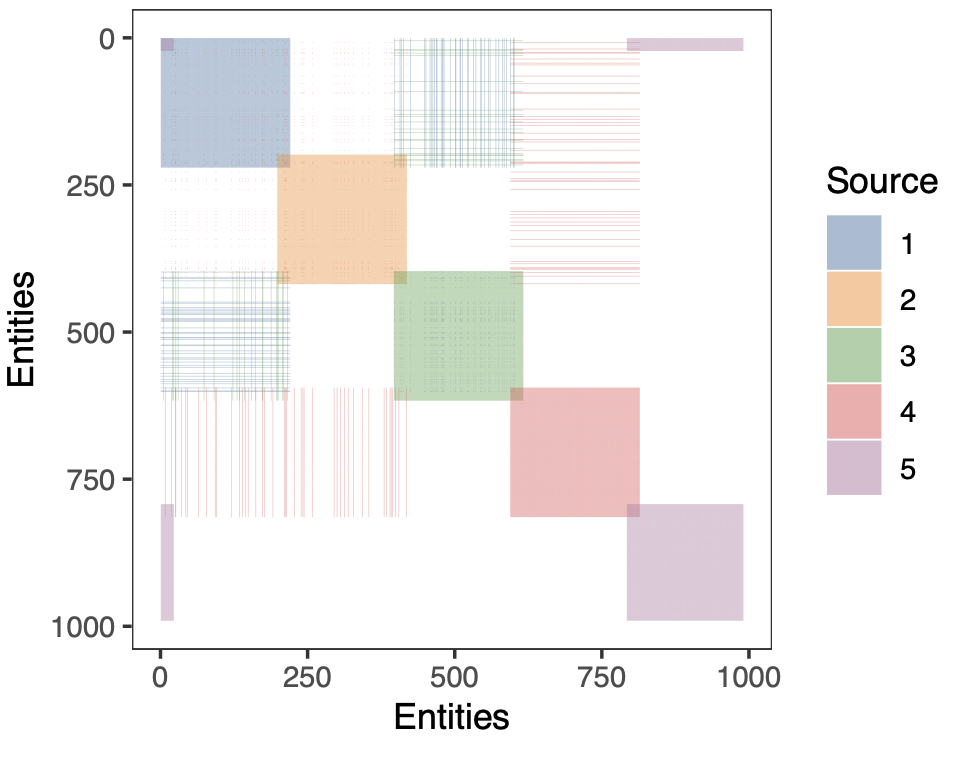}
    \includegraphics[width=0.55\textwidth]{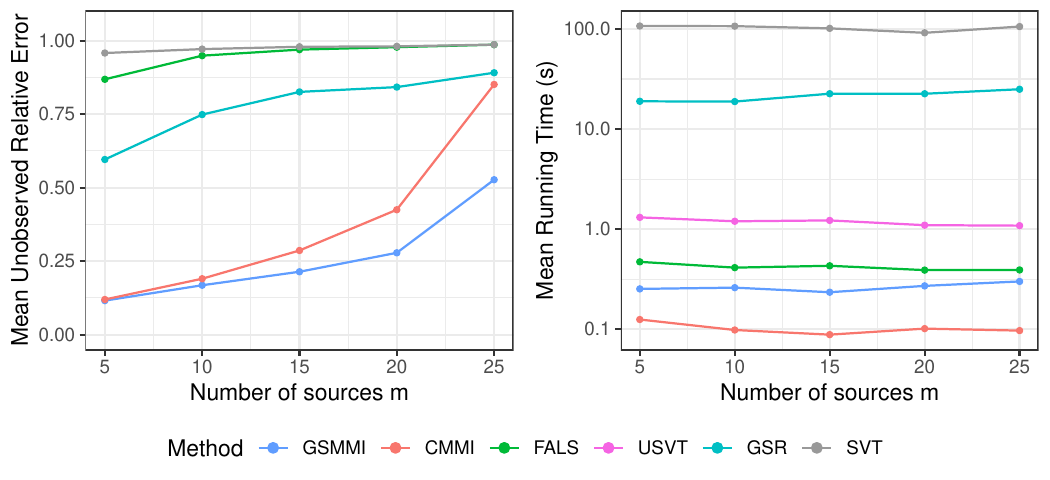}
    \caption{\footnotesize Performance of GSMMI, CMMI, and four matrix-completion baselines (FALS, USVT, GSR, SVT) under the ring-structured design with shortcuts described in Section~\ref{sec:simu_block}, as the number of sources $m$ varies over $\{5, 10, 15, 20, 25\}$, with overlap size $s = 0.1 n$, total size $N \approx 1000$, $d = 3$, and $\sigma = 1$. Left: the observation pattern, shown for $m = 5$ as an example, where each color corresponds to one source and overlapping regions are shown by blended colors. Middle: unobserved relative Frobenius error, $\|(\hat{\bm{P}} - \bm{P}) \circ \mathbb{I}_{\Omega^c}\|_F / \|\bm{P} \circ \mathbb{I}_{\Omega^c}\|_F$; the $y$-axis is capped at $1.05$ for readability, and USVT (with errors up to $2.2$) exceeds this range. Right: running time in seconds (log scale). Results are averaged over $100$ independent Monte Carlo replicates.}
    \label{fig:simu_block}
\end{figure}

Consistent with the CMMI results in \citet{zheng2026chain}, Figure~\ref{fig:simu_block} shows that GSMMI and CMMI outperform all four matrix-completion baselines in this blockwise setting.
GSMMI is more accurate than CMMI throughout, with a larger gap as $m$ grows: their errors are nearly identical at $m = 5$ ($0.116$ versus $0.120$), whereas at $m = 25$ GSMMI attains an error of $0.527$ against $0.851$ for CMMI, a relative reduction of about $38\%$.
This is consistent with the benefit of using additional overlaps discussed in Section~\ref{sec:simu_vary_m}.
Both methods have running times comparable to FALS and orders of magnitude smaller than GSR and SVT.

\subsection{Comparison with the theoretical detection threshold}\label{sec:simu_threshold}

We next compare the empirical behavior of GSMMI and CMMI with the detection lower bound for the matrix integration model established in Section~\ref{sec:low_mat}. That result shows that, in the rank-one case and over a sequence of problems of growing dimension with parameters converging appropriately, no method can (strongly) detect the signal when $\tilde\lambda < \tilde\lambda_*$, where $\tilde\lambda_* \colonequals \sqrt{N / \lambda_1(\widetilde{\bm B})}$ with $\widetilde{B}_{ij} = |\mathcal{U}_i \cap \mathcal{U}_j|$ for $i, j \in [m]$, and the problem is expected to become solvable non-trivially once $\tilde\lambda > \tilde\lambda_*$. We therefore expect the correlation between $\hat{\bm{P}}$ and $\bm{P}$ to be zero to the left of $\tilde\lambda_*$ and to gradually become positive to its right.

Following the setting of Section~\ref{sec:low_mat}, we take $d = 1$ and generate $\bm{P} = \frac{\tilde\lambda}{\sqrt{N}} \bm{x}^{(0)} \bm{x}^{(0)\top}$, where $\bm{x}^{(0)} \in \{\pm 1\}^N$ has i.i.d.\ Rademacher entries, so that $\lambda_1(\bm{P}) / \sqrt{N} = \tilde\lambda$. %
We set $N = 1000$ and $m = 20$, draw each $\mathcal{U}_i$ independently as a uniformly random subset of $[N]$ of size $n = \lfloor \alpha N \rfloor$, and observe $\bm{A}^{(i)} = \bm{P}_{\mathcal{U}_i, \mathcal{U}_i} + \bm{N}^{(i)}$, where $\bm{N}^{(i)}$ is generated as in Section~\ref{sec:simu_error} with $\sigma = 1$. We vary $\tilde\lambda \in \{0.1, 0.2, \dots, 4.0\}$ and $\alpha \in \{0.1, 0.15, 0.2, 0.3, 0.5, 1\}$, compute $\tilde\lambda_*$ from the realized index sets in each replicate, and report the matrix correlation $\langle \bm{P}, \hat{\bm{P}} \rangle_F / (\|\bm{P}\|_F \|\hat{\bm{P}}\|_F)$, averaged over $50$ Monte Carlo replicates.

Figure~\ref{fig:simu_threshold} shows the results, with $\tilde\lambda_*$ marked by a dashed vertical line. The behavior of both methods is consistent with the prediction above. To the left of $\tilde\lambda_*$, the correlations of GSMMI and CMMI are essentially zero for every $\alpha$, in agreement with the lower bound of Section~\ref{sec:low_mat}. To the right of $\tilde\lambda_*$, the correlations gradually become positive and increase with $\tilde\lambda$. Moreover, for every value of $\alpha$, the correlation of GSMMI becomes positive at a smaller $\tilde\lambda$ than that of CMMI and is at least as high throughout, with the two curves coinciding once $\tilde\lambda$ is large enough for both methods to align the blocks perfectly.

\begin{figure}[htbp]
    \centering
    \includegraphics[width=\textwidth]{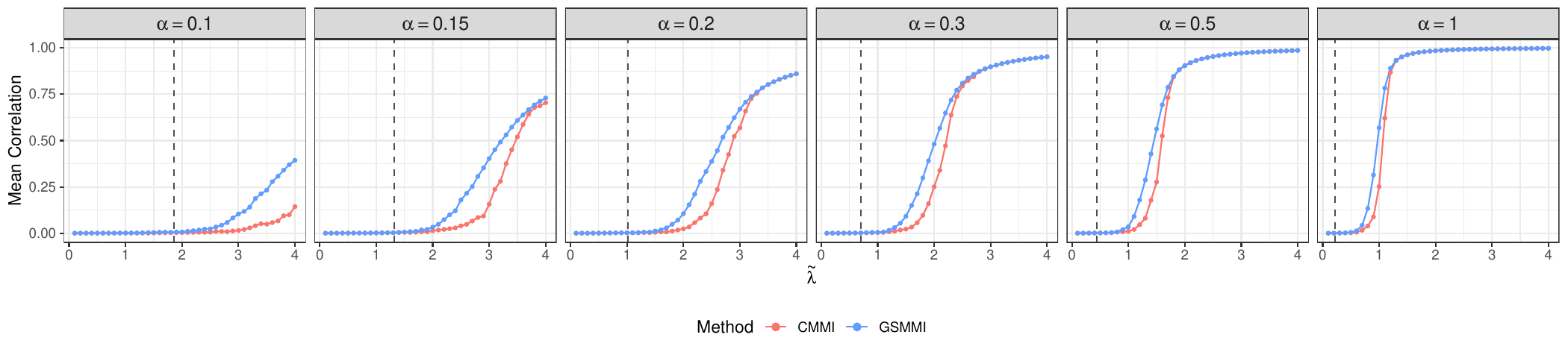}
    \caption{\footnotesize Correlation $\langle \bm{P}, \hat{\bm{P}} \rangle_F / (\|\bm{P}\|_F \|\hat{\bm{P}}\|_F)$ of GSMMI and CMMI in the rank-one setting of Section~\ref{sec:simu_threshold}, as $\tilde\lambda = \lambda_1(\bm{P})/\sqrt{N}$ varies over $\{0.1, 0.2, \dots, 4.0\}$ and $\alpha \in \{0.1, 0.15, 0.2, 0.3, 0.5, 1\}$, with $N = 1000$, $m = 20$, and $\sigma = 1$. The dashed vertical line marks the detection threshold $\tilde\lambda_* = \sqrt{N / \lambda_1(\widetilde{\bm B})}$ from Section~\ref{sec:low_mat}. Results are averaged over $50$ independent Monte Carlo replicates.}
    \label{fig:simu_threshold}
\end{figure}

\section{Experiments on realistic data}
\label{sec:real}

In this section, we use two real data examples to demonstrate that GSMMI can further improve upon CMMI in practice. We first consider the MNIST database of handwritten digit images (Section~\ref{sec:MNIST}), and then the MEDLINE co-occurrence database (Section~\ref{sec:MEDLINE}).

\subsection{MNIST handwritten digits}
\label{sec:MNIST}
We illustrate the advantage of GSMMI over CMMI and several existing matrix-completion methods on the MNIST database of handwritten digit images \citep{lecun1998mnist}. The MNIST database consists of grayscale images of handwritten digits, each of size $28 \times 28$ pixels and viewed as a vector in $\mathbb{R}^{784}$. We consider a population of $N$ entities, each assigned a fixed label drawn uniformly from the digits $\{3, 5, 8\}$, which are relatively difficult to distinguish. The $m$ sources are arranged in the same ring-structured block design with random shortcuts as in Section~\ref{sec:simu_block}. For each number of observed blocks $m$, we choose the block size $n$ so that the total number of entities $N$ remains approximately $3000$ across all settings. For each entity in a given source, we sample an image uniformly at random from the MNIST images sharing its label and normalize it to unit norm. As a result, the same entity is represented by different images of the same digit across sources, which serves as a source of noise. We then form the observed block $\bm{A}^{(i)} = \bm{M}^{(i)} \bm{M}^{(i)\top} \in \mathbb{R}^{n \times n}$, where $\bm{M}^{(i)} \in \mathbb{R}^{n \times 784}$ stacks the $n$ unit-normalized images for the entities in source $i$.

Given the collection $\{\bm{A}^{(i)}\}_{i=1}^{m}$, we compare how accurately different methods can cluster the entities by their underlying digit. We apply CMMI and GSMMI to recover the latent positions of all entities, using embedding dimension $d = 36$, chosen by applying the dimensionality selection procedure of \cite{zhu2006automatic} to the MNIST images. For comparison, we also include three of the matrix-completion baselines from Section~\ref{sec:simu_block}, namely GSR, FALS, and USVT; SVT is excluded here because it is prohibitively slow on matrices of this larger size.

We evaluate the clustering accuracy of each method as follows. CMMI and GSMMI directly return $d$-dimensional latent positions, which we cluster into three groups using $K$-means. The matrix-completion baselines instead return a completed matrix: for these, we first assemble a single partially observed $N \times N$ similarity matrix from the $m$ source blocks, in which unobserved pairs are left missing and entries observed by multiple sources are averaged; we then apply each completion method to fill in the full matrix, extract its leading $d$-dimensional eigen-embedding, and cluster the embedding with $K$-means. We measure accuracy against the true digit labels using the Adjusted Rand Index (ARI), which ranges from $-1$ to $1$, with higher values indicating closer agreement. We vary the number of sources over $m \in \{3, 4, 5, 6, 7, 8, 9\}$ and report results averaged over $25$ Monte Carlo replicates.

In Figure~\ref{fig:realdata_MNIST}, GSMMI and CMMI outperform all three completion baselines, whose ARIs remain below $0.11$ throughout.
GSMMI is more accurate than CMMI at every value of $m$, with a widening gap.
CMMI's clustering accuracy peaks at about $0.41$ for $m = 4$ and falls to $0.13$ at $m = 9$, whereas GSMMI stays in the range $0.40$--$0.49$.
The two methods have comparable running times, both orders of magnitude smaller than those of the completion baselines.

\begin{figure}[htbp]
    \centering
    \includegraphics[width=0.55\textwidth]{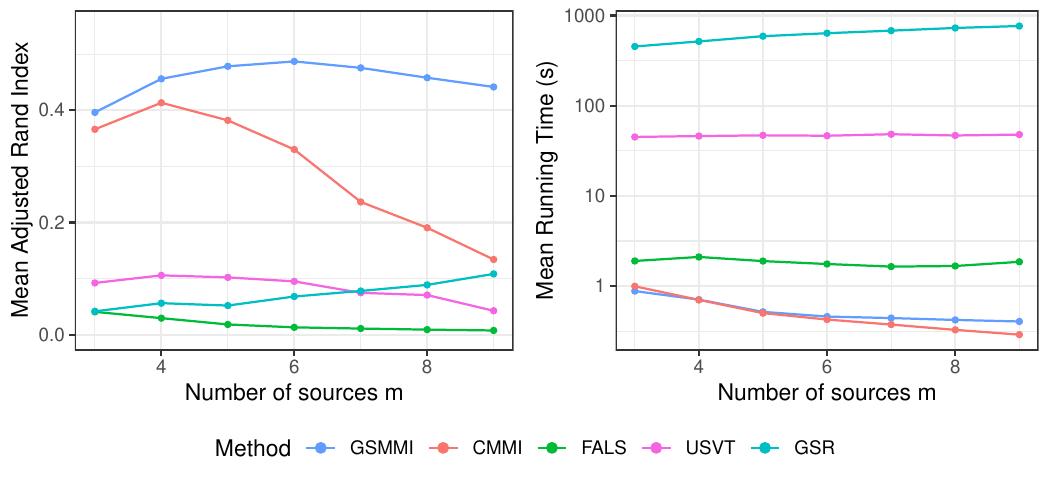}
    \caption{
    \footnotesize
    Performance of GSMMI, CMMI, and three matrix-completion baselines (FALS, USVT, GSR) on the MNIST experiment under the ring-structured design with shortcuts described in Section~\ref{sec:simu_block}, as the number of sources $m$ varies over $\{3, 4, 5, 6, 7, 8, 9\}$, with entities labeled by the digits $\{3, 5, 8\}$, overlap size $s = 0.1 n$, total size $N \approx 3000$, and embedding dimension $d = 36$. Left: Adjusted Rand Index for clustering against the true digit labels. Right: running time in seconds (log scale). Results are averaged over $25$ independent Monte Carlo replicates.
    }
    \label{fig:realdata_MNIST}
\end{figure}

\subsection{MEDLINE co-occurrences}
\label{sec:MEDLINE}

Our second example uses the MEDLINE co-occurrence database \citep{medline1a}, following the same setting as in Section~B.1 of the supplementary material of \citet{zheng2026chain}. In that work, CMMI was compared against GSR, FALS, SVT, and USVT, and was shown to substantially outperform all of them in both accuracy and computational efficiency. We therefore omit these baselines and focus our comparison on CMMI, since establishing that GSMMI improves upon CMMI is sufficient to demonstrate its advantage over these competing methods as well.

The MEDLINE co-occurrence database summarizes the MeSH descriptors that co-occur in MEDLINE/PubMed citations over a period of years. A standard approach for handling such data is to transform the (normalized) co-occurrence counts into pointwise mutual information (PMI), an association measure widely used in natural language processing. Specifically, the PMI between two concepts $x$ and $y$ is defined as
$\mathrm{PMI}(x,y)=\log\tfrac{\mathbb{P}(x,y)}{\mathbb{P}(x)\mathbb{P}(y)}$,
where $\mathbb{P}(x)$ and $\mathbb{P}(y)$ are the marginal occurrence probabilities of $x$ and $y$, and $\mathbb{P}(x,y)$ is their joint co-occurrence probability.

If we partition the citations by year, each year tends to feature a somewhat different set of frequently occurring clinical concepts, and the PMIs computed among the high-frequency concepts within a given year tend to be more accurate than those involving rarely occurring concepts. Motivated by this observation, we extract the PMIs for the top-frequency concepts in each year and integrate them using CMMI or GSMMI, aiming to recover more accurate co-occurrence relationships between the clinical concepts than the PMIs computed directly from data aggregated across all the years.

Concretely, we consider MEDLINE co-occurrence data from the years 1993 to 2022. For each year, we extract a PMI submatrix based on the co-occurrences of the top $1000$ most frequent clinical concepts. Given a number of observed years $K$, equivalently the number of observed PMI submatrices, our goal is to integrate these submatrices to recover the full PMI matrix over the union of concepts. In this experiment, we always select the most temporally distant years for the integration task. For example, for $K = 2$ we integrate the PMI submatrices corresponding to the years $1993$ and $2022$, which together involve $N = 1540$ unique clinical concepts, and we aim to recover the unobserved entries (approximately $25\%$) in the resulting $1540 \times 1540$ PMI matrix.
We determine the embedding dimension $d$ by applying the automatic dimensionality selection procedure of \cite{zhu2006automatic} to each observed submatrix and taking the largest resulting value to retain sufficient information. For example, when $K = 2$, the procedure yields dimensions $12$ and $16$ for the two submatrices, so we set $d = 16$. We then apply both CMMI and GSMMI to recover the full PMI matrix.

To evaluate the recovery quality, we use the pre-trained clinical concept embeddings from \cite{beam2020clinical}, which are learned from massive multimodal medical data. Given pre-trained embedding vectors $\bm{v}_1, \dots, \bm{v}_N$ for the $N$ clinical concepts, we construct a similarity matrix $\bm{P} \in \mathbb{R}^{N \times N}$ whose entry $\bm{P}_{ij}$ is the cosine similarity between $\bm{v}_i$ and $\bm{v}_j$. We then measure the similarity between the estimated PMIs of the unobserved entries and the corresponding entries in $\bm{P}$ using Spearman's rank correlation $\rho$, which takes values in $[-1, 1]$, with larger values indicating a stronger monotone relationship.
 
We vary $K$ from $2$ to $15$ and report the results in Figure~\ref{fig:realdata_MEDLINE_dense}.
The dashed black line shows the baseline obtained by computing PMIs from data aggregated across the selected years.
Both integration methods achieve higher rank correlations with the reference similarities: at $K = 2$, the baseline correlation is $0.048$, compared with $0.273$ for both CMMI and GSMMI.
GSMMI matches or exceeds CMMI across all values of $K$.
At $K = 2$, the methods coincide because the single pair of overlapping blocks leaves no additional overlaps for GSMMI to use.
For larger $K$, using all overlaps gives a consistent improvement, shown in the middle panel, at a running time close to that of CMMI.

\begin{figure}[htbp]
\centering
\subfigure{\includegraphics[width=0.8\textwidth]{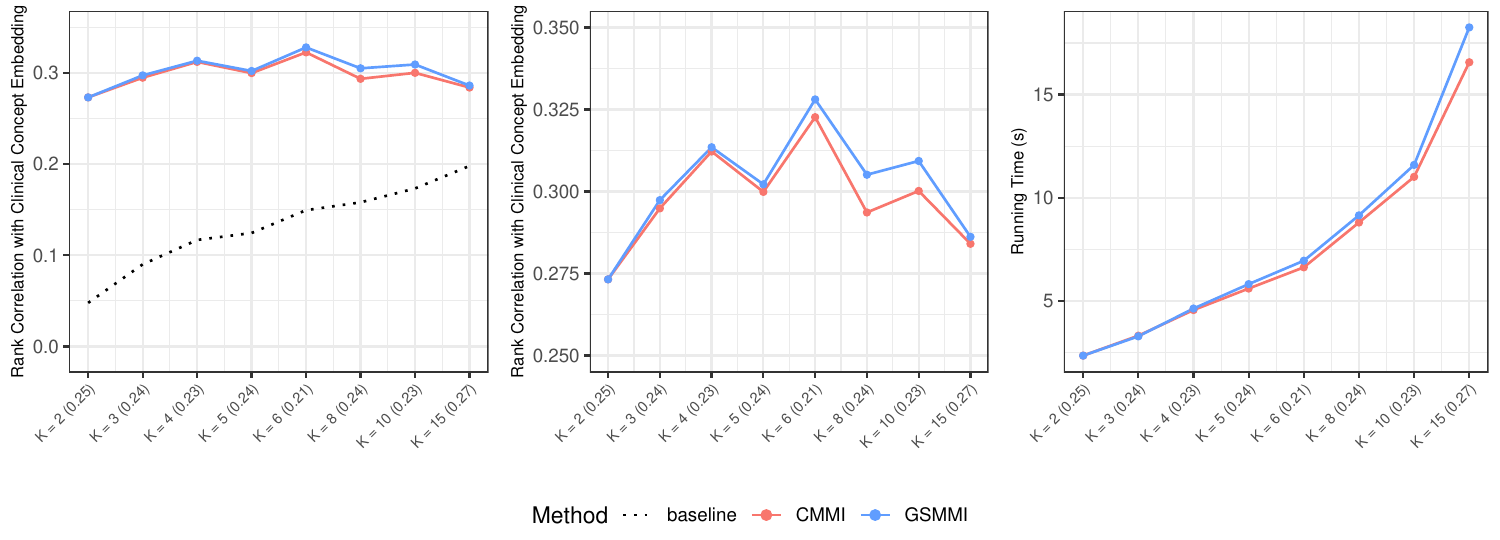}}
\caption{
\footnotesize
Performance of GSMMI and CMMI in the MEDLINE co-occurrence experiment, integrating per-year PMI submatrices as $K$ varies from $2$ to $15$.
On the x-axis, each $K$ is followed (in parentheses) by the proportion of unobserved entries to be recovered.
Left: Spearman's rank correlation between the recovered PMIs and the clinical concept embeddings of \cite{beam2020clinical}, for the baseline (PMIs computed directly from the co-occurrence data aggregated across the selected years), CMMI, and GSMMI.
Middle: the same correlation, zoomed in to compare CMMI and GSMMI.
Right: running time in seconds.
}
\label{fig:realdata_MEDLINE_dense}
\end{figure}

\section{Extensions to indefinite and rectangular matrices}
\label{sec:extension}

The GSMMI algorithm presented in Section~\ref{sec:method} assumes that the population matrix $\mpp$ is positive semidefinite, a setting we used to introduce the main idea of the method and to develop its analysis. In many real-world applications, however, the matrix of interest can instead be symmetric but indefinite, or asymmetric and possibly rectangular. Symmetric indefinite matrices arise, for example, in generalized random dot product graph models \citep{rubin2022statistical} and in similarity or kernel matrices induced by indefinite kernels, whereas asymmetric or rectangular matrices are central to applications such as genomic data integration \citep{letunic2004smart}, single-cell data integration \citep{stuart2019comprehensive}, and recommendation systems. Integrating multiple such matrices from different sources naturally leads to multiple matrix integration problems in these non-positive-semidefinite settings.

In this section, we describe how GSMMI can be extended to the cases of symmetric indefinite matrices and asymmetric or rectangular matrices. As in the positive semidefinite setting, the main idea is that GSMMI replaces the spanning-tree alignment of CMMI with a global synchronization over all pairwise overlaps, while providing modifications appropriate to each case. These modifications are also important for the practical feasibility and effectiveness of the algorithm. For example, for symmetric indefinite matrices, the alignment matrices should be indefinite orthogonal matrices, but no direct fast optimizer over the indefinite orthogonal group is available; we therefore approximate the per-iteration solution by relaxing the problem into unconstrained least-squares problems and then projecting the result back onto the indefinite orthogonal group. As another example, for asymmetric or rectangular matrices, the alignment matrices are general invertible matrices, and without care the iterations may accumulate scale drift; we therefore rescale the estimates during the iterations to keep them well-conditioned and improve numerical stability.

\subsection{GSMMI for symmetric indefinite matrices}
\label{sec:extension_indef}

Suppose $\mpp \in \mathbb{R}^{N \times N}$ is a symmetric indefinite low-rank matrix. Let $d_+$ and $d_-$ denote the number of positive and negative eigenvalues of $\mpp$, and set $d = d_+ + d_- \ll N$. We denote the non-zero eigenvalues of $\mpp$ by $\lambda_1(\mpp) \geq \dots \geq \lambda_{d_+}(\mpp) > 0 > \lambda_{N-d_-+1}(\mpp) \geq \dots \geq \lambda_N(\mpp)$. Let $\mLambda_+  \colonequals  \diag(\lambda_1(\mpp), \dots, \lambda_{d_+}(\mpp))$, $\mLambda_-  \colonequals  \diag(\lambda_{N-d_-+1}(\mpp), \dots, \lambda_N(\mpp))$, and let the orthonormal columns of $\muu_+ \in \mathbb{R}^{N \times d_+}$ and $\muu_- \in \mathbb{R}^{N \times d_-}$ constitute the corresponding eigenvectors. The eigendecomposition of $\mpp$ is then $\muu \mLambda \muu^\top$ with $\mLambda  \colonequals  \diag(\mLambda_+, \mLambda_-)$ and $\muu  \colonequals  [\muu_+|\muu_-]$. The latent positions are given by $\mx = \muu |\mLambda|^{1/2}$, so that $\mpp = \mx \mi_{d_+, d_-} \mx^\top$, where $\mi_{d_+, d_-} = \diag(\mi_{d_+}, -\mi_{d_-})$. For each source $i$, we have $\mpp^{(i)} = \mpp_{\mathcal{U}_i, \mathcal{U}_i} = \mx_{\mathcal{U}_i} \mi_{d_+, d_-} \mx_{\mathcal{U}_i}^\top$.

From each observed submatrix $\ma^{(i)}$, we obtain an estimated latent position matrix $\hat{\mx}^{(i)} \in \mathbb{R}^{n_i \times d}$ such that $\ma^{(i)} \approx \hat{\mx}^{(i)} \mi_{d_+, d_-} \hat{\mx}^{(i)\top}$. As in the positive semidefinite case (see Remark~\ref{rem:individual est}), various methods can be used to obtain such an estimate; for example, one may set $\hat{\mx}^{(i)} = \hat{\muu}^{(i)} |\hat{\mLambda}^{(i)}|^{1/2}$, where $\hat{\mLambda}^{(i)}$ collects the $d_+$ largest positive and $d_-$ largest (in magnitude) negative eigenvalues of $\ma^{(i)}$, and $\hat{\muu}^{(i)}$ contains the corresponding eigenvectors.

Each $\hat{\mx}^{(i)}$ estimates $\mx_{\mathcal{U}_i}$ only up to a transformation, which is now an element of the indefinite orthogonal group $\mathcal{O}_{d_+, d_-}  \colonequals  \{\mo \in \mathbb{R}^{d \times d} : \mo \mi_{d_+, d_-} \mo^\top = \mi_{d_+, d_-}\}$. To exploit all pairwise overlaps, we again formulate the alignment as a synchronization problem, seeking transformation matrices $\{\hat{\mw}^{(i)}\}$ that jointly align all estimated latent positions across all overlapping entities:
\[
\{\hat{\mw}^{(i)}\}_{i=1}^m = \argmin_{\mo^{(1)}, \ldots, \mo^{(m)} \in \mathcal{O}_{d_+, d_-}}
\sum_{i < j : \mathcal{U}_i \cap \mathcal{U}_j \neq \varnothing}
\pi_{i,j} \, \|\hat{\mx}^{(i)}_{\langle \mathcal{U}_i \cap \mathcal{U}_j \rangle} \mo^{(i)} - \hat{\mx}^{(j)}_{\langle \mathcal{U}_i \cap \mathcal{U}_j \rangle} \mo^{(j)}\|_F^2.
\]
We again fix $\hat{\mw}^{(1)} = \mi_d$ as the reference and iteratively update each $\hat{\mw}^{(i)}$ for $i = 2, \ldots, m$, initializing from the CMMI solution. For a fixed set of transformation matrices $\{\hat{\mw}^{(j)}\}_{j \neq i}$, the optimal update for $\hat{\mw}^{(i)}$ solves
\begin{equation}\label{eq:update_indef_constrained}
\hat{\mw}^{(i)} = \argmin_{\mo \in \mathcal{O}_{d_+, d_-}} \sum_{j \neq i : \mathcal{U}_i \cap \mathcal{U}_j \neq \varnothing} \pi_{i,j} \, \|\hat{\mx}^{(i)}_{\langle \mathcal{U}_i \cap \mathcal{U}_j \rangle} \mo - \hat{\mx}^{(j)}_{\langle \mathcal{U}_i \cap \mathcal{U}_j \rangle} \hat{\mw}^{(j)}\|_F^2.
\end{equation}
Unlike the orthogonal case, this problem has no closed-form solution due to the indefinite orthogonality constraint. We therefore approximate it by a relaxation that drops the constraint and then projects the result back onto $\mathcal{O}_{d_+,d_-}$.
We also note that right-multiplying the alignment constraint $\hat{\mx}^{(i)}_{\langle \mathcal{U}_i \cap \mathcal{U}_j \rangle} \mo \approx \hat{\mx}^{(j)}_{\langle \mathcal{U}_i \cap \mathcal{U}_j \rangle} \hat{\mw}^{(j)}$ by $\mi_{d_+,d_-}\mo^\top$ and using $\mo\,\mi_{d_+,d_-}\,\mo^\top = \mi_{d_+,d_-}$, this constraint can equivalently be written in terms of the dual positions $\hat{\mx}^{(i)}_{\langle \mathcal{U}_i \cap \mathcal{U}_j \rangle}\mi_{d_+,d_-}$ as $\hat{\mx}^{(i)}_{\langle \mathcal{U}_i \cap \mathcal{U}_j \rangle}\mi_{d_+,d_-} \approx \hat{\mx}^{(j)}_{\langle \mathcal{U}_i \cap \mathcal{U}_j \rangle}\hat{\mw}^{(j)}\mi_{d_+,d_-}\,\mo^\top$. Relaxing the constraint in each of these two forms yields two unconstrained least-squares problems
\begin{align}
\hat{\mw}_L^{(i)} &= \argmin_{\mo \in \mathbb{R}^{d \times d}} \sum_{j \neq i : \mathcal{U}_i \cap \mathcal{U}_j \neq \varnothing} \pi_{i,j} \, \|\hat{\mx}^{(i)}_{\langle \mathcal{U}_i \cap \mathcal{U}_j \rangle} \mo - \hat{\mx}^{(j)}_{\langle \mathcal{U}_i \cap \mathcal{U}_j \rangle} \hat{\mw}^{(j)}\|_F^2, \label{eq:update_indef_L}\\
\hat{\mw}_R^{(i)} &= \argmin_{\mo \in \mathbb{R}^{d \times d}} \sum_{j \neq i : \mathcal{U}_i \cap \mathcal{U}_j \neq \varnothing} \pi_{i,j} \, \|\hat{\mx}^{(i)}_{\langle \mathcal{U}_i \cap \mathcal{U}_j \rangle}\mi_{d_+,d_-}- \hat{\mx}^{(j)}_{\langle \mathcal{U}_i \cap \mathcal{U}_j \rangle}\hat{\mw}^{(j)}\mi_{d_+,d_-}  \mo^\top \|_F^2.\label{eq:update_indef_R}
\end{align}
The solution to \eqref{eq:update_indef_L} is
\[
\hat{\mw}_L^{(i)} = \Biggl( \sum_{j \neq i : \mathcal{U}_i \cap \mathcal{U}_j \neq \varnothing} \pi_{i,j} (\hat{\mx}^{(i)}_{\langle \mathcal{U}_i \cap \mathcal{U}_j \rangle})^\top \hat{\mx}^{(i)}_{\langle \mathcal{U}_i \cap \mathcal{U}_j \rangle} \Biggr)^{+}\Biggl( \sum_{j \neq i : \mathcal{U}_i \cap \mathcal{U}_j \neq \varnothing} \pi_{i,j} (\hat{\mx}^{(i)}_{\langle \mathcal{U}_i \cap \mathcal{U}_j \rangle})^\top \hat{\mx}^{(j)}_{\langle \mathcal{U}_i \cap \mathcal{U}_j \rangle} \hat{\mw}^{(j)} \Biggr),
\]
where $(\cdot)^{+}$ denotes the Moore-Penrose pseudoinverse. Similarly, the solution to \eqref{eq:update_indef_R} is
\begin{equation*}
\begin{aligned}
\hat{\mw}_R^{(i)} = &\Biggl( \sum_{j \neq i : \mathcal{U}_i \cap \mathcal{U}_j \neq \varnothing} \pi_{i,j} \, \mi_{d_+,d_-} (\hat{\mx}^{(i)}_{\langle \mathcal{U}_i \cap \mathcal{U}_j \rangle})^\top \hat{\mx}^{(j)}_{\langle \mathcal{U}_i \cap \mathcal{U}_j \rangle} \hat{\mw}^{(j)} \mi_{d_+,d_-} \Biggr)\\
& \Biggl( \sum_{j \neq i : \mathcal{U}_i \cap \mathcal{U}_j \neq \varnothing} \pi_{i,j} \, \mi_{d_+,d_-} (\hat{\mw}^{(j)})^\top (\hat{\mx}^{(j)}_{\langle \mathcal{U}_i \cap \mathcal{U}_j \rangle})^\top \hat{\mx}^{(j)}_{\langle \mathcal{U}_i \cap \mathcal{U}_j \rangle} \hat{\mw}^{(j)} \mi_{d_+,d_-} \Biggr)^{+}.
\end{aligned}
\end{equation*}
We combine them by setting $\hat{\mw}^{(i)} = \frac{1}{2}(\hat{\mw}_L^{(i)} + \hat{\mw}_R^{(i)})$.
Since the relaxation drops the indefinite orthogonality constraint, we project the estimate back onto $\mathcal{O}_{d_+, d_-}$ before the next update, which improves numerical stability and accuracy.
The projection is carried out by taking a few gradient-type steps that iteratively reduce the constraint violation $\|\mo\, \mi_{d_+,d_-}\, \mo^\top - \mi_{d_+,d_-}\|_F$, starting from $\mo = \hat{\mw}^{(i)}$.
Only a few iterations are needed in practice.
We cycle through the updates for $i = 2, \ldots, m$ until convergence.
After obtaining $\{\hat{\mw}^{(i)}\}$, we aggregate the aligned latent positions to form $\hat{\mx}$, and estimate $\hat{\mpp} = \hat{\mx} \mi_{d_+, d_-} \hat{\mx}^\top$. The complete procedure is summarized in Algorithm~\ref{alg:gsmmi_indef}.

\begin{algorithm}[htbp]
\caption{GSMMI for symmetric indefinite matrices}
\label{alg:gsmmi_indef}
\begingroup
\small
\setlength{\baselineskip}{10pt}
\setlength{\itemsep}{0pt}
\setlength{\parskip}{0pt}
\setlength{\abovedisplayskip}{4pt}
\setlength{\belowdisplayskip}{4pt}
\begin{algorithmic}
\REQUIRE Observed submatrices $\{\ma^{(i)}\}_{i\in[m]}$ with index sets $\{\mathcal{U}_i\}_{i\in[m]} \subseteq [N]$, embedding dimensions $d_+$ and $d_-$ with $d = d_+ + d_-$, pairwise alignment weights $\{\pi_{i,j}\}_{i\neq j\in[m]}$ and aggregation weights $\{\tau_i\}_{i\in[m]}$ (see Remark~\ref{rem:weights_indef} for choices).

\STATE \textbf{Step 1: Obtain local latent position estimates $\{\hat{\mx}^{(i)}\}_{i\in[m]}$}
\STATE Obtain $\hat{\mx}^{(i)} \in \mathbb{R}^{n_i \times d}$ such that $\ma^{(i)} \approx \hat{\mx}^{(i)} \mi_{d_+,d_-} \hat{\mx}^{(i)\top}$ for each $i \in [m]$, where $n_i = |\mathcal{U}_i|$.

\STATE \textbf{Step 2: Obtain alignment matrices $\{\hat{\mw}^{(i)}\}_{i\in[m]}$ via synchronization}
\STATE Solve for $\{\hat{\mw}^{(i)}\}_{i\in[m]}$ by
\[
\{\hat{\mw}^{(i)}\}_{i=1}^m = \argmin_{\mo^{(1)}, \ldots, \mo^{(m)} \in \mathcal{O}_{d_+,d_-}}
\sum_{\substack{(i,j): \\ \mathcal{U}_i \cap \mathcal{U}_j \neq \varnothing}}
\pi_{i,j} \, \|\hat{\mx}^{(i)}_{\langle \mathcal{U}_i \cap \mathcal{U}_j \rangle} \mo^{(i)} - \hat{\mx}^{(j)}_{\langle \mathcal{U}_i \cap \mathcal{U}_j \rangle} \mo^{(j)}\|_F^2
\]
using the alternating optimization based on \eqref{eq:update_indef_L} and \eqref{eq:update_indef_R} with averaging and projection onto $\mathcal{O}_{d_+,d_-}$, $\hat{\mw}^{(1)} = \mi_d$ fixed, and initialization from CMMI.

\STATE \textbf{Step 3: Aggregate $\{\hat{\mx}^{(i)} \hat{\mw}^{(i)}\}_{i\in[m]}$}
\STATE For each entity $k \in [N]$, compute
$\hat{\bm{x}}_k = \frac{\sum_{i: k\in\mathcal{U}_i} \tau_i (\hat{\mx}^{(i)}\hat{\mw}^{(i)})_{\langle k\rangle}}{\sum_{i: k\in\mathcal{U}_i} \tau_i}.$
\STATE Form $\hat{\mx} = [\hat{\bm{x}}_1|\cdots|\hat{\bm{x}}_N]^\top$ and compute $\hat{\mpp} = \hat{\mx} \mi_{d_+,d_-} \hat{\mx}^\top$.

\ENSURE Estimated latent position matrix $\hat{\bm{X}}$ and estimated complete matrix $\hat{\bm{P}}$.
\end{algorithmic}
\endgroup
\end{algorithm}

\begin{Remark}
\label{rem:weights_indef}
The discussion of the weights $\{\pi_{i,j}\}$ and $\{\tau_i\}$ in Remark~\ref{rem:weights} applies to Algorithm~\ref{alg:gsmmi_indef} without change, except that $\hat{\mpp}^{(i)} \colonequals \hat{\mx}^{(i)}\mi_{d_+,d_-}\hat{\mx}^{(i)\top}$ in the noise-variance estimate $\hat{\sigma}_i^2$.
\end{Remark}

We compare GSMMI for symmetric indefinite matrices in Algorithm~\ref{alg:gsmmi_indef} with CMMI and the four matrix-completion baselines under the same ring-structured block design as in Section~\ref{sec:simu_block}, the difference being that we now take $\bm{\Lambda} = \diag(\lambda, \tfrac{3}{4}\lambda, -\tfrac{1}{2}\lambda)$, so that $\bm{P}$ is indefinite with $d_+ = 2$ positive and $d_- = 1$ negative eigenvalues.
Other settings, including the evaluation by the unobserved relative error, are as in Section~\ref{sec:simu_block}.

Figure~\ref{fig:simu_varym_indef} shows that both methods outperform the completion baselines.
GSMMI is more accurate than CMMI throughout, with a widening gap: at $m = 21$, their errors are $0.324$ and $0.792$, respectively, a relative reduction of about $59\%$.
Their running times remain comparable and orders of magnitude smaller than the slower baselines.

\begin{figure}[htbp]
    \centering
    \includegraphics[width=0.55\textwidth]{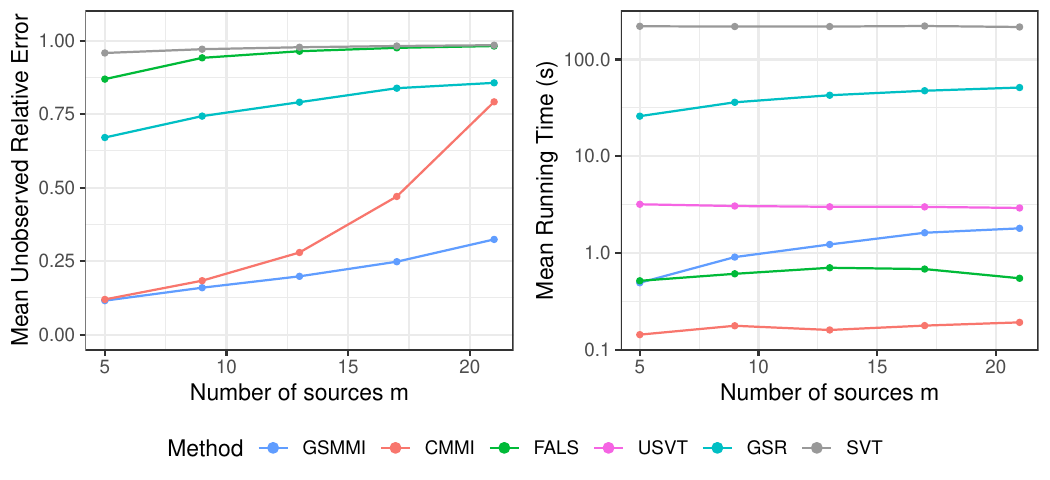}
    \caption{\footnotesize Performance of GSMMI, CMMI, and four matrix-completion baselines (FALS, USVT, GSR, SVT) for symmetric indefinite matrices under the ring-structured design with shortcuts described in Section~\ref{sec:simu_block}, with $\bm{\Lambda} = \diag(\lambda, \tfrac{3}{4}\lambda, -\tfrac{1}{2}\lambda)$ ($d_+ = 2$, $d_- = 1$), as $m$ varies over $\{5, 9, 13, 17, 21\}$. Other settings are as in Figure~\ref{fig:simu_block}. Left: unobserved relative Frobenius error; the $y$-axis is capped at $1.05$ for readability, and USVT (with errors up to $1.8$) exceeds this range. Right: running time in seconds (log scale). Results are averaged over $100$ independent Monte Carlo replicates.}
    \label{fig:simu_varym_indef}
\end{figure}

 \subsection{GSMMI for asymmetric or rectangular matrices}
\label{sec:extension_asy}

Suppose $\mpp \in \mathbb{R}^{N \times M}$ is a low-rank matrix of rank $d$, with singular value decomposition $\mpp = \muu \mSigma \mv^\top$, where $\mSigma \in \mathbb{R}^{d \times d}$ contains the singular values in descending order, and the orthonormal columns of $\muu \in \mathbb{R}^{N \times d}$ and $\mv \in \mathbb{R}^{M \times d}$ are the left and right singular vectors. The left and right latent positions are $\mx = \muu \mSigma^{1/2} \in \mathbb{R}^{N \times d}$ and $\my = \mv \mSigma^{1/2} \in \mathbb{R}^{M \times d}$, so that $\mpp = \mx \my^\top$. For source $i$, let $\mathcal{U}_i \subseteq [N]$ and $\mathcal{V}_i \subseteq [M]$ denote the row and column index sets, so that $\mpp^{(i)} = \mpp_{\mathcal{U}_i, \mathcal{V}_i} = \mx_{\mathcal{U}_i} \my_{\mathcal{V}_i}^\top$.

From each observed $\ma^{(i)}$, we obtain estimated left and right latent positions $(\hat{\mx}^{(i)}, \hat{\my}^{(i)})$ such that $\ma^{(i)} \approx \hat{\mx}^{(i)} \hat{\my}^{(i)\top}$. For example, one may set $\hat{\mx}^{(i)} = \hat{\muu}^{(i)} (\hat{\mSigma}^{(i)})^{1/2}$ and $\hat{\my}^{(i)} = \hat{\mv}^{(i)} (\hat{\mSigma}^{(i)})^{1/2}$ from the leading singular vectors of $\ma^{(i)}$. Each estimate $(\hat{\mx}^{(i)}, \hat{\my}^{(i)})$ recovers $(\mx_{\mathcal{U}_i}, \my_{\mathcal{V}_i})$ only up to an invertible transformation $\mw^{(i)} \in \mathbb{R}^{d \times d}$, which acts on the left latent positions as $\mx_{\mathcal{U}_i} = \hat{\mx}^{(i)} \mw^{(i)}$ and on the right latent positions as $\my_{\mathcal{V}_i} = \hat{\my}^{(i)} (\mw^{(i)\top})^{-1}$.

To globally synchronize all sources, we align the left latent positions across overlapping rows and the right latent positions across overlapping columns simultaneously. A single transformation $\mw^{(i)}$ governs both, acting on the left latent positions as $\mw^{(i)}$ and on the right latent positions as $(\mw^{(i)\top})^{-1}$. We thus seek transformation matrices $\{\hat{\mw}^{(i)}\}$ that jointly minimize
\begin{equation}\label{eq:sync_asy}
\begin{aligned}
\{\hat{\mw}^{(i)}\}_{i=1}^m = \argmin_{\mo^{(1)}, \ldots, \mo^{(m)} \in \mathbb{R}^{d \times d}}
&\sum_{i < j : \mathcal{U}_i \cap \mathcal{U}_j \neq \varnothing}
\pi^{\mx}_{i,j} \, \|\hat{\mx}^{(i)}_{\langle \mathcal{U}_i \cap \mathcal{U}_j \rangle} \mo^{(i)} - \hat{\mx}^{(j)}_{\langle \mathcal{U}_i \cap \mathcal{U}_j \rangle} \mo^{(j)}\|_F^2 \\
&+ \sum_{i < j : \mathcal{V}_i \cap \mathcal{V}_j \neq \varnothing}
\pi^{\my}_{i,j} \, \|\hat{\my}^{(i)}_{\langle \mathcal{V}_i \cap \mathcal{V}_j \rangle} (\mo^{(i)\top})^{-1} - \hat{\my}^{(j)}_{\langle \mathcal{V}_i \cap \mathcal{V}_j \rangle} (\mo^{(j)\top})^{-1}\|_F^2.
\end{aligned}
\end{equation}
We set $\pi^{\mx}_{i,j}=1$ on row overlaps and $\pi^{\my}_{i,j}=1$ on column overlaps.
The two terms depend on the transformation differently, since the second involves it through $(\mo^{(i)\top})^{-1}$. We therefore solve \eqref{eq:sync_asy} by an alternating scheme initialized from the CMMI solution, with $\hat{\mw}^{(1)} = \mi_d$ fixed as the reference. At each step, to update $\hat{\mw}^{(i)}$ for $i = 2, \ldots, m$ with the remaining transformations held fixed, we compute two separate least squares estimates. Aligning the left latent positions gives
\begin{equation}\label{eq:update_asy_x}
\hat{\mw}_{\mx}^{(i)} = \argmin_{\mo \in \mathbb{R}^{d \times d}} \sum_{j \neq i : \mathcal{U}_i \cap \mathcal{U}_j \neq \varnothing} \pi^{\mx}_{i,j} \, \|\hat{\mx}^{(i)}_{\langle \mathcal{U}_i \cap \mathcal{U}_j \rangle} \mo - \hat{\mx}^{(j)}_{\langle \mathcal{U}_i \cap \mathcal{U}_j \rangle} \hat{\mw}^{(j)}\|_F^2,
\end{equation}
which estimates $\hat{\mw}^{(i)}$ from the row overlaps, and aligning the right latent positions gives
\begin{equation}\label{eq:update_asy_y}
\hat{\mw}_{\my}^{(i)} = \argmin_{\mo \in \mathbb{R}^{d \times d}} \sum_{j \neq i : \mathcal{V}_i \cap \mathcal{V}_j \neq \varnothing} \pi^{\my}_{i,j} \, \|\hat{\my}^{(i)}_{\langle \mathcal{V}_i \cap \mathcal{V}_j \rangle} \mo - \hat{\my}^{(j)}_{\langle \mathcal{V}_i \cap \mathcal{V}_j \rangle} (\hat{\mw}^{(j)\top})^{-1}\|_F^2,
\end{equation}
which estimates $(\hat{\mw}^{(i)\top})^{-1}$ from the column overlaps. Both \eqref{eq:update_asy_x} and \eqref{eq:update_asy_y} are standard least squares problems with closed-form solutions via the Moore-Penrose pseudoinverse.

Since $\hat{\mw}_{\mx}^{(i)}$ and $(\hat{\mw}_{\my}^{(i)\top})^{-1}$ both estimate $\hat{\mw}^{(i)}$, we combine them into a single update.
We weight them by the total squared norms of the embeddings on their respective overlaps,
\begin{equation}\label{eq:asy_weights}
\kappa_{\mx}^{(i)}  \colonequals  \sum_{j \neq i : \mathcal{U}_i \cap \mathcal{U}_j \neq \varnothing} \pi^{\mx}_{i,j} \, \|\hat{\mx}^{(i)}_{\langle \mathcal{U}_i \cap \mathcal{U}_j \rangle}\|_F^2, \qquad
\kappa_{\my}^{(i)}  \colonequals  \sum_{j \neq i : \mathcal{V}_i \cap \mathcal{V}_j \neq \varnothing} \pi^{\my}_{i,j} \, \|\hat{\my}^{(i)}_{\langle \mathcal{V}_i \cap \mathcal{V}_j \rangle}\|_F^2,
\end{equation}
which equal the traces of the Gram matrices in \eqref{eq:update_asy_x} and \eqref{eq:update_asy_y} and quantify the total amount of overlap information available in each direction, and update
\begin{equation}\label{eq:asy_combine}
\hat{\mw}^{(i)} = \frac{\kappa_{\mx}^{(i)}\, \hat{\mw}_{\mx}^{(i)} + \kappa_{\my}^{(i)}\, (\hat{\mw}_{\my}^{(i)\top})^{-1}}{\kappa_{\mx}^{(i)} + \kappa_{\my}^{(i)}}.
\end{equation}
Since the two least squares estimates are unconstrained, the combined $\hat{\mw}^{(i)}$ may become poorly scaled and accumulate scale drift over the iterations. We therefore rescale its singular values toward unity after each update to keep it well-conditioned, which improves numerical stability. We repeat these updates for $i = 2, \ldots, m$ until convergence. After obtaining $\{\hat{\mw}^{(i)}\}$, we aggregate the aligned left and right latent positions separately, as in \eqref{eq:aggregation}, to form $\hat{\mx}$ and $\hat{\my}$, and estimate $\hat{\mpp} = \hat{\mx} \hat{\my}^\top$. The complete procedure is summarized in Algorithm~\ref{alg:gsmmi_asy}.

\begin{algorithm}[htbp]
\caption{GSMMI for asymmetric matrices}
\label{alg:gsmmi_asy}
\begingroup
\small
\setlength{\itemsep}{0pt}
\setlength{\parskip}{0pt}
\setlength{\abovedisplayskip}{4pt}
\setlength{\belowdisplayskip}{4pt}
\begin{algorithmic}
\REQUIRE Observed submatrices $\{\ma^{(i)}\}_{i\in[m]}$ with row and column index sets $\{(\mathcal{U}_i, \mathcal{V}_i)\}_{i\in[m]} \subseteq [N] \times [M]$, embedding dimension $d$, row and column alignment weights $\{\pi^{\mx}_{i,j}\}_{i\neq j\in[m]}$, $\{\pi^{\my}_{i,j}\}_{i\neq j\in[m]}$ and aggregation weights $\{\tau^{\mx}_i\}_{i\in[m]}$, $\{\tau^{\my}_i\}_{i\in[m]}$ (see Remark~\ref{rem:weights_asy} for choices).

\STATE \textbf{Step 1: Obtain local latent position estimates $\{(\hat{\mx}^{(i)}, \hat{\my}^{(i)})\}_{i\in[m]}$}
\STATE Obtain $\hat{\mx}^{(i)} \in \mathbb{R}^{|\mathcal{U}_i| \times d}$ and $\hat{\my}^{(i)} \in \mathbb{R}^{|\mathcal{V}_i| \times d}$ such that $\ma^{(i)} \approx \hat{\mx}^{(i)} \hat{\my}^{(i)\top}$ for each $i \in [m]$.

\STATE \textbf{Step 2: Obtain alignment matrices $\{\hat{\mw}^{(i)}\}_{i\in[m]}$ via synchronization}
\STATE Solve for $\{\hat{\mw}^{(i)}\}_{i\in[m]}$ by
\[
\begin{aligned}
\{\hat{\mw}^{(i)}\}_{i=1}^m = 	\argmin_{\mo^{(1)}, \ldots, \mo^{(m)} \in \mathbb{R}^{d \times d}}
\sum_{\substack{i < j : \\ \mathcal{U}_i \cap \mathcal{U}_j \neq \varnothing}}
&\pi^{\mx}_{i,j} \, \|\hat{\mx}^{(i)}_{\langle \mathcal{U}_i \cap \mathcal{U}_j \rangle} \mo^{(i)} - \hat{\mx}^{(j)}_{\langle \mathcal{U}_i \cap \mathcal{U}_j \rangle} \mo^{(j)}\|_F^2
\\&+ \sum_{\substack{i < j : \\ \mathcal{V}_i \cap \mathcal{V}_j \neq \varnothing}}
\pi^{\my}_{i,j} \, \|\hat{\my}^{(i)}_{\langle \mathcal{V}_i \cap \mathcal{V}_j \rangle} (\mo^{(i)\top})^{-1} - \hat{\my}^{(j)}_{\langle \mathcal{V}_i \cap \mathcal{V}_j \rangle} (\mo^{(j)\top})^{-1}\|_F^2,
\end{aligned}
\]
using the alternating optimization based on \eqref{eq:update_asy_x} and \eqref{eq:update_asy_y} with the weighted combination \eqref{eq:asy_combine} 
and rescaling of the singular values toward unity, with $\hat{\mw}^{(1)} = \mi_d$ fixed and initialization from CMMI.

\STATE \textbf{Step 3: Aggregate the aligned latent positions}
\STATE For each row entity $k \in [N]$ and column entity $l \in [M]$, compute
\[
\begin{aligned}
\hat{\bm{x}}_k &= \frac{\sum_{i: k\in\mathcal{U}_i} \tau^{\mx}_i (\hat{\mx}^{(i)}\hat{\mw}^{(i)})_{\langle k\rangle}}{\sum_{i: k\in\mathcal{U}_i} \tau^{\mx}_i},\quad
\hat{\bm{y}}_l &= \frac{\sum_{i: l\in\mathcal{V}_i} \tau^{\my}_i (\hat{\my}^{(i)}(\hat{\mw}^{(i)\top})^{-1})_{\langle l\rangle}}{\sum_{i: l\in\mathcal{V}_i} \tau^{\my}_i}.
\end{aligned}
\]
\STATE Form $\hat{\mx} = [\hat{\bm{x}}_1|\cdots|\hat{\bm{x}}_N]^\top$ and $\hat{\my} = [\hat{\bm{y}}_1|\cdots|\hat{\bm{y}}_M]^\top$, and compute $\hat{\mpp} = \hat{\mx}\hat{\my}^\top$.

\ENSURE Estimated latent position matrices $\hat{\bm{X}}$ and $\hat{\bm{Y}}$, and estimated complete matrix $\hat{\bm{P}}$.
\end{algorithmic}
\endgroup
\end{algorithm}

We note, as also mentioned for CMMI in \citet{zheng2026chain}, that two sources of an asymmetric matrix can be synchronized as long as they share sufficiently many row indices \emph{or} sufficiently many column indices; overlap in both is not required. This is naturally accommodated by the weighted update \eqref{eq:asy_combine}. If source $i$ has no column overlap with any other source, then $\kappa_{\my}^{(i)} = 0$ and the update reduces to $\hat{\mw}^{(i)} = \hat{\mw}_{\mx}^{(i)}$, relying solely on the row overlaps; symmetrically, if it has no row overlap, then $\kappa_{\mx}^{(i)} = 0$ and the update reduces to $\hat{\mw}^{(i)} = (\hat{\mw}_{\my}^{(i)\top})^{-1}$.

\begin{Remark}
\label{rem:weights_asy}
As in Remark~\ref{rem:weights}, the row and column weights $\{\pi^{\mx}_{i,j}\}$, $\{\pi^{\my}_{i,j}\}$ and $\{\tau^{\mx}_i\}$, $\{\tau^{\my}_i\}$ are left as inputs to Algorithm~\ref{alg:gsmmi_asy}, with the simplest choice being uniform weighting, $\pi^{\mx}_{i,j} \equiv \pi^{\my}_{i,j} \equiv 1$ and $\tau^{\mx}_i \equiv \tau^{\my}_i \equiv 1$. The inverse-variance choice carries over analogously, with the error levels $c^{\mx}_i$ and $c^{\my}_i$ now reflecting the estimation accuracy of the left and right latent positions, respectively, and can be estimated by
$(c^{\mx}_i)^2 = \hat{\sigma}_i^2\,\operatorname{tr}\big((\hat{\my}^{(i)\top}\hat{\my}^{(i)})^{-1}\big)$ and
$(c^{\my}_i)^2 = \hat{\sigma}_i^2\,\operatorname{tr}\big((\hat{\mx}^{(i)\top}\hat{\mx}^{(i)})^{-1}\big)$, where $\hat{\sigma}_i^2 = \|\ma^{(i)}-\hat{\mpp}^{(i)}\|_F^2/(|\mathcal{U}_i||\mathcal{V}_i|)$ and $\hat{\mpp}^{(i)} \colonequals \hat{\mx}^{(i)}\hat{\my}^{(i)\top}$. 
\end{Remark}

We compare GSMMI for asymmetric matrices in Algorithm~\ref{alg:gsmmi_asy} with CMMI and the four matrix-completion baselines under a ring-structured block design analogous to that in Section~\ref{sec:simu_block}, the only difference being that $\bm{P} = \bm{U}\bm{\Sigma}\bm{V}^\top \in \mathbb{R}^{N \times M}$ is now rectangular, with $\bm{\Sigma} = \diag(\lambda, \tfrac{3}{4}\lambda, \tfrac{1}{2}\lambda)$, and the ring-and-shortcut structure is applied independently to the row and column index sets, so that each source observes a submatrix of overlapping rows and columns; the left panel of Figure~\ref{fig:simu_varym_asy} illustrates the resulting observation pattern, taking $m = 5$ as an example.
We fix $N \approx 1000$ and $M \approx 1400$, and all other settings, including the evaluation by the unobserved relative error, are as in Section~\ref{sec:simu_block}.

Figure~\ref{fig:simu_varym_asy} shows a similar comparison: both methods outperform the completion baselines, and GSMMI is more accurate than CMMI throughout.
The gap widens with $m$, reaching a relative error reduction of about $55\%$ at $m = 25$ ($0.326$ versus $0.732$).
The two methods have comparable running times, both orders of magnitude smaller than those of the slower baselines.

\begin{figure}[htbp]
    \centering
    \includegraphics[width=0.33\textwidth]{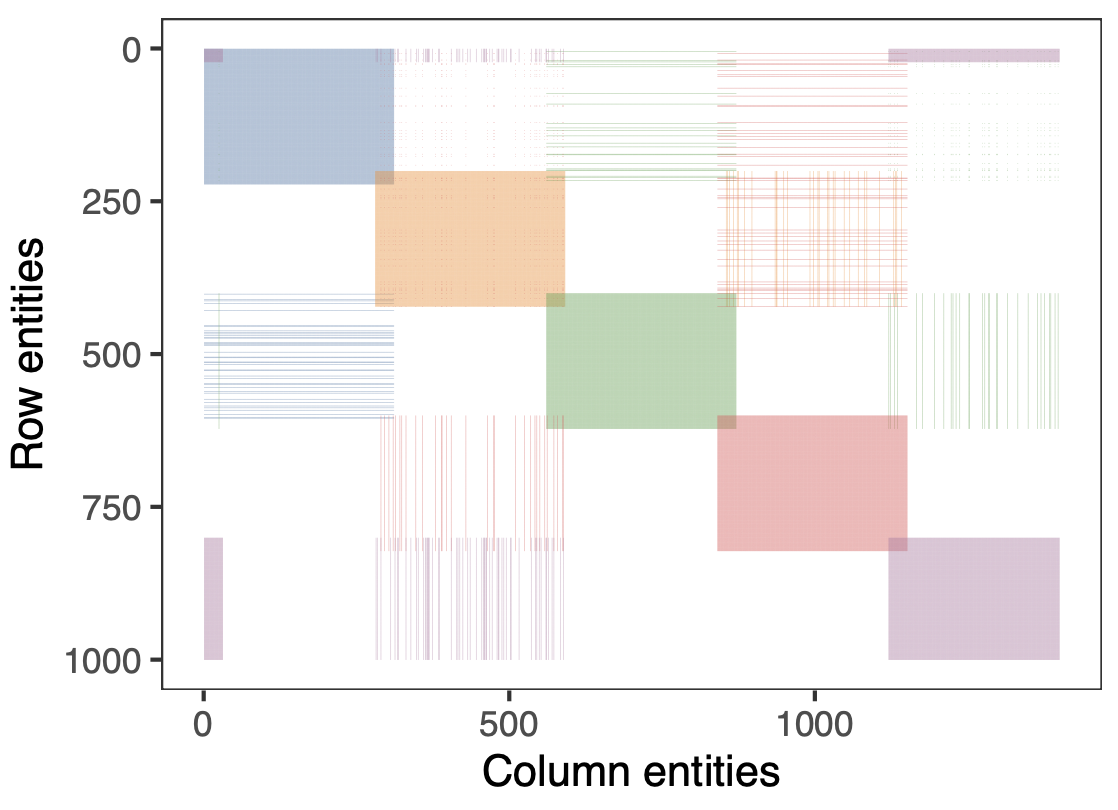}
    \includegraphics[width=0.55\textwidth]{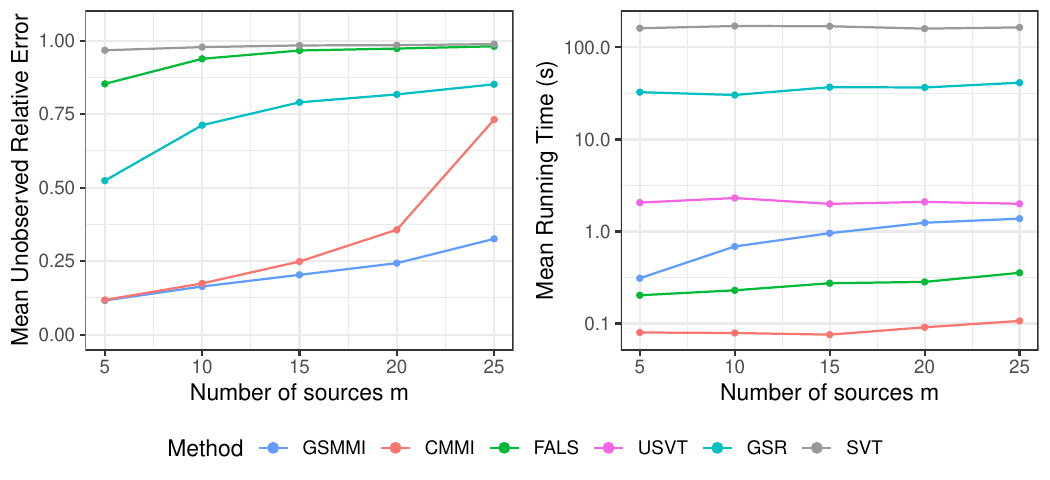}
    \caption{\footnotesize Performance of GSMMI, CMMI, and four matrix-completion baselines (FALS, USVT, GSR, SVT) for asymmetric matrices under a ring-structured design with shortcuts analogous to Section~\ref{sec:simu_block}, applied independently to rows and columns, with $\bm{\Sigma} = \diag(\lambda, \tfrac{3}{4}\lambda, \tfrac{1}{2}\lambda)$, as $m$ varies over $\{5, 10, 15, 20, 25\}$, and $N \approx 1000$, $M \approx 1400$. Left: the observation pattern, shown for $m = 5$ as an example, where each color corresponds to one source and overlapping regions are shown by blended colors where blocks share entities. Middle: unobserved relative Frobenius error; the $y$-axis is capped at $1.05$ for readability, and USVT (with errors up to $1.16$) exceeds this range. Right: running time in seconds (log scale). Results are averaged over $100$ independent Monte Carlo replicates.}
    \label{fig:simu_varym_asy}
\end{figure}

As noted above, two sources can be synchronized as long as they share sufficiently many row indices \emph{or} column indices, so overlap in only one direction already suffices.
To illustrate this empirically, we consider a more challenging variant in which the sources share \emph{no} columns at all: the rows still follow the ring-and-shortcut structure, but the column index sets are disjoint, so that each column entity is observed by exactly one source.
Alignment therefore relies entirely on the shared rows, and the column latent positions are recovered indirectly through the row-based synchronization.
All other settings are as in the previous asymmetric experiment; the left panel of Figure~\ref{fig:simu_varym_asy_coldisjoint} illustrates the resulting observation pattern for $m = 5$, where the dashed vertical lines mark the boundaries between the disjoint column blocks.

Both methods outperform the completion baselines in this column-disjoint setting (Figure~\ref{fig:simu_varym_asy_coldisjoint}).
GSMMI is more accurate throughout, with the gap reaching a relative reduction of about $46\%$ at $m = 24$ ($0.379$ versus $0.704$).
Thus, the accuracy gains persist when alignment uses only row overlaps.
The two methods have comparable running times, both orders of magnitude smaller than those of the slower baselines.

\begin{figure}[htbp]
    \centering
    \includegraphics[width=0.33\textwidth]{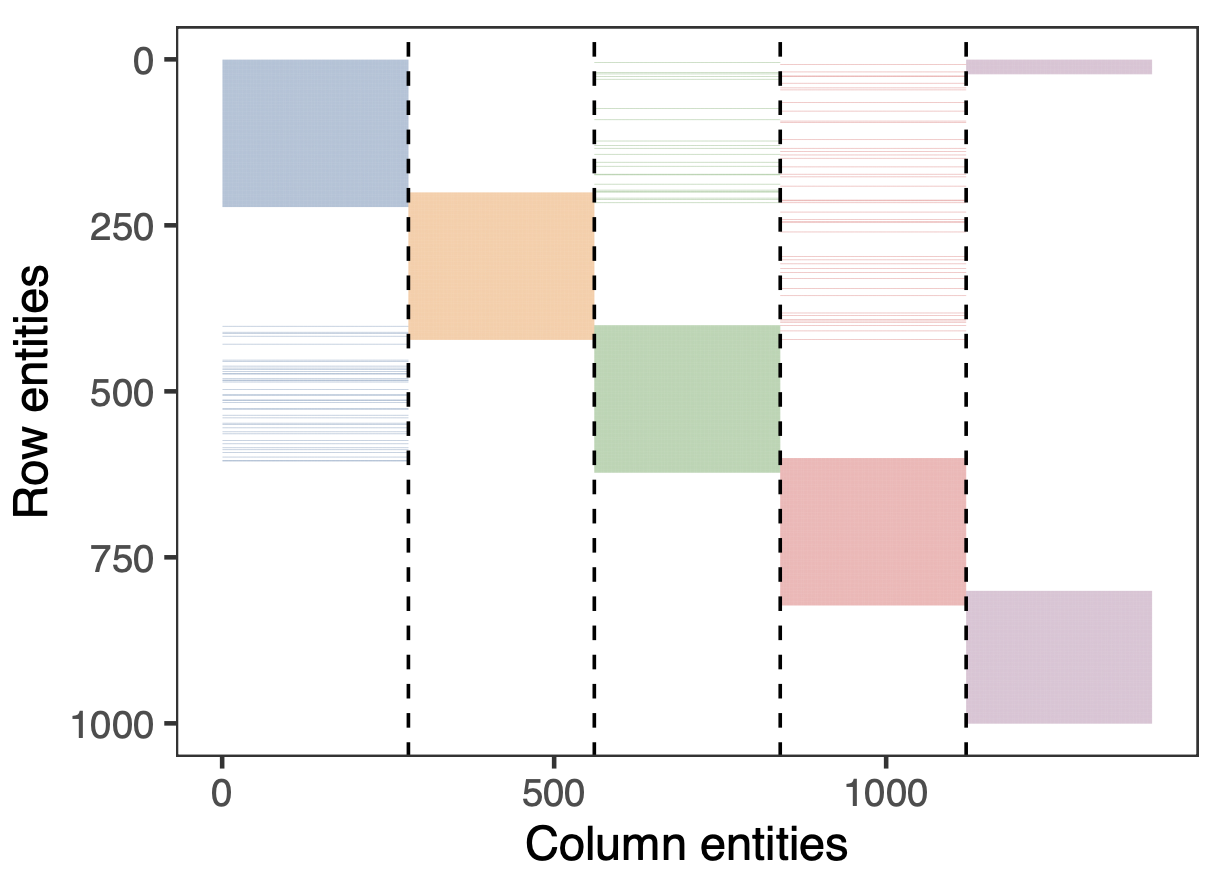}
    \includegraphics[width=0.55\textwidth]{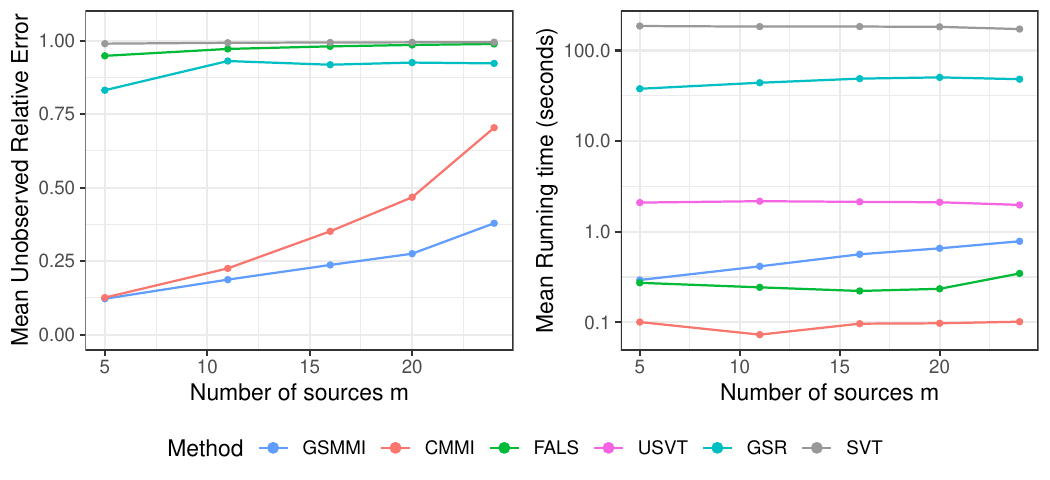}
    \caption{\footnotesize Performance of GSMMI, CMMI, and four matrix-completion baselines (FALS, USVT, GSR, SVT) for asymmetric matrices under the column-disjoint ring-structured design, where the rows follow the ring-structured design with shortcuts but the column blocks are disjoint, with $\bm{\Sigma} = \diag(\lambda, \tfrac{3}{4}\lambda, \tfrac{1}{2}\lambda)$, as $m$ varies over $\{5, 11, 16, 20, 24\}$, and $N \approx 1000$, $M \approx 1400$. Left: the observation pattern, shown for $m = 5$ as an example, where each color corresponds to one source and dashed vertical lines mark the boundaries between disjoint column blocks. Middle: unobserved relative Frobenius error; the $y$-axis is capped at $1.05$ for readability, and USVT (with errors up to $1.16$) exceeds this range. Right: running time in seconds (log scale). Results are averaged over $100$ independent Monte Carlo replicates.}
    \label{fig:simu_varym_asy_coldisjoint}
\end{figure}

\newpage

\bibliographystyle{unsrtnat}
\bibliography{ref}

\newpage

\appendix

\counterwithin{figure}{section}
\counterwithin{table}{section}

\section{Additional numerical comparisons}\label{sec:init CMMI}

As mentioned in Remark~\ref{rem:comp}, high-quality initialization from CMMI is crucial for the computational efficiency of GSMMI. To demonstrate this, we compare the proposed GSMMI with CMMI initialization against GSMMI with cold start (GSMMI-cold), where the initialization is automatically determined by the CVXR solver from a generic feasible point without using any problem-specific initialization.

We also consider an alternative approach to obtaining a fast algorithm for positive semidefinite matrix integration. Notice that for each $(i,j)$ with $\mathcal{U}_i \cap \mathcal{U}_j \neq \varnothing$, we have
$$
\begin{aligned}
	\|\hat{\mx}^{(i)}_{\langle\mathcal{U}_i \cap \mathcal{U}_j\rangle}\mo^{(i)}-\hat{\mx}^{(j)}_{\langle\mathcal{U}_i \cap \mathcal{U}_j\rangle}\mo^{(j)}\|_F^2
&=\operatorname{tr}\!\left(\hat{\mx}^{(i)\top}_{\langle\mathcal{U}_i \cap \mathcal{U}_j\rangle}\hat{\mx}^{(i)}_{\langle\mathcal{U}_i \cap \mathcal{U}_j\rangle}\right) +
\operatorname{tr}\!\left(\hat{\mx}^{(j)\top}_{\langle\mathcal{U}_i \cap \mathcal{U}_j\rangle}\hat{\mx}^{(j)}_{\langle\mathcal{U}_i \cap \mathcal{U}_j\rangle}\right)\\
&\quad-2\operatorname{tr}\!\left(\hat{\mx}^{(i)\top}_{\langle\mathcal{U}_i \cap \mathcal{U}_j\rangle}\hat{\mx}^{(j)}_{\langle\mathcal{U}_i \cap \mathcal{U}_j\rangle}\mo^{(j)}\mo^{(i)\top}\right).
\end{aligned}
$$
The first two traces do not depend on $\mo^{(i)},\mo^{(j)}$, so \eqref{eq:partial MLE} is equivalent to
\begin{equation}\label{eq:argmax}
	\argmax_{\mo^{(1)},\dots,\mo^{(m)}\in \mathcal{O}_d}
	\sum_{i<j:\,\mathcal{U}_i \cap \mathcal{U}_j \neq \varnothing}
\operatorname{tr}\!\left(	\pi_{i,j}\hat{\mx}^{(j)\top}_{\langle\mathcal{U}_i \cap \mathcal{U}_j\rangle}\hat{\mx}^{(i)}_{\langle\mathcal{U}_i \cap \mathcal{U}_j\rangle}\,\mo^{(i)}\mo^{(j)\top}\right).
\end{equation}
This is a little Grothendieck problem over the orthogonal group. For such problems, \citet{bandeira2016approximating} propose a fast approximation algorithm (SDP-approx) based on a semidefinite programming relaxation with Gaussian random projections to recover feasible orthogonal matrices.

Figure~\ref{fig:simu_accuracy_full} compares GSMMI, GSMMI-cold, SDP-approx, and CMMI under the setting of Section~\ref{sec:simu_error}.
GSMMI and GSMMI-cold achieve nearly identical accuracy, both more accurate than CMMI.
SDP-approx is slightly less accurate than the GSMMI methods but more accurate than CMMI.
Summed over all settings and replicates, the total running times are $286.8$ seconds for CMMI, $960.6$ seconds for GSMMI, $28861.3$ seconds for GSMMI-cold, and $5124.5$ seconds for SDP-approx, so GSMMI requires only about $0.03$ times the running time of GSMMI-cold and $0.19$ times that of SDP-approx.

\begin{figure}[htbp]
    \centering
    \includegraphics[width=\textwidth]{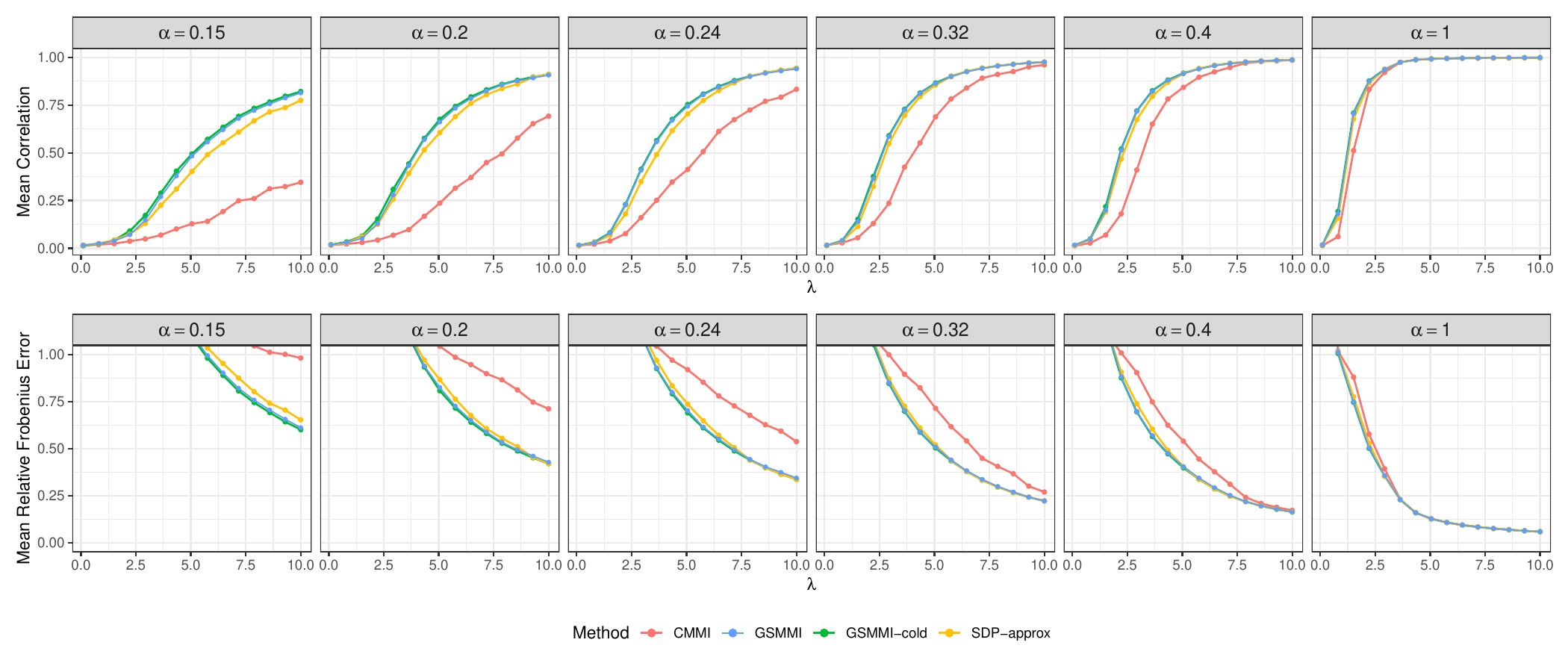}
    \caption{\footnotesize Estimation accuracy of GSMMI, GSMMI-cold, SDP-approx, and CMMI as we vary $\lambda$ and $\alpha$, under the same setting as Figure~\ref{fig:simu_accuracy}. Top: correlation $\langle \bm{P}, \hat{\bm{P}} \rangle_F / (\|\bm{P}\|_F \|\hat{\bm{P}}\|_F)$. Bottom: relative Frobenius error $\|\bm{P} - \hat{\bm{P}}\|_F / \|\bm{P}\|_F$.}
    \label{fig:simu_accuracy_full}
\end{figure}


\section{Proof of main results}

\subsection{Proof of Theorem~\ref{thm:xi}}

Fix $k$ with $\rho_k^*>1$ and suppress the source index.
Write $\lambda,\bm{u},\hat\lambda,\hat{\bm{u}}$ for $\lambda_k^{(i)},\bm{u}_k^{(i)},\hat\lambda_k^{(i)},\hat{\bm{u}}_k^{(i)}$, respectively, and set $\zeta=1-\sigma^2n/\lambda^2$.
The assumptions give $\lambda/(\sigma\sqrt n)\xrightarrow{\mathbb{P}}\rho_k^*$.
Lemmas~\ref{lemma:lambdai phase transition} and~\ref{lemma:ui phase transition} therefore give
\[
\frac{\hat\lambda}{\sigma\sqrt n}\xrightarrow{\mathbb{P}}\rho_k^*+\frac{1}{\rho_k^*},
\qquad
\frac{\lambda}{\sigma}
[\hat{\bm{u}}-\sqrt{\zeta}\bm{u}]_{\mathcal I}
\xrightarrow{\mathcal D}\mathcal N(\bm 0,\bm I_{|\mathcal I|}).
\]

We first check the estimated scaling factors.
For GOE noise, $\|\bm N^{(i)}\|_F^2/(\sigma^2n^2)\xrightarrow{\mathbb{P}}1$ and $\|\bm N^{(i)}\|/(\sigma\sqrt n)$ is bounded in probability.
Since $d$ is fixed and $\|\bm P^{(i)}\|/(\sigma\sqrt n)$ is bounded in probability, it follows that
\[
\frac{\hat\sigma_i^2}{\sigma^2}
=
\frac{\|\bm A^{(i)}\|_F^2-\sum_{\ell=1}^d(\hat\lambda_\ell^{(i)})^2}{\sigma^2n^2}
\xrightarrow{\mathbb{P}}1.
\]
Set $\hat\zeta\colonequals1-\hat\sigma_i^2n/\phi(\hat\lambda)^2$.
The formula for $\phi$ and the normalized eigenvalue limit now imply
\[
\frac{\phi(\hat\lambda)}{\lambda}\xrightarrow{\mathbb{P}}1,
\qquad
\hat\zeta-\zeta\xrightarrow{\mathbb{P}}0.
\]
Because $\rho_k^*>1$, the debiasing factors are real and positive with probability tending to one.
Define
\[
a\colonequals\frac{\sqrt\lambda}{\sqrt\zeta},
\qquad
\hat a\colonequals\frac{\sqrt{\phi(\hat\lambda)}}{\sqrt{\hat\zeta}}.
\]
Then $\hat a/a\xrightarrow{\mathbb{P}}1$, and $\hat{\bm x}_k^{(i)}=\hat a\hat{\bm u}$ on this event.

The incoherence bound and Lemma~\ref{lemma:ui phase transition} imply that $\sqrt n[\hat{\bm u}]_{\mathcal I}$ is bounded in probability.
Using $\bm x_k^{(i)}=\sqrt\lambda\bm u$, we obtain
\begin{align*}
\frac{\sqrt{\lambda\zeta}}{\sigma}
[\hat{\bm x}_k^{(i)}-\bm x_k^{(i)}]_{\mathcal I}
&=
\frac{\lambda}{\sigma}[\hat{\bm u}-\sqrt\zeta\bm u]_{\mathcal I}
+
\frac{\lambda}{\sigma}\left(\frac{\hat a}{a}-1\right)[\hat{\bm u}]_{\mathcal I}.
\end{align*}
The first term converges to $\mathcal N(\bm 0,\bm I_{|\mathcal I|})$ by Lemma~\ref{lemma:ui phase transition}.
The second converges to zero in probability, since $\lambda/(\sigma\sqrt n)\to\rho_k^*$ and $\hat a/a\to1$ in probability.
Finally, $\lambda/(\alpha\lambda_k)\xrightarrow{\mathbb{P}}1$ and $n/(\alpha N)\to1$ show that the normalization in Theorem~\ref{thm:xi} divided by $\sqrt{\lambda\zeta}/\sigma$ converges to one in probability.
Slutsky's theorem completes the proof.

	\subsection{Proof of Theorem~\ref{thm:cmmi_error}}

For $i\neq j$, define
$A_{i,j}  \colonequals  \langle \hat{\bm{x}}^{(i)}_{\langle\mathcal{S}_{i,j}\rangle}, \hat{\bm{x}}^{(j)}_{\langle\mathcal{S}_{i,j}\rangle} \rangle$,
so that $\hat{w}^{(i,j)} = \operatorname{sgn}(A_{i,j})$. We decompose $A_{i,j}$ into a signal term and noise terms:
$$
A_{i,j} = M_{i,j} + (E_1)_{i,j} + (E_2)_{i,j} + (E_3)_{i,j},
$$
where
\begin{align*}
M_{i,j} & \colonequals  w^{(i)} w^{(j)} \langle \bm{x}_{\mathcal{S}_{i,j}}, \bm{x}_{\mathcal{S}_{i,j}} \rangle = w^{(i)} w^{(j)} \| \bm{x}_{\mathcal{S}_{i,j}} \|^2, \\
(E_1)_{i,j} & \colonequals  \gamma  w^{(i)} \langle \bm{x}_{\mathcal{S}_{i,j}}, (\bm{\delta}_j)_{\langle\mathcal{S}_{i,j}\rangle} \rangle, \\
(E_2)_{i,j} & \colonequals  \gamma   w^{(j)} \langle (\bm{\delta}_i)_{\langle\mathcal{S}_{i,j}\rangle}, \bm{x}_{\mathcal{S}_{i,j}} \rangle, \\
(E_3)_{i,j} & \colonequals  \gamma^2 \langle (\bm{\delta}_i)_{\langle\mathcal{S}_{i,j}\rangle}, (\bm{\delta}_j)_{\langle\mathcal{S}_{i,j}\rangle} \rangle.
\end{align*}
Write $q\colonequals\|\bm{x}_{\mathcal{S}_{i,j}}\|^2$ and $s\colonequals|\mathcal{S}_{i,j}|$.
Since $M_{i,j}=w^{(i)}w^{(j)}q$, a sign error requires $|E_{i,j}|\ge q$, where $E_{i,j}\colonequals(E_1)_{i,j}+(E_2)_{i,j}+(E_3)_{i,j}$.
The linear term $L\colonequals(E_1)_{i,j}+(E_2)_{i,j}$ is Gaussian with mean zero and variance $2\gamma^2q$, so
\[
\mathbb{P}[|L|\ge q/2]\le 2\exp\left(-\frac{q}{16\gamma^2}\right).
\]
For the remaining term $Q\colonequals(E_3)_{i,j}$, independence of the Gaussian coordinates gives
\[
\mathbb{E}\left[\exp\left(tQ/\gamma^2\right)\right]=(1-t^2)^{-s/2}
\text{ for } |t|<1.
\]
Chernoff's bound therefore gives, for an absolute constant $c>0$,
\[
\mathbb{P}[|Q|\ge q/2]\le 2\exp\left(-c\min\left\{\frac{q^2}{\gamma^4s},\frac{q}{\gamma^2}\right\}\right).
\]
Combining these two bounds by a union bound, and decreasing $c$ if needed, yields
\[
\mathbb{P}\left[\hat{w}^{(i,j)}\neq w^{(i)}w^{(j)}\right]
\le 4\exp\left(-c\min\left\{\frac{q}{\gamma^2},\frac{q^2}{\gamma^4s}\right\}\right).
\]
By \eqref{eq:gamma-lambda},
\[
\gamma^2=\frac{\lambda_1}{N(\rho_1^2-1)}.
\]
Since $\rho_1\to\rho_1^*>1$, $q\asymp\alpha^2\lambda_1$, and $s\asymp\alpha^2N$, both terms in the minimum are of order $\alpha^2N$.
Thus, for a constant $c_0>0$, uniformly over pairs satisfying the overlap assumptions,
\[
\mathbb{P}\left[\hat{w}^{(i,j)}\neq w^{(i)}w^{(j)}\right]
\le 4\exp(-c_0\alpha^2N).
\]

For a fixed spanning tree, a union bound over its $m-1$ edges gives
\[
\mathbb{P}[\text{at least one tree-edge error}]
\le 4(m-1)\exp(-c_0\alpha^2N).
\]
If the tree is selected using the observed data, we instead take a union bound over all pairs satisfying the overlap assumptions, replacing $m-1$ by at most $\binom{m}{2}$.
In either case, the bound tends to zero when $\alpha^2N=\omega(\log m+\log N)$.
When every tree-edge sign is correct, propagation along the tree recovers all source signs up to the common sign at the root, proving perfect alignment.

 \subsection{Proof of Theorem~\ref{thm:gsmmi_error}}

We decompose the pairwise correlation matrix as $\bm{A} = \bm{M} + \bm{E}$, where the signal matrix $\bm{M}$ has entries
$$
M_{i,j} = w^{(i)} w^{(j)} \|\bm{x}_{\mathcal{S}_{i,j}}\|^2 \quad \text{for }i\neq j, \qquad M_{i,i} = 0.
$$
Let $\bm{D}_{\bm{w}}  \colonequals  \operatorname{diag}(w^{(1)},\dots,w^{(m)})$, which is orthogonal. With the overlap matrix $\bm{S}$ from \eqref{eq:S}, the signal matrix factorizes as
$$
\bm{M} = \bm{D}_{\bm{w}} \bm{S} \bm{D}_{\bm{w}}.
$$
By the discussion preceding Assumption~\ref{asmp:overlap}, $\bm{S}$ has a simple leading eigenvalue $\lambda_1(\bm{S})>0$ with entrywise-positive eigenvector $\bm{s}$, and we write its eigendecomposition as
$$
\bm{S} = \lambda_1(\bm{S})\, \bm{s} \bm{s}^\top + \bm{R}
$$
with $\|\bm{R}\| = |\lambda_2(\bm{S})|$.
Using $\bm{D}_{\bm{w}} \bm{s} = \bm{w} \circ \bm{s}$, where $\circ$ is the Hadamard product, we obtain
$$
\bm{M} = \lambda_1(\bm{S})\, (\bm{w} \circ \bm{s})(\bm{w} \circ \bm{s})^\top + \bm{D}_{\bm{w}} \bm{R} \bm{D}_{\bm{w}}.
$$
Hence the leading eigenvector of $\bm{M}$ is
$$
\bm{v}  \colonequals  \bm{w} \circ \bm{s} \quad (\text{up to sign}),
$$
with $$\lambda_1(\bm{M}) = \lambda_1(\bm{S}),\quad
\lambda_1(\bm{M}) - \lambda_2(\bm{M}) = \lambda_1(\bm{S}) - \lambda_2(\bm{S}).
$$
{\color{black}By Assumption~\ref{asmp:overlap}(i)}, this gap satisfies $\lambda_1(\bm{M}) - \lambda_2(\bm{M}) \gtrsim \lambda_1 \alpha^2 m$, and the leading eigenvector $\bm{v} = \bm{w} \circ \bm{s}$ satisfies $\operatorname{sgn}(\bm{v}) = \bm{w}$ with $\min_i |v_i| = \min_i s_i \gtrsim 1/\sqrt{m}$ {\color{black}by Assumption~\ref{asmp:overlap}(ii)}.

We next bound the noise matrix $\bm{E} = \bm{A} - \bm{M}$, whose entries decompose as $E_{i,j} = (E_1)_{i,j} + (E_2)_{i,j} + (E_3)_{i,j}$, where
\begin{align*}
(E_1)_{i,j} &= \gamma w^{(i)} \langle \bm{x}_{\mathcal{S}_{i,j}}, (\bm{\delta}_j)_{\langle\mathcal{S}_{i,j}\rangle} \rangle, \\
(E_2)_{i,j} &= \gamma w^{(j)} \langle (\bm{\delta}_i)_{\langle\mathcal{S}_{i,j}\rangle}, \bm{x}_{\mathcal{S}_{i,j}} \rangle, \\
(E_3)_{i,j} &= \gamma^2 \langle (\bm{\delta}_i)_{\langle\mathcal{S}_{i,j}\rangle}, (\bm{\delta}_j)_{\langle\mathcal{S}_{i,j}\rangle} \rangle.
\end{align*}
Let $\bm{E}_k = ((E_k)_{i,j})$ for $k=1,2,3$.
Since the overlaps $\{\mathcal{S}_{i,j}\}$ are pairwise disjoint, distinct entries of $\bm{E}$ in the upper triangle depend on disjoint sets of noise coordinates and are therefore independent.
The entries $(E_1)_{i,j}$ and $(E_2)_{i,j}$ are centered Gaussian with variance $\gamma^2\|\bm{x}_{\mathcal{S}_{i,j}}\|^2$, while $(E_3)_{i,j}$ is a sum of $|\mathcal{S}_{i,j}|$ independent products of standard normals, which is sub-exponential with variance $\gamma^4|\mathcal{S}_{i,j}|$. {\color{black}Under the assumptions $\max_{i,j}\|\bm{x}_{\mathcal{S}_{i,j}}\|^2 \lesssim \lambda_1\alpha^2$ and $\max_{i,j}|\mathcal{S}_{i,j}| \lesssim \alpha^2 N$}, Corollary~4.4.8 in \cite{vershynin2018high} gives
$$
\|\bm{E}_1 + \bm{E}_2\| \lesssim \gamma \lambda_1^{1/2} \alpha (m^{1/2} + t)
$$
with probability at least $1 - 4\exp(-t^2)$; taking $t = C\log^{1/2} N$ for large enough $C>0$ yields
$$
\|\bm{E}_1 + \bm{E}_2\| \lesssim \gamma \lambda_1^{1/2} \alpha (m^{1/2} + \log^{1/2} N)
$$
with high probability. For the matrix $\bm{E}_3$ with sub-exponential entries, Theorem~1.1 in \cite{dai2024tail} gives
$$
\|\bm{E}_3\| \lesssim \gamma^2 (\alpha^2 N)^{1/2} (m^{1/2} + t) = \gamma^2 \alpha N^{1/2} (m^{1/2} + t)
$$
with probability at least $1 - O(\exp(-c\min\{t^2, t\}))$ for some constant $c>0$; taking $t = C\log N$ for large enough $C>0$ yields
$$
\|\bm{E}_3\| \lesssim \gamma^2 \alpha N^{1/2} (m^{1/2} + \log N)
$$
with high probability. Combining these,
$$
\|\bm{E}\| \lesssim \gamma \lambda_1^{1/2} \alpha (m^{1/2} + \log^{1/2} N) + \gamma^2 \alpha N^{1/2} (m^{1/2} + \log N)
$$
with high probability.

By Wedin's $\sin\Theta$ theorem, using $\lambda_1(\bm{M}) - \lambda_2(\bm{M}) \gtrsim \lambda_1 \alpha^2 m$, we have
$$
|\sin \Theta(\hat{\bm{v}}, \bm{v})|
\leq \frac{\|\bm{E}\|}{\lambda_1(\bm{M}) - \lambda_2(\bm{M})}
\lesssim \frac{\gamma}{\lambda_1^{1/2} \alpha m^{1/2}}
+ \frac{\gamma \log^{1/2} N}{\lambda_1^{1/2} \alpha m}
+ \frac{\gamma^2 N^{1/2}}{\lambda_1 \alpha m^{1/2}}
+ \frac{\gamma^2 N^{1/2} \log N}{\lambda_1 \alpha m}
$$
with high probability. Consequently,
$$
\min\{\|\hat{\bm{v}} - \bm{v}\|, \|\hat{\bm{v}} + \bm{v}\|\} \leq 2|\sin \Theta(\hat{\bm{v}}, \bm{v})|,
$$
and without loss of generality we take $\|\hat{\bm{v}} - \bm{v}\|$ to be the minimum.

Finally, we analyze sign recovery. Recall $\hat{\bm{w}} = \operatorname{sgn}(\hat{\bm{v}})$ and $\operatorname{sgn}(\bm{v}) = \bm{w}$. A coordinate $i$ is misclassified, i.e., $\hat{w}^{(i)} \neq w^{(i)}$, only if $|\hat{v}_i - v_i| \geq |v_i| = s_i$. Therefore,
$$
\#\{i : \hat{w}^{(i)} \neq w^{(i)}\}
\leq \frac{\|\hat{\bm{v}} - \bm{v}\|^2}{\min_{i} s_i^2}
\lesssim m\, \|\hat{\bm{v}} - \bm{v}\|^2,
$$
where the last step uses {\color{black}$\min_i s_i \gtrsim 1/\sqrt{m}$ from Assumption~\ref{asmp:overlap}(ii)}. Hence $\hat{\bm{w}} = \bm{w}$ provided $m\,\|\hat{\bm{v}} - \bm{v}\|^2 = o(1)$, which holds with high probability under
$$
\frac{\gamma}{\lambda_1^{1/2} \alpha}
+ \frac{\gamma \log^{1/2} N}{\lambda_1^{1/2} \alpha m^{1/2}}
+ \frac{\gamma^2 N^{1/2}}{\lambda_1 \alpha}
+ \frac{\gamma^2 N^{1/2} \log N}{\lambda_1 \alpha m^{1/2}} = o(1).
$$
{\color{black}Substituting $\gamma = \sigma\sqrt{\frac{\lambda_1}{\alpha\lambda_1^2 - \sigma^2 N}}$ from Theorem~\ref{thm:xi}, together with $\rho_1 = \sqrt{\frac{\alpha\lambda_1^2}{\sigma^2 N}} \to \rho_1^* > 1$}, this simplifies to
$$
\alpha^2 N = \omega(1)
\quad \text{and} \quad
m \alpha^2 N = \omega(\log^2 N).
$$

\subsection{Proof of Theorem~\ref{thm:cmmi_error_d2}}

Under model~\eqref{eq:simple_d2}, $\hat{\bm{X}}^{(i)} = \bm{X}_{\mathcal{U}_i}\bm{W}^{(i)\top} + \bm{\Delta}^{(i)}\bm{\Gamma}$, so the pairwise alignment $\hat{\bm{W}}^{(i,j)}$ estimates $\bm{W}^{(i)}\bm{W}^{(j)\top}$. Note that
$$
	\bm{W}^{(i)\top}\hat{\bm{W}}^{(i,j)}\bm{W}^{(j)} = \argmin_{\bm{O} \in \mathcal{O}_d}
	\left\| \hat{\bm{X}}^{(i)}_{\langle\mathcal{S}_{i,j}\rangle} \bm{W}^{(i)}\bm{O} - \hat{\bm{X}}^{(j)}_{\langle\mathcal{S}_{i,j}\rangle} \bm{W}^{(j)}\right\|_F,
$$
and denote
$$
\bm{F} \colonequals \bm{W}^{(i)\top}\hat{\bm{X}}^{(i)\top}_{\langle\mathcal{S}_{i,j}\rangle}\hat{\bm{X}}^{(j)}_{\langle\mathcal{S}_{i,j}\rangle}\bm{W}^{(j)}
-\bm{X}_{\mathcal{S}_{i,j}}^\top \bm{X}_{\mathcal{S}_{i,j}}.
$$
By perturbation bounds for polar decompositions, we have
\begin{equation}\label{eq:CMMI_d=2_W-WW}
\|\hat{\bm{W}}^{(i,j)}-\bm{W}^{(i)}\bm{W}^{(j)\top}\|
=\|\bm{W}^{(i)\top}\hat{\bm{W}}^{(i,j)}\bm{W}^{(j)}-\bm{I}\|
\leq \frac{2\|\bm{F}\|}{\sigma_{\min}(\bm{X}_{\mathcal{S}_{i,j}}^\top \bm{X}_{\mathcal{S}_{i,j}})},
\end{equation}
provided $\|\bm{F}\|<\sigma_{\min}(\bm{X}_{\mathcal{S}_{i,j}}^\top \bm{X}_{\mathcal{S}_{i,j}})$. {\color{black}Since the eigenvalues of $\bm{U}_{\mathcal{S}_{i,j}}^\top \bm{U}_{\mathcal{S}_{i,j}}$ are of order $\alpha^2$}, we have $\bm{X}_{\mathcal{S}_{i,j}}^\top \bm{X}_{\mathcal{S}_{i,j}} = \bm{\Lambda}^{1/2}\bm{U}_{\mathcal{S}_{i,j}}^\top \bm{U}_{\mathcal{S}_{i,j}}\bm{\Lambda}^{1/2}$ with $\sigma_{\min}(\bm{X}_{\mathcal{S}_{i,j}}^\top \bm{X}_{\mathcal{S}_{i,j}}) \asymp \alpha^2\lambda_2 \asymp \alpha^2\lambda$, where $\lambda \colonequals \lambda_1\asymp \lambda_2$.

We now bound $\|\bm{F}\|$. Substituting $\hat{\bm{X}}^{(i)} = \bm{X}_{\mathcal{U}_i}\bm{W}^{(i)\top}+\bm{\Delta}^{(i)}\bm{\Gamma}$, we have
$$
\begin{aligned}
	\bm{F}
	&=\bm{W}^{(i)\top}
(\bm{X}_{\mathcal{U}_i} \bm{W}^{(i)\top}+ \bm{\Delta}^{(i)} \bm{\Gamma})^\top_{\langle\mathcal{S}_{i,j}\rangle}
(\bm{X}_{\mathcal{U}_i} \bm{W}^{(j)\top}+ \bm{\Delta}^{(j)} \bm{\Gamma})_{\langle\mathcal{S}_{i,j}\rangle}
\bm{W}^{(j)}
-\bm{X}_{\mathcal{S}_{i,j}}^\top \bm{X}_{\mathcal{S}_{i,j}}\\
&=\underbrace{\bm{X}_{\mathcal{S}_{i,j}}^\top
\bm{\Delta}^{(j)}_{\langle\mathcal{S}_{i,j}\rangle}
 \bm{\Gamma}
\bm{W}^{(j)}}_{\bm{F}_1}
+ \underbrace{\bm{W}^{(i)\top}\bm{\Gamma}\bm{\Delta}^{(i)\top} _{\langle\mathcal{S}_{i,j}\rangle}\bm{X}_{\mathcal{S}_{i,j}}}_{\bm{F}_2}
+\underbrace{\bm{W}^{(i)\top}\bm{\Gamma}\bm{\Delta}^{(i)\top} _{\langle\mathcal{S}_{i,j}\rangle}
\bm{\Delta}^{(j)}_{\langle\mathcal{S}_{i,j}\rangle}
 \bm{\Gamma}
\bm{W}^{(j)}}_{\bm{F}_3}.
\end{aligned}
$$
We bound each term separately.

For $\bm{F}_1$, since $\bm{X} = \bm{U}\bm{\Lambda}^{1/2}$, we have $\bm{F}_1=\bm{\Lambda}^{1/2} \bm{U}_{\mathcal{S}_{i,j}}^\top \bm{\Delta}^{(j)}_{\langle\mathcal{S}_{i,j}\rangle} \bm{\Gamma}\bm{W}^{(j)}$. Consider
$$
\bm{G}  \colonequals  \bm{U}_{\mathcal{S}_{i,j}}^\top \bm{\Delta}^{(j)}_{\langle\mathcal{S}_{i,j}\rangle}
=\sum_{s\in\mathcal{S}_{i,j}}\bm{u}_{s} \bm{\delta}^{(j)\top}_{s}
=:\sum_{s\in\mathcal{S}_{i,j}}\bm{M}_s,
$$
where $\bm{u}_s$ and $\bm{\delta}^{(j)}_s$ are the $s$th rows of $\bm{U}$ and $\bm{\Delta}^{(j)}$, respectively. Recall that $\bm{\Delta}^{(j)}$ has i.i.d. $\mathcal{N}(0,1)$ entries.
The $\{\bm{M}_s\}$ are independent mean-zero random matrices. We have
$$
\begin{aligned}
	\Big\|\sum_{s\in\mathcal{S}_{i,j}}\bm{M}_s^\top\bm{M}_s\Big\|
	&=\Big\|\sum_{s\in\mathcal{S}_{i,j}}\bm{\delta}^{(j)}_{s}\bm{u}_{s}^\top \bm{u}_{s} \bm{\delta}^{(j)\top}_{s}\Big\|=	\Big\|\sum_{s\in\mathcal{S}_{i,j}}\|\bm{u}_{s}\|^2 \bm{\delta}^{(j)}_{s}\bm{\delta}^{(j)\top}_{s}\Big\|\\
	&\lesssim N^{-1}	\Big\|\sum_{s\in\mathcal{S}_{i,j}}\bm{\delta}^{(j)}_{s}\bm{\delta}^{(j)\top}_{s}\Big\|
	\\&\lesssim N^{-1} \cdot \alpha^2 N
	= \alpha^2
\end{aligned}
$$
with high probability, where the first inequality uses {\color{black}$\|\bm{U}\|_{2\to\infty}\lesssim d^{1/2}N^{-1/2}$} and the second inequality uses the property for Wishart matrices with {\color{black}$|\mathcal{S}_{i,j}|\asymp\alpha^2 N$}. We also have
$$
\begin{aligned}
	\Big\|\sum_{s\in\mathcal{S}_{i,j}}\bm{M}_s\bm{M}_s^\top\Big\|
	&=\Big\|\sum_{s\in\mathcal{S}_{i,j}}\bm{u}_{s} \bm{\delta}^{(j)\top}_{s}\bm{\delta}^{(j)}_{s}\bm{u}_{s}^\top \Big\|
	=\Big\|\sum_{s\in\mathcal{S}_{i,j}}\|\bm{\delta}^{(j)}_{s}\|^2 \bm{u}_{s}\bm{u}_{s}^\top\Big\|\\
	&\leq \max_{s\in\mathcal{S}_{i,j}}\|\bm{\delta}^{(j)}_{s}\|^2 \cdot \left\|\sum_{s\in\mathcal{S}_{i,j}}\bm{u}_{s}\bm{u}_{s}^\top \right\|\\
	&\lesssim \log N\Big\|\sum_{s\in\mathcal{S}_{i,j}}\bm{u}_{s}\bm{u}_{s}^\top \Big\|
	= \log N\Big\|\muu^\top_{\mathcal{S}_{i,j}}\muu_{\mathcal{S}_{i,j}}\Big\|
	\\&\lesssim \alpha^2 \log N
\end{aligned}
$$
with high probability, where the second inequality uses that $\|\bm{\delta}^{(j)}_{s}\|^2 \sim \chi^2_d$ has $\|\bm{\delta}^{(j)}_{s}\|^2\lesssim d+\log N$ with high probability, and the last inequality uses {\color{black}$\|\bm{U}_{\mathcal{S}_{i,j}}^\top\bm{U}_{\mathcal{S}_{i,j}}\|\asymp\alpha^2$}. By Matrix Hoeﬀding for rectangular matrices, we have
$$
\|\bm{G}\|\lesssim \alpha\log N
$$
with high probability. Note that {\color{black}under the condition $\rho_k=\sqrt{\frac{\alpha\lambda_k^2}{\sigma^2 N}}\to\rho_k^*>1$}, we have  $\lambda\asymp\frac{\sigma N^{1/2}}{\alpha^{1/2}}$ and $\gamma  \colonequals  \max_k\gamma_k \asymp \frac{\sigma^{1/2}}{\alpha^{1/4}N^{1/4}}$.
Therefore
$$
\begin{aligned}
	\|\bm{F}_1\|
	\lesssim \|\bm{\Lambda}\|^{1/2}\cdot \|\mg\| \cdot \|\bm{\Gamma}\|\lesssim \lambda^{1/2} \cdot \alpha\log N \cdot \gamma
	\lesssim \frac{\sigma^{1/2} N^{1/4}}{\alpha^{1/4}}
	\cdot \alpha\log N \cdot \frac{\sigma^{1/2}}{\alpha^{1/4}N^{1/4}}\lesssim \sigma \alpha^{1/2}\log N
\end{aligned}
$$
with high probability. By a symmetric argument, $\|\bm{F}_2\|\lesssim \sigma \alpha^{1/2}\log N$ with high probability.

For $\bm{F}_3$, the entries of $\bm{\Delta}^{(i)\top}_{\langle\mathcal{S}_{i,j}\rangle}\bm{\Delta}^{(j)}_{\langle\mathcal{S}_{i,j}\rangle}\in\mathbb{R}^{2\times 2}$ are sums of $|\mathcal{S}_{i,j}|\asymp\alpha^2 N$ independent mean-zero sub-exponential variables with $\psi_1$-norm $\lesssim 1$. By Theorem~B in \cite{pinelis2022improved}, {\color{black}provided $\alpha^2N=\Omega(\log N)$}, we have
$$
\big|[\bm{\Delta}^{(i)\top}_{\langle\mathcal{S}_{i,j}\rangle}\bm{\Delta}^{(j)}_{\langle\mathcal{S}_{i,j}\rangle}]_{k,\ell}\big|
\lesssim (\alpha^2N\cdot 1^2)^{1/2}\log^{1/2} N + \log N
\lesssim \alpha N^{1/2} \log^{1/2} N
$$
with high probability, so $\|\bm{\Delta}^{(i)\top}_{\langle\mathcal{S}_{i,j}\rangle}\bm{\Delta}^{(j)}_{\langle\mathcal{S}_{i,j}\rangle}\|\lesssim \alpha N^{1/2} \log^{1/2} N$. Hence
$$
\|\bm{F}_3\| \lesssim \gamma^2 \cdot \alpha N^{1/2} \log^{1/2} N
\lesssim \frac{\sigma}{\alpha^{1/2}N^{1/2}} \cdot \alpha N^{1/2} \log^{1/2} N
= \sigma \alpha^{1/2} \log^{1/2} N
$$
with high probability. Combining the bounds for $\bm{F}_1, \bm{F}_2, \bm{F}_3$, we have
\begin{equation}\label{eq:CMMI_d=2_F}
	\|\bm{F}\|\lesssim \sigma\alpha^{1/2}\log N
\end{equation}
with high probability.

Thus {\color{black}when $\alpha^2 N=\omega(\log^2 N)$}, we have $\|\bm{F}\| = o(\sigma_{\min}(\bm{X}_{\mathcal{S}_{i,j}}^\top \bm{X}_{\mathcal{S}_{i,j}}))$, so the condition of \eqref{eq:CMMI_d=2_W-WW} holds. Combining \eqref{eq:CMMI_d=2_W-WW}, \eqref{eq:CMMI_d=2_F}, and $\sigma_{\min}(\bm{X}_{\mathcal{S}_{i,j}}^\top \bm{X}_{\mathcal{S}_{i,j}})\asymp\alpha^2\lambda$, we have
$$
\|\hat{\bm{W}}^{(i,j)}-\bm{W}^{(i)}\bm{W}^{(j)\top}\|
\lesssim \frac{\sigma\alpha^{1/2}\log N}{\alpha^2\lambda}
\lesssim \frac{\sigma\alpha^{1/2}\log N}{\alpha^2 \cdot \frac{\sigma N^{1/2}}{\alpha^{1/2}}}
= \frac{\log N}{\alpha N^{1/2}}
$$
with high probability.

Finally, in CMMI, after constructing a spanning tree with root at submatrix $1$ (which amounts to fixing $\hat{\bm{W}}^{(1)} = \bm{I}$), the alignment matrix $\hat{\bm{W}}^{(i)}$ for each submatrix $i$ is computed by composing the pairwise alignment matrices along the unique path connecting the root to node $i$. Since each pairwise estimate $\hat{\bm{W}}^{(i,j)}$ approximates $\bm{W}^{(i)}\bm{W}^{(j)\top}$, composing along the path from the root telescopes to $\hat{\bm{W}}^{(i)} \approx \bm{W}^{(i)}\bm{W}^{(1)\top} = \bm{W}^{(i)}$, where we fix the reference coordinate system by setting $\bm{W}^{(1)} = \bm{I}$.
If this path has length $D_i$, then by the analysis of Lemma~D.3 from \citet{zheng2026chain}, the accumulated error satisfies
$$
\|\hat{\bm{W}}^{(i)} - \bm{W}^{(i)}\|_F \lesssim D_i \cdot \frac{\log N}{\alpha N^{1/2}}
$$
with high probability.

\subsection{Proof of Theorem~\ref{thm:gsmmi_error_d2}}

We follow the strategy of the proof of Theorem~\ref{thm:gsmmi_error}, adapted to the complex model~\eqref{eq:complex_model}. Throughout, write $\lambda  \colonequals  \lambda_1+\lambda_2 = \|\bm{z}\|^2$ and $\gamma^2  \colonequals  \gamma_1^2 + \gamma_2^2$, so that $\lambda_1\asymp\lambda_2\asymp\lambda$ and $\gamma_1\asymp\gamma_2\asymp\gamma$. The condition $\rho_k=\sqrt{\frac{\alpha\lambda_k^2}{\sigma^2 N}}\to\rho_k^*>1$ gives $\lambda\asymp\frac{\sigma N^{1/2}}{\alpha^{1/2}}$ and $\gamma  \colonequals  \max_k\gamma_k \asymp \frac{\sigma^{1/2}}{\alpha^{1/4} N^{1/4}}$.

We decompose $\bm{A} = \bm{M} + \bm{E}$, where the signal matrix $\bm{M}$ has entries
$$
M_{i,j} = \|\bm{z}_{\mathcal{S}_{i,j}}\|^2\, w^{(i)} \overline{w^{(j)}} \quad \text{ for }i\neq j, \qquad M_{i,i} = 0.
$$
Indeed, substituting $\hat{\bm{z}}^{(i)}=\overline{w^{(i)}}\bm{z}_{\mathcal{U}_i}+\bm{\epsilon}_i$ into $A_{i,j} = (\hat{\bm{z}}^{(i)}_{\langle\mathcal{S}_{i,j}\rangle})^\dagger \hat{\bm{z}}^{(j)}_{\langle\mathcal{S}_{i,j}\rangle}$ and keeping the signal-signal term gives $w^{(i)}\overline{w^{(j)}}\,\|\bm{z}_{\mathcal{S}_{i,j}}\|^2$. With the overlap matrix $\bm{S}$ from \eqref{eq:S_d2} and the phase matrix $\bm{D}_{\bm{w}}  \colonequals  \operatorname{diag}(w^{(1)}, \dots, w^{(m)})$, the signal matrix factorizes as
$$
\bm{M} = \bm{D}_{\bm{w}} \bm{S} \bm{D}_{\bm{w}}^\dagger.
$$
As in the $d=1$ case, $\bm{S}$ is real, nonnegative, and symmetric, so again by the Perron--Frobenius theorem its leading eigenvalue $\lambda_1(\bm{S})>0$ is simple with entrywise-positive eigenvector $\bm{s}$. Let $\bm{w}  \colonequals  (w^{(1)},\dots,w^{(m)})^\top$ collect the true phases, so that $\bm{D}_{\bm{w}} = \operatorname{diag}(\bm{w})$ and $\bm{D}_{\bm{w}}\bm{s} = \bm{w}\circ\bm{s}$. Writing $\bm{S} = \lambda_1(\bm{S})\, \bm{s}\bm{s}^\top + \bm{R}$ with $\|\bm{R}\| = |\lambda_2(\bm{S})|$, the leading eigenvector of $\bm{M} = \bm{D}_{\bm{w}}\bm{S}\bm{D}_{\bm{w}}^\dagger$ is
$$
\bm{v}  \colonequals  \bm{w} \circ \bm{s} \quad (\text{up to a global phase}),
$$
with $\lambda_1(\bm{M}) = \lambda_1(\bm{S})$ and $\lambda_1(\bm{M}) - \lambda_2(\bm{M}) = \lambda_1(\bm{S}) - \lambda_2(\bm{S})$. {\color{black}By Assumption~\ref{asmp:overlap}(i)}, $\lambda_1(\bm{M}) - \lambda_2(\bm{M}) \gtrsim \lambda\alpha^2 m$. For each $i\in[m]$, since $\bm{s}$ is entrywise positive and $|w^{(i)}|=1$, we have $|\bm{v}_i| = s_i$ and $\bm{v}_i/|\bm{v}_i| = w^{(i)}$, and by {\color{black}Assumption~\ref{asmp:overlap}(ii)}, we further have $|\bm{v}_i| = s_i \gtrsim 1/\sqrt{m}$.

For the noise matrix $\bm{E} = \bm{A} - \bm{M}$, we decompose $E_{i,j} = (E_1)_{i,j} + (E_2)_{i,j} + (E_3)_{i,j}$ as in the $d=1$ case. 
Since the overlaps $\{\mathcal{S}_{i,j}\}$ are pairwise disjoint, the entries of $\bm{E}$ are independent, and using $\max_{i,j}\|\bm{z}_{\mathcal{S}_{i,j}}\|^2 \lesssim \lambda\alpha^2$ and $\max_{i,j}|\mathcal{S}_{i,j}| \lesssim \alpha^2 N$, the same argument as for $d=1$ yields
$$
\|\bm{E}\| \lesssim \gamma \lambda^{1/2} \alpha (m^{1/2} + \log^{1/2} N) + \gamma^2 \alpha N^{1/2} (m^{1/2} + \log N)
$$
with high probability. By Wedin's $\sin\Theta$ theorem with $\lambda_1(\bm{M}) - \lambda_2(\bm{M})\gtrsim \lambda\alpha^2 m$, and substituting $\gamma\asymp\frac{\sigma^{1/2}}{\alpha^{1/4}N^{1/4}}$ and $\lambda\asymp\frac{\sigma N^{1/2}}{\alpha^{1/2}}$, we have
$$
\|\hat{\bm{v}} - \bm{v}\| \leq \sqrt{2}\, |\sin \Theta(\hat{\bm{v}}, \bm{v})|
\lesssim \frac{1}{\alpha N^{1/2} m^{1/2}} + \frac{\log N}{\alpha N^{1/2} m}
$$
with high probability, where the global phase ambiguity is fixed by setting $\bm{W}^{(1)} = \hat{\bm{W}}^{(1)} = \bm{I}$.

We first show that for each $i\in[m]$, $\hat{\bm{v}}_i$ does not vanish. By the reverse triangle inequality,
$$
\bigl|\,|\hat{\bm{v}}_i| - |\bm{v}_i|\,\bigr| \leq \|\hat{\bm{v}} - \bm{v}\|
\lesssim \frac{1}{\alpha N^{1/2} m^{1/2}} + \frac{\log N}{\alpha N^{1/2} m}
$$
with high probability.
Recall $|\bm{v}_i| \gtrsim 1/\sqrt{m}$. Thus {\color{black}when $\alpha^2 N = \omega(1)$ and $m\alpha^2 N = \omega(\log^2 N)$}, the right-hand side is $o(1/\sqrt{m})$, and
$$
|\hat{\bm{v}}_i| \geq |\bm{v}_i| - o(1/\sqrt{m}) \gtrsim 1/\sqrt{m}
$$
with high probability. Since $\hat{w}^{(i)} = \hat{\bm{v}}_i/|\hat{\bm{v}}_i|$ estimates $w^{(i)}=\bm{v}_i/|\bm{v}_i|$, with $|\hat{\bm{v}}_i|, |\bm{v}_i| \gtrsim 1/\sqrt{m}$,
$$
|w^{(i)} - \hat{w}^{(i)}|
\lesssim \frac{|\hat{\bm{v}}_i - \bm{v}_i|}{\min\{|\hat{\bm{v}}_i|, |\bm{v}_i|\}}
\lesssim \sqrt{m}\,\|\hat{\bm{v}} - \bm{v}\|
\lesssim \frac{1}{\alpha N^{1/2}} + \frac{\log N}{\alpha N^{1/2} m^{1/2}}
$$
with high probability. Since $\hat{\bm{W}}^{(i)}\bm{c} = \hat{w}^{(i)}\bm{c}$ and $\bm{W}^{(i)}\bm{c} = w^{(i)}\bm{c}$, we have $\|\hat{\bm{W}}^{(i)} - \bm{W}^{(i)}\|_F = \sqrt{2}\,|\hat{w}^{(i)} - w^{(i)}|$, so
$$
\|\hat{\bm{W}}^{(i)} - \bm{W}^{(i)}\|_F
\lesssim \frac{1}{\alpha N^{1/2}} + \frac{\log N}{\alpha N^{1/2} m^{1/2}}
$$
with high probability.

\subsection{Preliminaries on low-degree polynomial algorithms}

We recall the notions concerning low-degree polynomials used in the proofs below.
See the surveys~\cite{KWB-2022-LowDegreeNotes,Wein-2025-LowDegreeSurvey} for further discussion.

Let $\mathbb{P}_N$ and $\mathbb{Q}_N$ be probability measures on $\RR^{M_N}$.
A sequence of functions $f_N:\RR^{M_N}\to\RR$ with finite second moments under both measures achieves \emph{strong separation} if
\begin{equation} \label{eq:strong-sep}
    \sqrt{\operatorname{Var}_{\mathbb{P}_N}[f_N]}
    +\sqrt{\operatorname{Var}_{\mathbb{Q}_N}[f_N]}
    =o\!\left(\left|\E_{\mathbb{P}_N}[f_N]-\E_{\mathbb{Q}_N}[f_N]\right|\right)
\end{equation}
as $N\to\infty$, as in Definition~1.2 of~\cite{de2026computational}.
By Chebyshev's inequality, thresholding $f_N$ halfway between its two means, with the appropriate direction of comparison, achieves strong detection in the sense of Section~\ref{sec:low_syn}.
Conversely, tests that achieve strong detection themselves achieve strong separation.
Thus, the two notions are equivalent when arbitrary functions are allowed.

For probability measures $\mathbb{P}\ll\mathbb{Q}$ on $\RR^M$ with finite moments of all orders, let $\RR[\bm y]_{\leq D}$ denote the space of polynomials of total degree at most $D$ in the observation $\bm y$.
Following Definition~3.23 of~\cite{de2026computational}, define the \emph{degree-$D$ chi-squared divergence} by
\[
    \chi^2_{\leq D}(\mathbb{P}\mid\mathbb{Q})
    \colonequals
    \sup_{\substack{f\in\RR[\bm y]_{\leq D}\\\E_{\mathbb{Q}}[f^2]>0}}
    \frac{(\E_{\mathbb{P}}[f])^2}{\E_{\mathbb{Q}}[f^2]}-1.
\]
Writing $L\colonequals d\mathbb{P}/d\mathbb{Q}$, the full chi-squared divergence is
\[
    \chi^2(\mathbb{P}\mid\mathbb{Q})
    \colonequals \E_{\mathbb{Q}}[(L-1)^2].
\]
When $L\in L^2(\mathbb{Q})$, the low-degree divergence equals $\|L^{\leq D}-1\|_{L^2(\mathbb{Q})}^2$, where $L^{\leq D}$ is the orthogonal projection of $L$ onto $\RR[\bm y]_{\leq D}$.
In particular, $\chi^2_{\leq D}(\mathbb{P}\mid\mathbb{Q})\leq\chi^2(\mathbb{P}\mid\mathbb{Q})$.
If $h_0\equiv1,h_1,\ldots,h_K$ is an orthonormal basis for this polynomial space in $L^2(\mathbb{Q})$, then
\[
    1+\chi^2_{\leq D}(\mathbb{P}\mid\mathbb{Q})
    =\sum_{j=0}^K(\E_{\mathbb{P}}[h_j])^2,
\]
as in Proposition~8.3 of~\cite{Wein-2025-LowDegreeSurvey}.
For a standard Gaussian null, we use normalized multivariate Hermite polynomials as this basis.

Centering $f$ in the variational definition gives, for every polynomial of degree at most $D$,
\[
    \left|\E_{\mathbb{P}}[f]-\E_{\mathbb{Q}}[f]\right|^2
    \leq \chi^2_{\leq D}(\mathbb{P}\mid\mathbb{Q})
    \operatorname{Var}_{\mathbb{Q}}[f].
\]
Consequently, the following implications hold as $N\to\infty$:
\begin{enumerate}
    \item If $\chi^2_{\leq D(N)}(\mathbb{P}_N\mid\mathbb{Q}_N)=O(1)$, then no sequence of polynomials of degree at most $D(N)$ achieves strong separation.
    \item If $\chi^2(\mathbb{P}_N\mid\mathbb{Q}_N)=O(1)$, then no sequence of functions achieves strong separation, and no sequence of tests achieves strong detection.
    \item If $\chi^2(\mathbb{P}_N\mid\mathbb{Q}_N)=o(1)$, then no sequence of tests achieves weak detection.
\end{enumerate}
The first two statements are Proposition~3.24 of~\cite{de2026computational}.
The third follows from $\|\mathbb{P}_N-\mathbb{Q}_N\|_{\mathrm{TV}}\leq\frac12\sqrt{\chi^2(\mathbb{P}_N\mid\mathbb{Q}_N)}$.
A bound on the low-degree divergence alone does not in general rule out detection by thresholding low-degree polynomials, as discussed in Section~4.3 of~\cite{Wein-2025-LowDegreeSurvey}.
Likewise, divergence of $\chi^2_{\leq D(N)}$ does not by itself imply strong separation, since it does not control the variance under $\mathbb{P}_N$.

\subsection{Proof of Theorem~\ref{thm:lb-sync}}

\begin{lemma}
\label{lem:symmetric-bound}
    Let $\bm{\tau}=(\tau_1,\ldots,\tau_N)$ have independent Rademacher
    entries. There exist universal constants $C_0,c_0>0$ such that,
    for every real symmetric matrix $\bm{R}\in\mathbb{R}^{N\times N}$,
    \[
        \E_{\bm{\tau}}\left[\exp\left(
            \frac{s}{2}\left(\bm{\tau}^\top \bm{R}\bm{\tau}-\operatorname{tr}(\bm{R})\right)
        \right)\right]
        \leq
        \exp\left(C_0s^2\|\bm{R}\|_F^2\right)
    \]
    whenever $s\in\mathbb{R}$ satisfies $|s|\cdot\|\bm{R}\|\leq c_0$.
\end{lemma}

\begin{proof}
    Since $\tau_i^2=1$ almost surely,
    \[
        \bm{\tau}^\top \bm{R}\bm{\tau}-\operatorname{tr}(\bm{R})
        =\sum_{i\neq j}R_{ij}\tau_i\tau_j.
    \]
    The moment-generating-function bound for the off-diagonal
    quadratic form in Step~2 of the proof of Theorem~6.2.1
    in~\cite{vershynin2018high} therefore gives the desired estimate,
    after adjusting the universal constants.
\end{proof}

\begin{proof}[Proof of Theorem~\ref{thm:lb-sync}]
We write $\bm{B}=(B_{ij})_{i,j\in[m]}$ for the overlap matrix in the statement and set $\beta\colonequals(1+\gamma^2)^{-1/2}$, so that $\eta=\beta^4\|\bm{B}\|_F$.
We standardize each observation in \eqref{eq:simple} by dividing by $\sqrt{1+\gamma^2}$.
In the calculations below, $\mathbb{Q}$ denotes the independent Gaussian null with the same one-dimensional marginals as the alternative, so the standardized observations are independent standard Gaussians under this null.
We retain the notation $\mathbb{P}_{\gamma}$ and $\mathbb{Q}$ after this invertible rescaling, which preserves polynomial degree and the chi-squared divergences.

We first derive an expression for the low-degree chi-squared divergence that will also be used in the proof of Theorem~\ref{thm:ub-sync}.
Let
\[
    E \colonequals  \bigl\{\{i,j\} : 1 \leq i < j \leq m,\ B_{ij}>0\bigr\}.
\]
For each edge $e=\{i,j\}\in E$, write $b_e\colonequals B_{ij}$ and index
the shared coordinates by $t=1,\ldots,b_e$. Denote the corresponding
pair of standardized observations by $(A_{e,t},C_{e,t})$. Under the alternative,
conditional on $\bm{w}=(w^{(1)},\ldots,w^{(m)})$,
\[
    (A_{e,t},C_{e,t}) \mid \bm{w}
    \sim \mathcal{N}\left(
        \bm{0},
        \begin{pmatrix}
            1 & \beta^2w^{(i)}w^{(j)} \\
            \beta^2w^{(i)}w^{(j)} & 1
        \end{pmatrix}
    \right).
\]
By the disjointness assumption on the overlaps, these pairs are
independent conditional on $\bm{w}$.
Coordinates observed by only one source are independent standard Gaussians under both distributions and do not contribute nonconstant Hermite coefficients.

Let $h_k\colonequals H_k/\sqrt{k!}$, where $H_k$ is the $k$th Hermite polynomial. For jointly standard Gaussian random variables
$A$ and $C$ with correlation $r$, we have
\[
    \E[h_p(A)h_q(C)] = \bm{1}_{\{p=q\}}r^p.
\]
Consequently, a nonzero Hermite coefficient must assign the same
degree $k_{e,t}$ to both entries of each pair $(A_{e,t},C_{e,t})$.
The resulting product polynomial has total degree
$2\sum_{e,t}k_{e,t}$. Conditional independence gives
\begin{align*}
    \E\left[
        \prod_{e\in E}\prod_{t=1}^{b_e}
        h_{k_{e,t}}(A_{e,t})h_{k_{e,t}}(C_{e,t})
        \,\middle|\,\bm{w}
    \right] =
    \beta^{2\sum_{e,t}k_{e,t}}
    \prod_{i=1}^m
    \bigl(w^{(i)}\bigr)^{\sum_{e:\,i\in e}\sum_{t=1}^{b_e}k_{e,t}}.
\end{align*}
Averaging over the independent Rademacher signs therefore yields
\begin{align*}
    \E\left[
        \prod_{e\in E}\prod_{t=1}^{b_e}
        h_{k_{e,t}}(A_{e,t})h_{k_{e,t}}(C_{e,t})
    \right]  =
    \beta^{2\sum_{e,t}k_{e,t}}
    \bm{1}_{\left\{
        \sum_{e:\,i\in e}\sum_{t=1}^{b_e}k_{e,t}
        \text{ is even for every }i
    \right\}}.
\end{align*}

For each $e\in E$, set $k_e\colonequals \sum_{t=1}^{b_e}k_{e,t}$.
There are $\binom{b_e+k_e-1}{k_e}$ ways to distribute $k_e$
among the $b_e$ shared coordinates. Squaring the Hermite
coefficients and summing over all product polynomials of total
degree at most $D$, we obtain
\begin{equation}\label{eq:sync-hermite-expansion}
    1+\chi^2_{\leq D}(\mathbb{P}_{\gamma}\mid\mathbb{Q})
    =
    \sum_{\substack{
        \bm{k}\in\mathbb{N}_0^E\\
        2\sum_{e\in E}k_e\leq D
    }}
    \bm{1}_{\left\{
        \sum_{e:\,i\in e}k_e
        \text{ is even for every }i
    \right\}}
    \prod_{e\in E}
    \binom{b_e+k_e-1}{k_e}\beta^{4k_e}.
\end{equation}

\paragraph{Part 1: Vanishing signal regime}
We next consider the full chi-squared divergence. Since $\gamma>0$, we have $0<\beta<1$.
Let
$\bm{\tau}=(\tau_1,\ldots,\tau_m)$ have independent
Rademacher entries. Using the parity representation and the negative binomial series expansion, we obtain
\begin{align*}
    1+\chi^2(\mathbb{P}_{\gamma}\mid\mathbb{Q})
    &=
    \sum_{\bm{k}\in\mathbb{N}_0^E}
    \bm{1}_{\left\{
        \sum_{e:\,i\in e}k_e
        \text{ is even for every }i
    \right\}}
    \prod_{e\in E}
    \binom{b_e+k_e-1}{k_e}\beta^{4k_e} \\
    &=
    \E_{\bm{\tau}}\left[
        \prod_{e=\{i,j\}\in E}
        \sum_{k_e=0}^{\infty}
        \binom{b_e+k_e-1}{k_e}
        (\beta^4\tau_i\tau_j)^{k_e}
    \right] \\
    &=
    \E_{\bm{\tau}}\left[
        \prod_{i<j}
        (1-\beta^4\tau_i\tau_j)^{-B_{ij}}
    \right].
\end{align*}
The interchange of summation and expectation is justified by
absolute convergence.

Define
\[
    t\colonequals \operatorname{arctanh}(\beta^4)
    =\frac12\log\left(1+\frac{2}{\gamma^2(\gamma^2+2)}\right).
\]
For $s\in\{-1,1\}$,
\[
    -\log(1-\beta^4s)
    =-\frac12\log(1-\beta^8)+ts.
\]
Since $\bm{B}$ is symmetric with zero diagonal, it follows that
\[
    1+\chi^2(\mathbb{P}_{\gamma}\mid\mathbb{Q})
    =
    (1-\beta^8)^{-\frac12\sum_{i<j}B_{ij}}
    \E_{\bm{\tau}}\left[
        \exp\left(\frac{t}{2}\bm{\tau}^{\top}
        \bm{B}\bm{\tau}\right)
    \right].
\]

Our theorem and its conclusions are motivated by the expression
$\E_{\bm{\tau}}\left[\exp\left(\frac{t}{2}\bm{\tau}^{\top}\bm{B}\bm{\tau}\right)\right]$,
as related exponential moments of quadratic forms exhibit sharp
phase transitions, as demonstrated in~\cite{de2026computational}.

Let $\bm{g}\sim\mathcal{N}(\bm{0},\bm{I}_m)$. Because $\bm{B}$ is entrywise
nonnegative, the expansion of $(\bm{z}^\top\bm{B}\bm{z})^n$ has nonnegative
coefficients. Moreover, the moments of a standard Gaussian
dominate the corresponding Rademacher moments. Thus,
\[
    \E\left[(\bm{\tau}^{\top}\bm{B}\bm{\tau})^n\right]
    \leq
    \E\left[(\bm{g}^\top \bm{B}\bm{g})^n\right]
    \qquad\text{for every }n\geq0.
\]
Since $\bm{B}$ is symmetric and entrywise nonnegative,
$\|\bm{B}\|=\lambda_1(\bm{B})$. Consequently, if
$r\colonequals t\lambda_1(\bm{B})<1$, the exponential series are absolutely
integrable, and
\begin{align*}
    \E\left[\exp\left(\frac{t}{2}\bm{\tau}^{\top}
    \bm{B}\bm{\tau}\right)\right] \leq
    \E\left[\exp\left(\frac{t}{2}\bm{g}^\top \bm{B}\bm{g}\right)\right] =\det(\bm{I}_m-t\bm{B})^{-1/2}.
\end{align*}
Therefore,
\[
    \log\left(1+\chi^2(\mathbb{P}_{\gamma}\mid\mathbb{Q})\right)
    \leq
    -\frac12\log(1-\beta^8)\sum_{i<j}B_{ij}
    -\frac12\log\det(\bm{I}_m-t\bm{B}).
\]

We bound these two terms separately. Since the entries of $\bm{B}$
are nonnegative integers,
\[
    \sum_{i<j}B_{ij}
    \leq \sum_{i<j}B_{ij}^2
    =\frac12\|\bm{B}\|_F^2.
\]
Together with $-\log(1-u)\leq u/(1-u)$ for $0\leq u<1$,
this gives
\[
    -\frac12\log(1-\beta^8)\sum_{i<j}B_{ij}
    \leq
    \frac{\beta^8\|\bm{B}\|_F^2}{4(1-\beta^8)}.
\]
For the determinant term, use $\operatorname{tr}(\bm{B})=0$ and
$|t\lambda_j(\bm{B})|\leq r<1$ to obtain
\begin{align*}
    -\frac12\log\det(\bm{I}_m-t\bm{B})
    &=
    \frac12\sum_j
    \left[-\log(1-t\lambda_j(\bm{B}))-t\lambda_j(\bm{B})\right] \\
    &=
    \frac12\sum_j\sum_{n=2}^{\infty}
    \frac{(t\lambda_j(\bm{B}))^n}{n} \\
    &\leq
    \frac14\sum_j\sum_{n=2}^{\infty}
    |t\lambda_j(\bm{B})|^n \\
    &\leq
    \frac{t^2\|\bm{B}\|_F^2}{4(1-r)}.
\end{align*}
Hence, whenever $r<1$,
\[
    \log\left(1+\chi^2(\mathbb{P}_{\gamma}\mid\mathbb{Q})\right)
    \leq
    \frac{\beta^8\|\bm{B}\|_F^2}{4(1-\beta^8)}
    +
    \frac{t^2\|\bm{B}\|_F^2}{4(1-r)}.
\]

Suppose now that $\eta=\beta^4\|\bm{B}\|_F=o(1)$.
If $\bm{B}=0$, the two distributions coincide, so it suffices to
consider the indices for which $\bm{B}\neq0$. Along these indices,
$\|\bm{B}\|_F\geq1$, and hence $\beta=o(1)$. Therefore,
\[
    t=\beta^4(1+o(1)),
    \qquad
    t\|\bm{B}\|_F=(1+o(1))\eta=o(1).
\]
In particular,
\[
    0\leq r=t\lambda_1(\bm{B})\leq t\|\bm{B}\|_F=o(1),
\]
so the preceding bound applies for all sufficiently large $N$.
Since
\[
    \beta^8\|\bm{B}\|_F^2=\eta^2=o(1),
    \qquad
    t^2\|\bm{B}\|_F^2=(1+o(1))\eta^2=o(1),
\]
we conclude that
\[
    0\leq
    \log\left(1+\chi^2(\mathbb{P}_{\gamma}\mid\mathbb{Q})\right)
    =o(1).
\]
Thus $\chi^2(\mathbb{P}_{\gamma}\mid\mathbb{Q})=o(1)$, so no sequence of tests achieves weak detection, proving Part 1 of the result.

\paragraph{Part 2: Constant signal regime}
Suppose $\eta=\Theta(1)$ and $\zeta<1$.
Since $r=t\lambda_1(\bm{B})\to\zeta$, there exists a constant $\delta\in(0,1)$ such that
\[
    r\colonequals t\lambda_1(\bm{B})\leq1-\delta
\]
for all sufficiently large $N$.

Since $\eta=\Theta(1)$, we have $\bm{B}\neq0$ for all sufficiently
large $N$. The entries of $\bm{B}$ are nonnegative integers, so
$\lambda_1(\bm{B})\geq1$. Consequently,
\[
    t\leq1-\delta,
    \qquad
    \beta^4=\tanh(t)\leq\tanh(1-\delta)<1.
\]
Thus $\beta$ is bounded away from $1$.

Using the expansion of $\operatorname{arctanh}$, we obtain
\[
    t
    =\sum_{n=0}^{\infty}\frac{\beta^{8n+4}}{2n+1}
    \leq\frac{\beta^4}{1-\beta^8}.
\]
The bound established above therefore gives
\begin{align*}
    \log\left(1+\chi^2(\mathbb{P}_{\gamma}\mid\mathbb{Q})\right)
    &\leq
    \frac{\eta^2}{4(1-\beta^8)}
    +\frac{t^2\|\bm{B}\|_F^2}{4(1-r)} \\
    &\leq
    \frac{\eta^2}{4(1-\beta^8)}
    +\frac{\eta^2}{4(1-\beta^8)^2(1-r)} \\
    &=O(1),
\end{align*}
where we used $\eta=\Theta(1)$, the bound on $\beta$, and
$1-r\geq\delta$. Hence,
\[
    \chi^2(\mathbb{P}_{\gamma}\mid\mathbb{Q})=O(1).
\]
Consequently, no sequence of functions achieves strong separation,
and no sequence of tests achieves strong detection.
\end{proof}

The following gives roughly matching lower bounds on low-degree and full chi-squared divergences, showing that, at least under some natural further assumptions, our analysis above is tight.

\begin{Theorem}
    \label{thm:ub-sync}
    In the setting of Theorem~\ref{thm:lb-sync}, the following statements also hold as $N\to\infty$.
\begin{enumerate}
    \item If $\eta=\omega(1)$, then $\chi^2_{\leq4}(\mathbb{P}_{\gamma}\mid\mathbb{Q})=\omega(1)$.\footnote{This provides evidence, but does not prove, that degree-$4$ polynomials may achieve
    strong separation between $\mathbb{P}_{\gamma}$ and $\mathbb{Q}$.}
    \item There is an absolute constant $c_0 > 0$ such that the following holds.
    Suppose that $\eta = \Theta(1)$, the limit defining $\zeta$ exists, and $\zeta > 1$.
    Choose constants $\phi, L > 0$ such that $\zeta > 1 + \phi$ and $\eta \leq L$ for all $N$ sufficiently large.
    Suppose that there exists $\epsilon \in (0, \min\{1, c_0 \frac{\phi}{2L}\})$ such that
    \[ \max_{\{j \in [m]: |\lambda_j(\bm{B})| > \epsilon \lambda_1(\bm{B})\}} \|\bm{v}_j(\bm{B})\|_{\infty} = o(1), \]
    where $\bm{v}_j(\bm{B})$ is the unit eigenvector of $\bm{B}$ associated to eigenvalue $\lambda_j(\bm{B})$.
    Suppose that, for all sufficiently large $N$, the number of $j \in [m]$ such that $\lambda_j(\bm{B}) > \epsilon \lambda_1(\bm{B})$ is some constant $k$ and that, for each such $j$, $(\frac{1}{1 + \gamma^2})^2 \lambda_j(\bm{B})$ converges as $N \to \infty$.
    Then, for any sequence $D = D(N) = \omega(1)$, we have $\chi^2_{\leq D}(\mathbb{P}_{\gamma} \mid \mathbb{Q}) = \omega(1)$.\footnote{As above, this provides evidence, but does not prove, that a polynomial-time test statistic may achieve strong separation between $\mathbb{P}_{\gamma}$ and $\mathbb{Q}$, such as a spectral algorithm that, while efficiently computable, requires diverging polynomial degree to express approximately as a low-degree polynomial.}
\end{enumerate}
\end{Theorem}

\begin{proof}[Proof of Theorem~\ref{thm:ub-sync}]
We retain the notations $\beta,t,E$ and $b_e$ from the proof of Theorem~\ref{thm:lb-sync}.
We use the expansion \eqref{eq:sync-hermite-expansion}.

\paragraph{Part 1: Diverging signal regime}
The parity constraint requires every vertex of the multigraph
with edge multiplicities $(k_e)_{e\in E}$ to have even degree.
At polynomial degree at most four, the only nonconstant
contributions therefore come from a single doubled edge:
$k_e=2$ for one edge and $k_{e'}=0$ for all other edges.
It follows that
\[
    \chi^2_{\leq 4}(\mathbb{P}_{\gamma}\mid\mathbb{Q})
    =
    \beta^8\sum_{i<j}\binom{B_{ij}+1}{2}
    =
    \frac{\beta^8}{2}\sum_{i<j}B_{ij}(B_{ij}+1).
\]
Since each $B_{ij}$ is a nonnegative integer,
\[
    B_{ij}^2 \leq B_{ij}(B_{ij}+1) \leq 2B_{ij}^2.
\]
Using $\|\bm{B}\|_F^2=2\sum_{i<j}B_{ij}^2$, we conclude that
\[
    \frac{\beta^8}{4}\|\bm{B}\|_F^2
    \leq
    \chi^2_{\leq 4}(\mathbb{P}_{\gamma}\mid\mathbb{Q})
    \leq
    \frac{\beta^8}{2}\|\bm{B}\|_F^2.
\]
Thus, since $\eta=\beta^4\|\bm{B}\|_F$,
\[
    \frac{\eta^2}{4}
    \leq
    \chi^2_{\leq 4}(\mathbb{P}_{\gamma}\mid\mathbb{Q})
    \leq
    \frac{\eta^2}{2}.
\]
Hence,
$\chi^2_{\leq 4}(\mathbb{P}_{\gamma}\mid\mathbb{Q})=\omega(1)$
if and only if $\eta=\omega(1)$, proving Part 1.

\paragraph{Part 2: Constant signal regime}
We next prove that the low-degree chi-squared divergence diverges
in the uniformly supercritical regime
\[
    \zeta=\lim_{N\to\infty}t\lambda_1(\bm{B})>1,
    \qquad \eta=\Theta(1).
\]
Let $\phi,L$ and $\epsilon$ be as in the statement, so that
\[
    0<\epsilon<
    \min\left\{1,\frac{c_0\phi}{2L}\right\}
\]
and
\[
    \max_{\{j:\,|\lambda_j(\bm{B})|
        >\epsilon\lambda_1(\bm{B})\}}
    \|\bm{v}_j(\bm{B})\|_\infty=o(1),
\]
where the eigenvectors have unit Euclidean norm.
Here $c_0$ is the constant in
Lemma~\ref{lem:symmetric-bound}, and $\phi,L>0$ are the
constants in the statement.
Recall that
\[
    1+\chi^2_{\leq D}(\mathbb{P}_{\gamma}\mid\mathbb{Q})
    =
    \sum_{\substack{
        \bm{k}\in\mathbb{N}_0^E\\
        2\sum_{e\in E}k_e\leq D
    }}
    \bm{1}_{\left\{
        \sum_{e:\,i\in e}k_e
        \text{ is even for every }i
    \right\}}
    \prod_{e\in E}
    \binom{b_e+k_e-1}{k_e}\beta^{4k_e}.
\]

We first show that the flatness assumption, together with
$\eta=\Theta(1)$, implies $\beta=o(1)$. Let
\[
    S\colonequals \bigcup_{e\in E}e
\]
be the set of vertices incident to at least one edge.
Since the nonzero entries of $\bm{B}$ are positive integers,
\[
    |S|\leq2|E|
    \leq2\sum_{i<j}B_{ij}^2
    =\|\bm{B}\|_F^2.
\]
Let $\bm{v}=\bm{v}_1(\bm{B})$ be a unit eigenvector corresponding to
$\lambda_1(\bm{B})>0$. For every $i\notin S$, the $i$th row of
$\bm{B}$ is zero, and hence
\[
    \lambda_1(\bm{B})v_i=\sum_jB_{ij}v_j=0.
\]
Thus $\bm{v}$ vanishes outside $S$, so
\[
    1=\|\bm{v}\|_2^2
    =\sum_{i\in S}v_i^2
    \leq |S|\|\bm{v}\|_\infty^2
    \leq\|\bm{B}\|_F^2\|\bm{v}\|_\infty^2.
\]
Because $\epsilon<1$, the flatness assumption applies to
$\bm{v}$ and gives $\|\bm{v}\|_\infty=o(1)$. Therefore,
$\|\bm{B}\|_F=\omega(1)$. Using
$\eta=\beta^4\|\bm{B}\|_F=\Theta(1)$, we obtain
\[
    \beta^4=\frac{\eta}{\|\bm{B}\|_F}=o(1).
\]
Consequently, $\beta=o(1)$ and $t=\beta^4(1+o(1))$.

Define
\[
    \bm{K}_N\colonequals \beta^4\bm{B}.
\]
Since $\lambda_1(\bm{K}_N)\to\zeta>1+\phi$, we have
\[
    \lambda_1(\bm{K}_N)\geq1+\phi
\]
for all sufficiently large $N$. Moreover,
\[
    \|\bm{K}_N\|_F=\beta^4\|\bm{B}\|_F=\eta=\Theta(1),
\]
and the choice of $L$ in the statement gives $\|\bm{K}_N\|_F\leq L$
for all sufficiently large $N$.

Let $\bm{\tau}=(\tau_1,\ldots,\tau_m)$ have
independent Rademacher entries, and define
\[
    Q_N\colonequals \frac12\bm{\tau}^\top \bm{K}_N\bm{\tau}
    =\beta^4\sum_{e=\{i,j\}\in E}b_e\tau_i\tau_j.
\]
For every integer $b\geq1$ and $k\geq0$,
\[
    \binom{b+k-1}{k}
    =\frac{b(b+1)\cdots(b+k-1)}{k!}
    \geq\frac{b^k}{k!},
\]
with the empty product interpreted as $1$. Therefore,
\begin{align*}
    1+\chi^2_{\leq D}(\mathbb{P}_{\gamma}\mid\mathbb{Q})
    &\geq
    \sum_{\substack{
        \bm{k}\in\mathbb{N}_0^E\\
        2\sum_{e\in E}k_e\leq D
    }}
    \bm{1}_{\left\{
        \sum_{e:\,i\in e}k_e
        \text{ is even for every }i
    \right\}}
    \prod_{e\in E}\frac{(b_e\beta^4)^{k_e}}{k_e!} \\
    &=
    \sum_{q=0}^{\lfloor D/2\rfloor}
    \sum_{\substack{
        \bm{k}\in\mathbb{N}_0^E\\
        \sum_{e\in E}k_e=q
    }}
    \E_{\bm{\tau}}\left[
        \prod_{e=\{i,j\}\in E}
        \frac{(b_e\beta^4\tau_i\tau_j)^{k_e}}{k_e!}
    \right] \\
    &=
    \sum_{q=0}^{\lfloor D/2\rfloor}
    \frac{\E[Q_N^q]}{q!} \\
    &\geq
    \sum_{\ell=0}^{\lfloor D/4\rfloor}
    \frac{\E[Q_N^{2\ell}]}{(2\ell)!}.
\end{align*}
The penultimate equality follows from the multinomial theorem.
For the final inequality, every moment $\E[Q_N^q]$ is
nonnegative: each term in its multinomial expansion has a
nonnegative coefficient, and the expectation of the corresponding
Rademacher monomial is either $0$ or $1$. We may therefore
discard the odd powers.

Write
\[
    \bm{K}_N=\bm{A}_N+\bm{R}_N,
\]
where $\bm{A}_N$ retains the eigenvalues satisfying
$|\lambda_j(\bm{K}_N)|>\epsilon\lambda_1(\bm{K}_N)$. Then
\[
    \operatorname{rank}(\bm{A}_N)
    \leq
    \frac{\|\bm{K}_N\|_F^2}
         {\epsilon^2\lambda_1(\bm{K}_N)^2}
    \leq
    \frac{L^2}{\epsilon^2(1+\phi)^2}.
\]
Furthermore,
\[
    \|\bm{R}_N\|
    \leq\epsilon\lambda_1(\bm{K}_N)
    \leq\epsilon L,
    \qquad
    \|\bm{R}_N\|_F\leq L.
\]
Define
\[
    X_N\colonequals \frac12\left(
        \bm{\tau}^\top \bm{A}_N\bm{\tau}
        -\operatorname{tr}(\bm{A}_N)
    \right),
    \qquad
    W_N\colonequals \frac12\left(
        \bm{\tau}^\top \bm{R}_N\bm{\tau}
        -\operatorname{tr}(\bm{R}_N)
    \right).
\]
Since $\operatorname{tr}(\bm{K}_N)=0$, we have
\[
    Q_N=X_N+W_N.
\]

Assume that, for all sufficiently large $N$, $\bm{A}_N$ has
$k_*$ retained eigenvalues
$\lambda_{1,N},\ldots,\lambda_{k_*,N}$, with corresponding
orthonormal eigenvectors $\bm{v}_{1,N},\ldots,\bm{v}_{k_*,N}$, and that
\[
    \lambda_{a,N}\longrightarrow\lambda_a,
    \qquad a=1,\ldots,k_*.
\]
Label them so that $\lambda_{1,N}=\lambda_1(\bm{K}_N)$.
Then $\lambda_1\geq1+\phi$. Since $\bm{K}_N$ is a positive scalar
multiple of $\bm{B}$, the flatness assumption also applies to
these eigenvectors. By the same argument as in the proof
of Theorem~1.13 of~\cite{de2026computational},
\[
    (\bm{v}_{1,N}^\top\bm{\tau},\ldots,
     \bm{v}_{k_*,N}^\top\bm{\tau})
    \xrightarrow{\mathrm d}(G_1,\ldots,G_{k_*}),
\]
where $G_1,\ldots,G_{k_*}$ are independent standard Gaussian
random variables. Consequently,
\[
    X_N\xrightarrow{\mathrm d}
    \Lambda\colonequals \frac12\sum_{a=1}^{k_*}\lambda_a(G_a^2-1).
\]

Choose $p\colonequals 1+\phi/2$. Our choice of $\epsilon$ ensures that
\[
    \frac{\|\bm{R}_N\|}{p-1}
    \leq\frac{2\epsilon L}{\phi}<c_0.
\]
Applying Lemma~\ref{lem:symmetric-bound}, we obtain
\[
    \E\left[\exp\left(-\frac{W_N}{p-1}\right)\right]
    \leq
    C_*\colonequals \exp\left(\frac{C_0L^2}{(p-1)^2}\right).
\]
Finally, since $\bm{K}_N$ has zero diagonal,
\[
    \E[Q_N^2]
    =\sum_{i<j}(\bm{K}_N)_{ij}^2
    =\frac12\|\bm{K}_N\|_F^2
    \leq\frac{L^2}{2}.
\]

Fix any sequence $D_N\to\infty$, and set
\[
    z_N\colonequals 1+\chi^2_{\leq D_N}(\mathbb{P}_{\gamma}\mid\mathbb{Q}).
\]
Suppose, toward a contradiction, that there exist a subsequence
$(N_i)$ and a constant $\mathcal A<\infty$ such that
\[
    z_{N_i}\leq\mathcal A
    \qquad\text{for every }i.
\]
For each integer $h\geq0$, define
\[
    F_h(q)\colonequals \sum_{\ell=0}^h\frac{q^{2\ell}}{(2\ell)!}.
\]
Since $D_N\to\infty$, we have $D_{N_i}\geq4h$ for all
sufficiently large $i$, and hence
\[
    \E[F_h(Q_{N_i})]
    \leq z_{N_i}
    \leq\mathcal A.
\]
For fixed $T>0$, let
\[
    r_h(T)\colonequals \sup_{|q|\leq T}\bigl(\cosh q-F_h(q)\bigr).
\]
Uniform convergence of the power series on $[-T,T]$ implies
$r_h(T)\to0$ as $h\to\infty$. Since $\exp(q)\leq2\cosh q$
and $F_h(q)\geq0$,
\[
    \E\left[
        \exp(Q_{N_i})\bm{1}_{\{|Q_{N_i}|\leq T\}}
    \right]
    \leq
    2\E[F_h(Q_{N_i})]+2r_h(T).
\]
Taking the limit superior as $i\to\infty$ and since $h$ is arbitrary, we obtain
\[
    \limsup_{i\to\infty}
    \E\left[
        \exp(Q_{N_i})\bm{1}_{\{|Q_{N_i}|\leq T\}}
    \right]
    \leq2\mathcal A.
\]

For $M>0$, define the bounded continuous function
\[
    f_M(x)\colonequals \min\{\exp(x/p),M\}.
\]
Using $X_{N_i}=Q_{N_i}-W_{N_i}$, H\"older's inequality,
and the second-moment bound for $Q_N$, we obtain
\begin{align*}
    \E[f_M(X_{N_i})]
    &\leq
    \E\left[
        \exp(X_{N_i}/p)\bm{1}_{\{|Q_{N_i}|\leq T\}}
    \right]
    +M\mathbb{P}\bigl[|Q_{N_i}|>T\bigr] \\
    &\leq
    \left(
        \E\left[
            \exp(Q_{N_i})\bm{1}_{\{|Q_{N_i}|\leq T\}}
        \right]
    \right)^{1/p}
    \left(
        \E[\exp(-W_{N_i}/(p-1))]
    \right)^{(p-1)/p}
    +\frac{ML^2}{2T^2}.
\end{align*}
Since $X_{N_i}\xrightarrow{\mathrm d}\Lambda$ and $f_M$
is bounded and continuous,
\[
    \E[f_M(X_{N_i})]\longrightarrow\E[f_M(\Lambda)].
\]
Taking limits therefore gives
\[
    \E[f_M(\Lambda)]
    \leq
    (2\mathcal A)^{1/p}C_*^{(p-1)/p}
    +\frac{ML^2}{2T^2}.
\]
Letting $T\to\infty$ for fixed $M$ as $T$ is arbitrary here, and then applying
monotone convergence as $M\to\infty$, yields
\[
    \E[\exp(\Lambda/p)]
    \leq(2\mathcal A)^{1/p}C_*^{(p-1)/p}<\infty.
\]

However, $\lambda_1\geq1+\phi>p$, so
\[
    \E\left[\exp\left(\frac{\lambda_1G_1^2}{2p}\right)\right]=\infty.
\]
Conditioning on $G_2,\ldots,G_{k_*}$ therefore gives
$\E[\exp(\Lambda/p)]=\infty$, a contradiction. Consequently, for every sequence $D_N \to \infty$ we have
\[ \chi^2_{\leq D_N}(\mathbb{P}_{\gamma}\mid\mathbb{Q}) =\omega(1), \]
completing the proof.
\end{proof}

\subsection{Proof of Theorem~\ref{thm:lb-mx}}

We first recall Corollary~3.27 of~\cite{de2026computational} concerning hypothesis testing in general additive Gaussian models.
Let $\mathcal{X}=\mathcal{X}_M$ be a probability measure on $\RR^M$, and let $\bm X\sim\mathcal{X}$ and $\bm Z\sim\mathcal{N}(\bm 0,\bm\Sigma)$ be independent, where $\bm\Sigma\in\RR^{M\times M}_{\mathrm{sym}}$ is positive definite.
Under the alternative $\mathbb{P}$, we observe $\bm Y=\bm X+\bm Z$, and under the null $\mathbb{Q}$, we observe $\bm Y=\bm Z$.
For independent copies $\bm X^1,\bm X^2\sim\mathcal{X}$, the full and low-degree chi-squared divergences satisfy
\begin{align*}
    \chi^2(\mathbb{P}\mid\mathbb{Q})
    &=\E_{\bm X^1,\bm X^2\sim\mathcal{X}}\left[
        \exp\left((\bm X^1)^\top\bm\Sigma^{-1}\bm X^2\right)
    \right]-1, \\
    \chi^2_{\leq D}(\mathbb{P}\mid\mathbb{Q})
    &=\E_{\bm X^1,\bm X^2\sim\mathcal{X}}\left[
        \exp^{\leq D}\left((\bm X^1)^\top\bm\Sigma^{-1}\bm X^2\right)
    \right]-1,
\end{align*}
where $\exp^{\leq D}(u)\colonequals\sum_{k=0}^D u^k/k!$.

\begin{lemma}
\label{lem:matrix-chi-squared}
For the model in Theorem~\ref{thm:lb-mx}, let $\operatorname{vec}_{\triangle}$ list the entries of a symmetric matrix on and above the diagonal, and define the stacked signal $\bm X\in\RR^M$, where $M=\sum_{i=1}^m |\mathcal{U}_i|(|\mathcal{U}_i|+1)/2$, by
\[
    \bm X=\frac{\tilde\lambda}{\sqrt{N}}
    \begin{pmatrix}
        \operatorname{vec}_{\triangle}\left(
            \bm{x}_{\mathcal{U}_1}^{(0)}\bm{x}_{\mathcal{U}_1}^{(0)\top}
        \right) \\
        \vdots \\
        \operatorname{vec}_{\triangle}\left(
            \bm{x}_{\mathcal{U}_m}^{(0)}\bm{x}_{\mathcal{U}_m}^{(0)\top}
        \right)
    \end{pmatrix}.
\]
For independent copies $\bm u,\bm v\sim\nu^{\otimes N}$ of $\bm x^{(0)}$, the divergences satisfy
\begin{align*}
    \chi^2(\mathbb{P}_{\tilde\lambda}\mid\mathbb{P}_0)
    &=\E_{\bm u,\bm v}\left[
        \exp\left(
            \frac{\tilde\lambda^2}{2N}
            (\bm u\circ\bm v)^\top
            (\bm\Phi+\diag(\bm\Phi))
            (\bm u\circ\bm v)
        \right)
    \right]-1, \\
    \chi^2_{\leq D}(\mathbb{P}_{\tilde\lambda}\mid\mathbb{P}_0)
    &=\E_{\bm u,\bm v}\left[
        \exp^{\leq D}\left(
            \frac{\tilde\lambda^2}{2N}
            (\bm u\circ\bm v)^\top
            (\bm\Phi+\diag(\bm\Phi))
            (\bm u\circ\bm v)
        \right)
    \right]-1,
\end{align*}
where $\circ$ denotes the coordinatewise product,
\[
    \bm\Phi\colonequals\sum_{i=1}^m
    \bm{1}_{\mathcal{U}_i}\bm{1}_{\mathcal{U}_i}^\top,
\]
$\bm{1}_{\mathcal{U}_i}\in\RR^N$ is the indicator vector of $\mathcal{U}_i$, and $\diag(\bm\Phi)$ retains the diagonal entries of $\bm\Phi$ and sets its off-diagonal entries to zero.
\end{lemma}

\begin{proof}
Stacking the upper triangular entries of the observations gives an additive Gaussian model with noise covariance $\bm I_M$.
For two copies of the stacked signal formed from $\bm u$ and $\bm v$, respectively,
\begin{align*}
    (\bm X^1)^\top\bm X^2
    =\frac{\tilde\lambda^2}{N}
      \sum_{i=1}^m\sum_{\substack{a,b\in\mathcal{U}_i\\a\leq b}}
      u_av_a u_bv_b
    =\frac{\tilde\lambda^2}{2N}
      (\bm u\circ\bm v)^\top(\bm\Phi+\diag(\bm\Phi))(\bm u\circ\bm v).
\end{align*}
The preceding Gaussian identities give the result.
\end{proof}

\begin{proof}[Proof of Theorem~\ref{thm:lb-mx}]
By assumption,
\[
    \max_{a\in[N]}\#\{i:a\in\mathcal{U}_i\}=O(1).
\]
Consequently, the entries of $\bm\Phi$ are uniformly bounded: there exists a constant $C>0$ such that, for all $a,b\in[N]$,
\[
    0\leq\Phi_{ab}
    =\sum_{i=1}^m\bm{1}_{\{a\in\mathcal{U}_i\}}\bm{1}_{\{b\in\mathcal{U}_i\}}
    \leq\#\{i:a\in\mathcal{U}_i\}\leq C.
\]
The divergence formulas in Lemma~\ref{lem:matrix-chi-squared} therefore coincide with those in Lemma~3.28 of~\cite{de2026computational} for the variance profile $\bm\Delta$ given by
\[
    \Delta_{ab}\colonequals
    \begin{cases}
        1/\Phi_{ab},&\Phi_{ab}>0,\\
        +\infty,&\Phi_{ab}=0.
    \end{cases}
\]
Here infinite variance means that an entry is unobserved, as in that reference, and every finite variance is at least $1/C$.
The divergence bounds in the proof of Theorem~1.11 of~\cite{de2026computational}, under the respective assumptions on $\nu$ in Parts 1 and 2 of this result, give
\[ \chi^2_{\leq D}(\mathbb{P}_{\tilde\lambda}\mid\mathbb{P}_0) = O(1) \qquad \text{if }D=o\left(\frac{\lambda_1(\bm\Phi)^3}{N^2}\right), \]
and
\[ \chi^2(\mathbb{P}_{\tilde\lambda}\mid\mathbb{P}_0) = O(1), \]
both whenever
\[
    \tilde\lambda<\lim_{N\to\infty}\sqrt{\frac{N}{\lambda_1(\bm\Phi)}},
\]
provided this limit exists.

It remains to show that $\lambda_1(\bm\Phi)=\lambda_1(\widetilde{\bm B})$, where $\widetilde{\bm B}=(\widetilde B_{ij})$ is the overlap matrix in the theorem.
Let $\bm A\in\RR^{N\times m}$ be the incidence matrix with entries $A_{ai}=\bm{1}_{\{a\in\mathcal{U}_i\}}$.
Then
\begin{align*}
    \bm\Phi&=\bm A\bm A^\top,\\
    \widetilde{\bm B}&=\bm A^\top\bm A.
\end{align*}
Since $\bm A\bm A^\top$ and $\bm A^\top\bm A$ have the same nonzero eigenvalues,
\[
    \lambda_1(\bm\Phi)=\lambda_1(\widetilde{\bm B})=\|\bm A\|^2.
\]
Substituting this identity into the threshold and degree condition, and applying the separation and detection implications recalled above, proves both parts of the theorem.
\end{proof}

\subsection{Auxiliary lemmas}

\begin{lemma}\label{lemma:lambdai phase transition}
Consider the setting of Theorem~\ref{thm:xi}.
For any $k\in[d]$ with $\rho_k^*>1$, we have
\begin{align*}
\frac{\hat\lambda_k^{(i)}}{\sigma_i\sqrt{n_i}}
&\xrightarrow{\mathbb P}\rho_k^*+\frac{1}{\rho_k^*},\\
\langle\hat{\bm u}_k^{(i)},\bm u_k^{(i)}\rangle^2
&\xrightarrow{\mathbb P}1-\frac{1}{(\rho_k^*)^2},\\
\langle\hat{\bm u}_k^{(i)},\bm u_\ell^{(i)}\rangle
&\xrightarrow{\mathbb P}0 \text{ for }\ell\ne k.
\end{align*}
\end{lemma}
\begin{proof}
Apply Theorems~2.1 and~2.2 and Section~3.1 of \cite{benaych2011eigenvalues} to $\bm A^{(i)}/(\sigma_i\sqrt{n_i})$.
The normalized population eigenvalues converge to the distinct limits $\rho_k^*$.
\end{proof}

\begin{lemma}\label{lemma:ui phase transition}
	Consider the setting of Theorem~\ref{thm:xi}. For any $k\in[d]$ with $\rho_k^*>1$, let
	$
\zeta_k^{(i)}
   \colonequals
    1 - \frac{\sigma_i^2 n_i}{(\lambda_k^{(i)})^2}.
$
Then for any finite ordered set of indices $\mathcal{I}\subset[n]$, we have
$$
\frac{\sqrt{n_i}\left[\hat{\bm{u}}_k^{(i)}-\sqrt{\zeta_k^{(i)}}\bm{u}_k^{(i)}\right]_{\mathcal{I}}}
{\sqrt{1-\zeta_k^{(i)}}}
\xrightarrow{\mathcal{D}}
  \mathcal{N}(\bm{0},\mi_{|\mathcal{I}|})
$$
with $\hat{\bm{u}}_k^{(i)}$ such that $\hat{\bm{u}}_k^{(i)\top}{\bm{u}}_k^{(i)}\geq 0$ (otherwise, consider $-\hat{\bm{u}}_k^{(i)}$).
\end{lemma}
\begin{proof}
This proof is adapted from the proof of Theorem 2 in \cite{lebeau2024asymptotic}. For ease of exposition, we suppress the superscript $(i)$ throughout this proof.

	Consider the tangent-normal decomposition
	$$
	\hat{\bm{u}}_k
	=\sum_{\ell=1}^d \tau_\ell \bm{u}_\ell
	+\sqrt{1-\|\bm{\tau}\|^2}\hat{\bm{u}}_k^\#,
	$$
	where $\hat{\bm{u}}_k^\#=(\bm{I}_n-\muu\muu^\top)\frac{\hat{\bm{u}}_k}{\sqrt{1-\|\bm\tau\|^2}}$ is a unit-norm vector orthogonal to the span of $\muu$, and $\bm\tau=[\tau_1,\dots,\tau_d]^\top$ with $\tau_\ell=\bm{u}_\ell^\top\hat{\bm{u}}_k$ measures the cosine of the angle between $\bm{u}_\ell$ and $\hat{\bm{u}}_k$.
By Lemma~\ref{lemma:lambdai phase transition} and the sign choice for $\hat{\bm u}_k$,
\[
\|\bm\tau-\sqrt{\zeta_k}\bm e_k\|\xrightarrow{\mathbb P}0,
\qquad
\sqrt{1-\|\bm\tau\|^2}-\sqrt{1-\zeta_k}\xrightarrow{\mathbb P}0,
\]
where $\bm e_k$ is the $k$th coordinate vector in $\mathbb R^d$.

	Consider any $n\times n$ orthogonal matrix $\mo$ such that $\mo\muu=\muu$, i.e., a rotational symmetry about the span of $\muu$, and define $\tilde\ma=\mo\ma\mo^\top$.
	Then
	$$
	\tilde\ma
	=(\mo\muu)\mLambda(\mo\muu)^\top
	+\mo\mn\mo^\top.
	$$
	Notice that $\mo\mn\mo^\top$ and $\mn$ are identically distributed.
	Thus, $\tilde\ma$ follows the same model as $\ma$ because $\mo\muu=\muu$.
	Following the same argument as in the proof of Theorem 2 in \cite{lebeau2024asymptotic}, this rotational symmetry implies that $\hat{\bm{u}}_k^\#$ is uniformly distributed on $\mathbb{S}^{n-1}\cap \muu_\perp$, and consequently
	\begin{equation}\label{eq:hatu_k2}
		  \sqrt{n}[\hat{\bm{u}}_k^\#]_\mathcal{I} \xrightarrow{\mathcal{D}}
  \mathcal{N}(\bm{0},\bm{I}_{|\mathcal{I}|})
	\end{equation}
  for any finite ordered set $\mathcal{I}\subset[n]$.
To control the remaining term at the same scale, use the exact decomposition
\begin{equation}\label{eq:hatu_k}
\sqrt n[\hat{\bm u}_k-\sqrt{\zeta_k}\bm u_k]_{\mathcal I}
=\sqrt n\,\muu_{\mathcal I}(\bm\tau-\sqrt{\zeta_k}\bm e_k)
+\sqrt{1-\|\bm\tau\|^2}\sqrt n[\hat{\bm u}_k^\#]_{\mathcal I}.
\end{equation}
Since $d$ and $|\mathcal I|$ are fixed, the incoherence bound gives $\sqrt n\|\muu_{\mathcal I}\|\lesssim1$.
Thus
\[
\|\sqrt n\,\muu_{\mathcal I}(\bm\tau-\sqrt{\zeta_k}\bm e_k)\|
\leq\sqrt n\|\muu_{\mathcal I}\|\|\bm\tau-\sqrt{\zeta_k}\bm e_k\|
\xrightarrow{\mathbb P}0.
\]
Combining this bound with \eqref{eq:hatu_k2} and applying Slutsky's theorem proves the result.
  \end{proof}

\begin{lemma}\label{lemma:lambdai and lambda}
Let $\bm{P} = \bm{U}\boldsymbol{\Lambda}\bm{U}^\top \in \mathbb{R}^{N \times N}$ where $\bm{U} \in \mathbb{R}^{N \times d}$ has orthonormal columns, $d$ is fixed, and $\boldsymbol{\Lambda} = \mathrm{diag}(\lambda_1, \ldots, \lambda_d)$ has distinct positive eigenvalues.
Assume that $\|\bm{U}\|_{2\to\infty}^2 \lesssim d/N$.
Fix $\alpha\in(0,1]$, let $\mathcal U_i$ consist of $n=\alpha N$ independent uniform draws from $[N]$, and define
$\bm P^{(i)}=\bm U_{\mathcal U_i}\boldsymbol\Lambda\bm U_{\mathcal U_i}^\top$.
Let $\lambda_k^{(i)}$ denote its $k$th largest eigenvalue.
Then, for each $k\in[d]$,
\[
\frac{\lambda_k^{(i)}}{\alpha\lambda_k}\xrightarrow{\mathrm{a.s.}}1.
\]
\end{lemma}

\begin{proof}
Set $\bm G=\bm U_{\mathcal U_i}^\top\bm U_{\mathcal U_i}$.
Lemma~\ref{lemma:U^TU to alpha I} gives $\delta_N\colonequals\|\alpha^{-1}\bm G-\bm I_d\|\to0$ almost surely.
Hence
\[
(1-\delta_N)\alpha\boldsymbol\Lambda
\preceq
\boldsymbol\Lambda^{1/2}\bm G\boldsymbol\Lambda^{1/2}
\preceq
(1+\delta_N)\alpha\boldsymbol\Lambda.
\]
The middle matrix has the same nonzero eigenvalues as $\bm P^{(i)}$.
The min-max principle therefore gives
$1-\delta_N\leq\lambda_k^{(i)}/(\alpha\lambda_k)\leq1+\delta_N$, proving the claim.
\end{proof}

\begin{lemma}\label{lemma:U^TU to alpha I}
Let $\bm{U} \in \mathbb{R}^{N \times d}$ have orthonormal columns, i.e., $\bm{U}^\top \bm{U} = \bm{I}_d$, with $d$ fixed.
Assume that $\|\bm{U}\|_{2\to\infty}^2 %
\lesssim\frac{d}{N}$. %
Suppose $\mathcal{U}$ consists of $n = \alpha N$  {\color{black}i.i.d. uniform} indices from $[N]$.
Define the subsampled matrix $\bm{U}_{\mathcal{U}} \in \mathbb{R}^{n \times d}$ by selecting the rows of $\bm{U}$ corresponding to indices in $\mathcal{U}$.
Then, as $N\to\infty$, we have
\[
\bm{U}_{\mathcal{U}}^\top \bm{U}_{\mathcal{U}} \xrightarrow{\text{a.s.}} \alpha \bm{I}_d.
\]
\end{lemma}

\begin{proof}
Let $\bm{u}_{j_1}, \dots, \bm{u}_{j_n} \in \mathbb{R}^d$ denote the rows of $\bm{U}$ sampled to form $\bm{U}_{\mathcal{U}}$.
Define the centered random matrices
$$
\bm{Z}_i  \colonequals  \bm{u}_{j_i} \bm{u}_{j_i}^\top - \frac{1}{N} \bm{I}_d, \quad i \in [n].
$$
Then $\{\bm{Z}_i\}_{i\in[n]}$ are independent, zero-mean, symmetric random matrices, and
$$
\sum_{i=1}^n \bm{Z}_i = \bm{U}_{\mathcal{U}}^\top \bm{U}_{\mathcal{U}} - \alpha \bm{I}_d.
$$
For each $i\in[n]$, by the uniform sampling assumption, we have
$$
\mathbb{E}[\bm{u}_{j_i} \bm{u}_{j_i}^\top] = \frac{1}{N} \sum_{j=1}^N \bm{u}_j \bm{u}_j^\top = \frac{1}{N} \bm{U}^\top \bm{U} = \frac{1}{N} \bm{I}_d,
$$
so $\mathbb{E}[\bm{Z}_i] = \bm{0}$.
Since $\bm{Z}_i$ is symmetric, we have $\bm{Z}_i^2 \preceq \|\bm{Z}_i\|^2 \bm{I}_d$.
And for $\|\bm{Z}_i\|$, we have
$$
\|\bm{Z}_i\|
\le \|\bm{u}_{j_i}\|^2 + \left\| \frac{1}{N} \bm{I}_d \right\|
\le C\frac{d}{N} + \frac{1}{N}
\le (C+1)\frac{d}{N}
$$
for some constant $C>0$. Therefore,
$$
\bm{Z}_i^2 \preceq \left( (C+1) \frac{d}{N} \right)^2 \bm{I}_d =: \bm{B}_i^2
$$
with $\left\| \sum_{i=1}^n \bm{B}_i^2 \right\| = n\left( (C+1) \frac{d}{N} \right)^2 = \alpha N \cdot \frac{(C+1)^2 d^2}{N^2} = \frac{\alpha (C+1)^2 d^2}{N}$.

Applying the matrix Hoeffding inequality (Theorem~1.3 in \cite{tropp2012user}), we obtain
$$
\mathbb{P}\left( \left\| \sum_{i=1}^n \bm{Z}_i \right\| \ge t \right)
\le d \cdot \exp\left( -\frac{t^2 N}{8 \alpha (C+1)^2 d^2} \right).
$$
For any fixed $\varepsilon > 0$, define the event
$
A_N  \colonequals  \left\{ \left\| \sum_{i=1}^n \bm{Z}_i \right\| \ge \varepsilon \right\}.
$
Setting $t = \varepsilon$, for sufficiently large $N$, we have
$$
\mathbb{P}(A_N) \le d \cdot \exp\left( -c N \right)
$$
for the constant $c  \colonequals  \frac{\varepsilon^2}{8 \alpha (C+1)^2 d^2} > 0$.
Since $\sum_{N=1}^\infty \mathbb{P}(A_N) < \infty$, the Borel-Cantelli lemma implies that %
$\mathbb{P}(A_N \text{ i.o.})=0$, i.e. $\left\|\sum_{i=1}^n \bm{Z}_i \right\| <\varepsilon $ eventually a.s. %
Since $\varepsilon > 0$ was arbitrary, we conclude that
$
\left\| \bm{U}_{\mathcal{U}}^\top \bm{U}_{\mathcal{U}} - \alpha \bm{I}_d \right\| \xrightarrow{\text{a.s.}} 0.
$
\end{proof}

\end{sloppypar}

\end{document}